\documentclass[a4paper,USenglish]{lipics-v2019}
\pdfoutput=1

\usepackage{microtype}\usepackage{multicol}
\usepackage{mathpartir, proof}
\usepackage{amssymb, latexsym}
\usepackage{float}
\usepackage{stmaryrd}
\usepackage{array}
\usepackage{xspace}
\usepackage{thmtools, amsthm, thm-restate}
\usepackage{epigraph}
\usepackage{xcolor}
\usepackage{todonotes}
\usepackage{etoolbox}

\newtoggle{techreport}
\toggletrue{techreport}
\newtoggle{fullproofs}
\togglefalse{fullproofs}

\theoremstyle{definition}

\definecolor{cornflower}{RGB}{65,105,225}
\definecolor{maroon}{RGB}{178,34,34}

\usepackage{listings}
\lstdefinestyle{p4}{
  language=C,
  backgroundcolor=\color{white},
  basicstyle=\linespread{0.8}\footnotesize\ttfamily,
  breaklines=true,
  commentstyle=\color{maroon},
  frame=none,
  tabsize=2,
  xleftmargin=11pt,
  morekeywords={
    header,
    header_type,
    fields,
    extract,
    modify_field,
    metadata,
    control,
    drop,
    parser,
    action,
    table,
    reads,
    actions,
    default_action,
    size,
    if,
    else,
    apply,
    ternary,
    lpm,
    exact,
    @pragma,
    valid
  },
  keywordstyle=\color{cornflower},  
  morekeywords=[2]{
    assume,
    assert,
    drop,
    implies,
    iff,
    and,
    or
  },
  keywordstyle=[2]{\color{violet}},  
}
\definecolor{warning}{rgb}{1.0, 0.8, 0.0}
\lstdefinestyle{fixup}{
  backgroundcolor=\color{white},
  basicstyle=\linespread{0.8}\footnotesize\ttfamily,
  breaklines=true,
  frame=none,
  tabsize=2,
  xleftmargin=11pt,
  moredelim=**[is][\color{gray}]{@}{@},
  moredelim=**[is][\color{cornflower}]{|}{|},
  moredelim=**[is][\color{maroon}]{+}{+},
  moredelim=**[is][\color{warning}]{!}{!},
  moredelim=**[is][\color{white}]{~}{~},
}

\newcommand{\name}{\textsc{SafeP4}\xspace}
\newcommand{\PPPP}[1]{\ensuremath{\textrm{P4}_{#1}}\xspace} \newcommand{\ourchecker}{\textsc{P4Check}\xspace}

\newcommand{\cmdtype}[4]{\ensuremath{#1\vdash #3\!:\!#2 \Mapsto #4}}

\newcommand{\syntax}[1]{\ensuremath{\mathit{#1}}}

\newcommand{\lbbar}{\{\kern-0.5ex|}
\newcommand{\rbbar}{|\kern-0.5ex\}}

\newcommand{\overbar}[1]{\mkern 1.5mu\overline{\mkern-1.5mu#1\mkern-1.5mu}\mkern 1.5mu}

\newcommand{\sos}[1]{\ensuremath{\{\{#1\}\}}} \newcommand{\teval}[1]{\ensuremath{\llbracket #1 \rrbracket}}
\newcommand{\T}{\ensuremath{\Theta}}

\newcommand{\boxedpage}[2]{\fbox{\begin{minipage}{#1}#2\end{minipage}}} 

\title{How to Avoid Making a Billion-Dollar Mistake:
Type-Safe Data Plane Programming with \name}
\titlerunning{Type-Safe Data Plane Programming with \name}

\author{Matthias Eichholz}{Technische Universität Darmstadt, Germany}{}{}{}\author{Eric Campbell}{Cornell University, USA}{}{}{}
\author{Nate Foster}{Cornell University, USA}{}{}{}
\author{Guido Salvaneschi}{Technische Universität Darmstadt, Germany}{}{}{}
\author{Mira Mezini}{Technische Universität Darmstadt, Germany}{}{}{}

\authorrunning{M. Eichholz, E. Campbell, N. Foster, G. Salvaneschi, and M. Mezini}
\Copyright{Matthias Eichholz, Eric Campbell, Nate Foster, Guido Salvaneschi, and Mira Mezini}


\category{}
\relatedversion{}
\supplement{}
\funding{
This work has been co-funded by the German Research Foundation (DFG) as part of the Collaborative Research Center (CRC) 1053 MAKI and 1119 CROSSING,
by the DFG projects SA 2918/2-1 and SA 2918/3-1,
by the Hessian LOEWE initiative within the Software-Factory 4.0 project,
by the German Federal Ministry of Education and Research and 
by the Hessian Ministry of Science and the Arts within CRISP, 
by the National Science Foundation under grants CNS-1413972 and CCF-1637532,
and by gifts from InfoSys and Keysight. 
}
\EventEditors{Alastair F. Donaldson}
\EventNoEds{1}
\EventLongTitle{33rd European Conference on Object-Oriented Programming (ECOOP 2019)}
\EventShortTitle{ECOOP 2019}
\EventAcronym{ECOOP}
\EventYear{2019}
\EventDate{July 15--19, 2019}
\EventLocation{London, United Kingdom}
\EventLogo{}
\SeriesVolume{134}
\ArticleNo{2}
\nolinenumbers 
\begin{document}

\maketitle

\begin{abstract}    
  The P4 programming language offers high-level, declarative
  abstractions that bring the flexibility of software to the domain of
  networking. Unfortunately, the main abstraction used to represent
  packet data in P4, namely header types, lacks basic safety guarantees.
  Over the last few years, experience with an increasing number of
  programs has shown the risks of the unsafe approach, which
  often leads to subtle software bugs.
    
  This paper proposes \name, a domain-specific language for
  programmable data planes in which all packet data is guaranteed to
  have a well-defined meaning and satisfy essential safety guarantees.
  We equip \name with a formal semantics and a static type system that
  statically guarantees header validity---a common source of safety
  bugs according to our analysis of real-world P4 programs. Statically
  ensuring header validity is challenging because the set of valid
  headers can be modified at runtime, making it a dynamic program
  property. Our type system achieves static safety by using a form of
  path-sensitive reasoning that tracks dynamic information from
  conditional statements, routing tables, and the control plane. Our
  evaluation shows that \name's type system can effectively
  eliminate common failures in many real-world programs.
\end{abstract}

 \section{Introduction}

\setlength{\epigraphwidth}{.8\textwidth}

\epigraph{I couldn't resist the temptation to put in a null reference
  [...] This has led to innumerable errors, vulnerabilities, and
  system crashes, which have probably caused a billion dollars of pain
  and damage in the last forty years.} {---Tony Hoare}

Modern languages offer features such as type systems, structured
control flow, objects, modules, etc. that make it possible to express
rich computations in terms of high-level abstractions rather than
machine-level code. Increasingly, many languages also offer
fundamental safety guarantees---e.g., well-typed programs do not go
wrong~\cite{Milner:1978}---that make entire categories of programming
errors simply impossible.

Unfortunately, although computer networks are critical infrastructure,
providing the communication fabric that underpins nearly all modern
systems, most networks are still programmed using low-level languages
that lack basic safety guarantees. Unsurprisingly, networks are
unreliable and remarkably insecure---e.g., the first step in a
cyberattack often involves compromising a router or other network
device~\cite{OConnor:2018aa,Kumar:2016aa}.

Over the past decade, there has been a shift to more flexible
platforms in which the functionality of the network is specified in
software. Early efforts related to software-defined networking
(SDN)~\cite{McKeown:2008aa, Casado:2007aa}, focused on the control plane
software that computes routes, balances load, and enforces security
policies, and modeled the data plane as a simple pipeline operating on
a fixed set of packet formats. However, there has been recent interest
in allowing the functionality of the data plane itself to be specified
as a program---e.g., to implement new protocols, make more efficient
use of hardware resources, or even relocate application-level
functionality into the network~\cite{Jin:2017, Jin:2018a}. In
particular, the P4 language~\cite{Bosshart:2014aa} enables the
functionality of a data plane to be programmed in terms of declarative
abstractions such as header types, packet parsers, match-action
tables, and structured control flow that a compiler maps down to an
underlying target device.

Unfortunately, while a number of P4's features were clearly inspired
by designs found in modern languages, the central abstraction for
representing packet data---header types---lacks basic safety
guarantees.  To a first approximation, a P4 header type can be thought
of as a record with a field for each component of the header. For
example, the header type for an IPv4 packet, would have a 4-bit
version field, an 8-bit time-to-live field, two 32-bit fields for the
source and destination addresses, and so on.

According to the P4 language specification, an instance of a header
type may either be valid or invalid: if the instance is valid, then
all operations produces a defined value, but if it is invalid, then
reading or writing a field yields an undefined result. In practice,
programs that manipulate invalid headers can exhibit a variety of
faults including dropping the packet when it should be forwarded, or
even leaking information from one packet to the next. In addition,
such programs are also not portable, since their behavior can vary
when executed on different targets.

The choice to model the semantics of header types in an unsafe way was
intended to make the language easier to implement on high-speed
routers, which often have limited amounts of memory. A typical P4
program might specify behavior for several dozen different protocols,
but any particular packet is likely to contain only a small handful of
headers. It follows that if the compiler only needs to represent the
valid headers at run-time, then memory requirements can be reduced.
However, while it may have benefits for language implementers, the
design is a disaster for programmers---it repeats Hoare's ``mistake,''
and bakes an unsafe feature deep into the design of a language that
has the potential to become the de-facto standard in a
multi-billion-dollar industry.

This paper investigates the design of a domain-specific language for
programmable data planes in which all packet data is guaranteed to
have a well-defined meaning and satisfy basic safety guarantees. In
particular, we present \name, a language with a precise semantics and
a static type system that can be used to obtain guarantees about the
validity of all headers read or written by the program. Although the
type system is mostly based on standard features, there are several
aspects of its design that stand out. First, to facilitate tracking
dependencies between headers---e.g. if the TCP header is valid, then
the IPv4 will also be valid---\name has an expressive algebra of types
that tracks validity information at a fine level of granularity.
Second, to accommodate the growing collection of extant P4 programs
with only modest modifications, \name uses a path-sensitive type
system that incorporates information from conditional statements,
forwarding tables, and the control plane to precisely track validity.

To evaluate our design for \name, we formalized the language and its
type system in a core calculus and proved the usual progress and
preservation theorems. We also implemented the \name type system in an
OCaml prototype, \ourchecker, and applied it to a suite of open-source
programs found on GitHub such as \texttt{switch.p4}, a large P4
program that implements the features found in modern data center
switches (specifically, it includes over four dozen different
switching, routing, and tunneling protocols, as well as multicast,
access control lists, among other features). We categorize common
failures and, for programs that fail to type-check, identify the root
causes and apply repairs to make them well-typed. We find that most
programs can be repaired with low effort from programmers, typically
by applying a modest number of simple repairs.

Overall, the main contributions of this paper are as follows:
\begin{itemize}
\item We propose \name, a type-safe enhancement of the P4 language that
  eliminates all errors related to header validity. 
\item We formalize the syntax and semantics of \name in a core
  calculus and prove that the type system is sound.
\item We implement our type checker in an OCaml prototype, \ourchecker.
\item We evaluate our type system empirically on over a dozen
  real-world P4 programs and identify common errors and repairs.
\end{itemize}

The rest of this paper is organized as follows. Section
\ref{sec:background} provides a more detailed introduction to P4 and
elaborates on the problems this work addresses. Section
\ref{sec:semantics} presents the design, operational semantics and
type system of \name and reports our type safety result. The results
of evaluating \name in the wild are presented in Section
\ref{sec:evaluation}. Section \ref{sec:related} surveys related work
and Section \ref{sec:conc} summaries the paper and outlines topics for
future work.

\section{Background and Problem Statement}
\label{sec:background}

This section introduces the main features of P4 and highlights the
problems caused by the unsafe semantics for header types.

\subsection{P4 Language}

P4 is a domain-specific language designed for processing
packets---i.e., arbitrary sequences of bits that can be divided into
(i) a set of pre-determined \emph{headers} that determine how the
packet will be forwarded through the network, and (ii) a
\emph{payload} that encodes application-level data. P4 is designed to
be protocol-independent, which means it handles both packets with
standard header formats (e.g., Ethernet, IP, TCP, etc.) as well as
packets with custom header formats defined by the programmer.
Accordingly, a P4 program first \emph{parses} the headers in the input
packet into a typed representation. Next, it uses a
\emph{match-action pipeline} to compute a transformation on those
headers---e.g., modifying fields, adding headers, or removing them.
Finally, a \emph{deparser} serializes the headers back into into a
packet, which can be output to the next device. A depiction of this
abstract forwarding model is shown in Figure~\ref{fig:p4-pipeline}.

The match-action pipeline relies on a data structure called a
\emph{match-action table}, which encodes conditional processing. More
specifically, the table first looks up the values being tested against
a list of possible entries, and then executes a further snippet of
code depending on which entry (if any) matched. However, unlike
standard conditionals, the entries in a match-action table are not
known at compile-time. Rather, they are inserted and removed at
run-time by the control plane, which may be logically centralized (as
in a software-defined network), or it may operate as a distributed
protocol (as in a conventional network).

The rest of this section describes P4's typed representation, how the
parsers, and deparsers convert between packets and this typed
representation, and how control flows through the match-action
pipeline.

\begin{figure}[t]
\centering \includegraphics[width=0.65\textwidth]{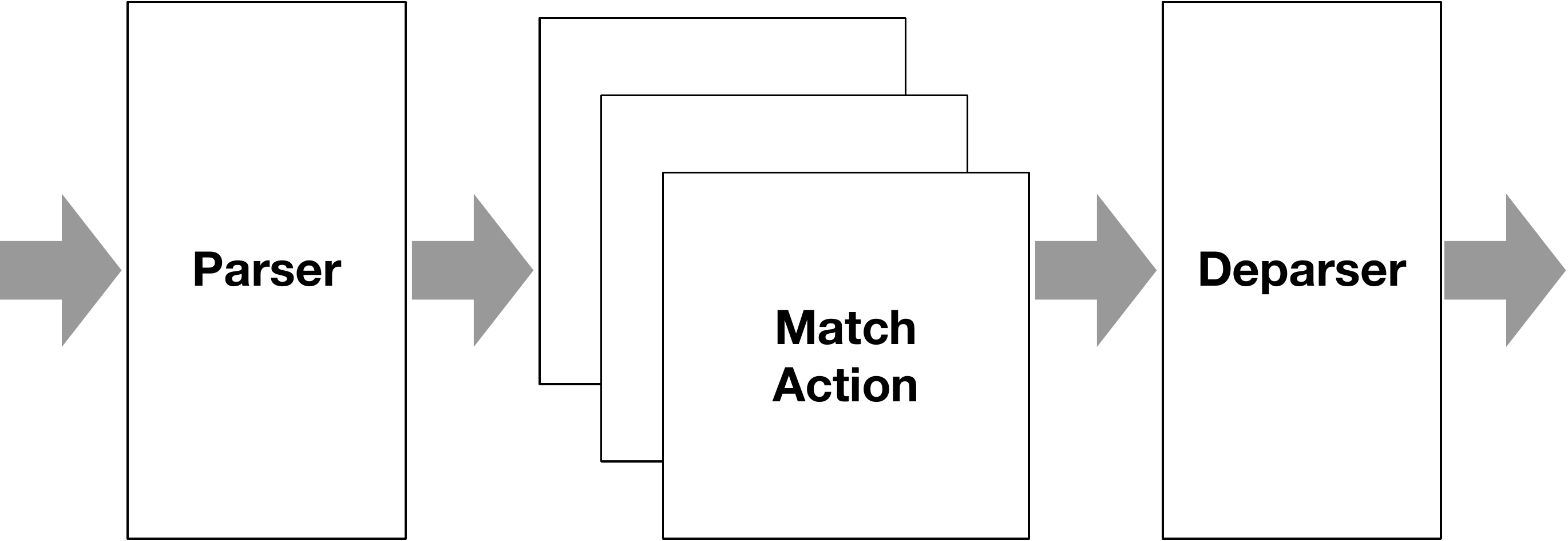}
\caption{Abstract forwarding model.}
\label{fig:p4-pipeline}
\end{figure}

\subparagraph{Header Types and Instances}
Header types specify the internal representation of packet data within
a P4 program. For example, the first few lines of the following
snippet of code:
\begin{lstlisting}[style=p4]
header_type ethernet_t {
  fields {
    dstAddr: 48;
    srcAddr: 48;
    etherType: 16;    
  }
}
header ethernet_t ethernet;
header ethernet_t inner_ethernet;
\end{lstlisting}
declare a type (\texttt{ethernet\_t}) for the Ethernet header with
fields \texttt{dstAddr}, \texttt{srcAddr}, and \texttt{etherType}. The
integer literals indicate the bit width of each field. The next two
lines declare two \texttt{ethernet\_t} instances (\texttt{ethernet}
and \texttt{inner\_ethernet}) with global scope.

\begin{figure}[t]
\noindent\begin{minipage}[c]{0.48\textwidth}
\includegraphics[width=\linewidth]{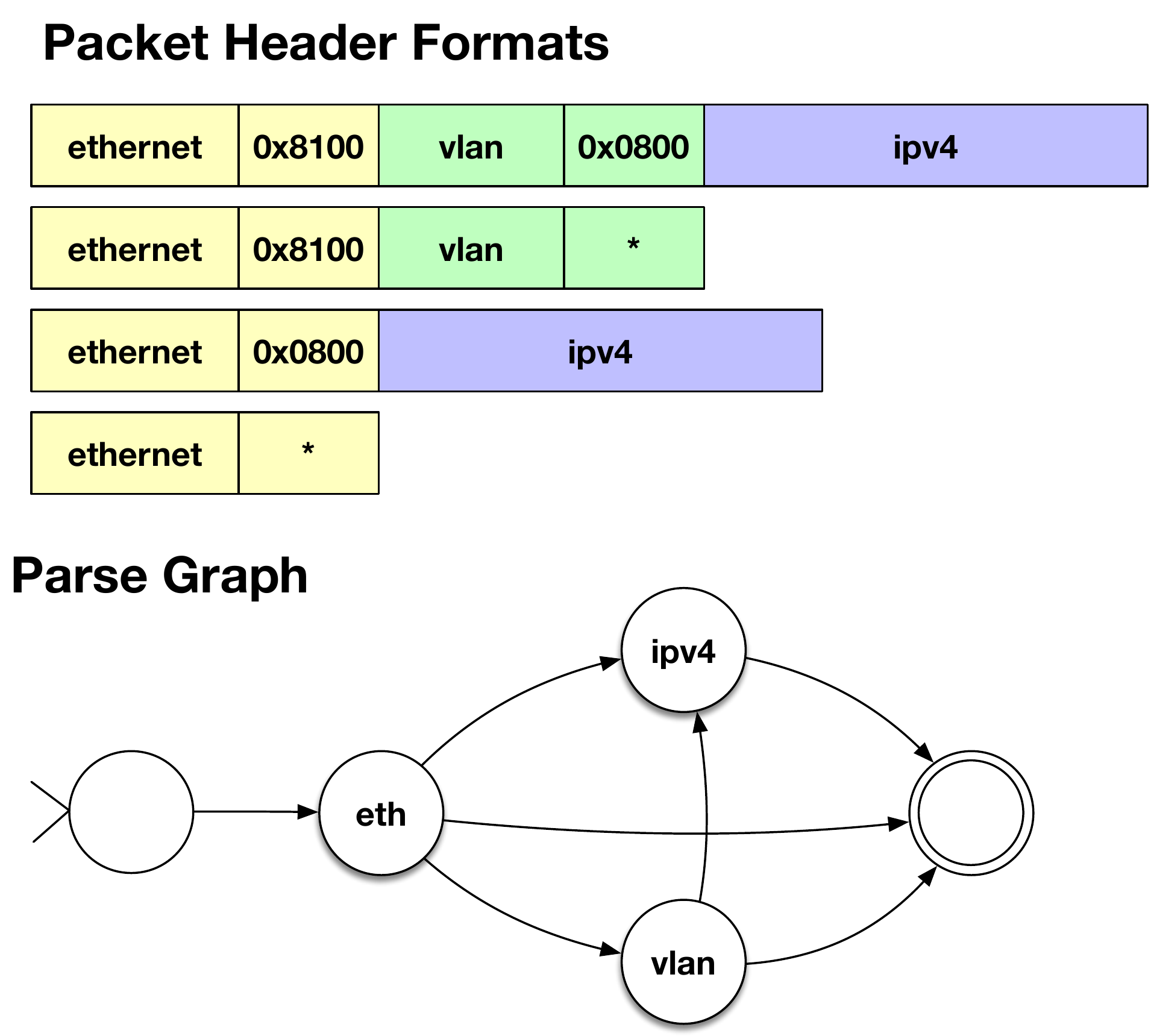}
\end{minipage}~\vrule~
\begin{minipage}[c]{0.5\textwidth}\vspace{4.5pt}
\begin{lstlisting}[style=p4,framexbottommargin=1.5pt]
parser start { 
  return parse_eth; 
}
parser parse_eth {
 extract(ethernet);
 return select(latest.etherType){
   0x8100 : parse_vlan;
   0x0800 : parse_ipv4;
   default: ingress;
  }
}
parser parse_vlan {
  extract vlan {
  return select(latest.etherType){
    0x0800: parse_ipv4;
    default: ingress;
  }
}
parser parse_ipv4 {
  extract(ipv4);
  return ingress;
}
\end{lstlisting}
\end{minipage}
  \caption{(Left) Header formats and parse graph that extracts an Ethernet header optionally followed by VLAN and/or IPv4 headers. (Right) P4 code implementing the same parser.}
  \label{fig:p4-parser}
\end{figure}

\subparagraph{Parsers}
A P4 parser specifies the order in which headers are extracted from
the input packet using a simple abstraction based on finite state
machines. Extracting into an header instance populates its fields with
the requisite bits of the input packet and marks the instance as
valid.
Figure~\ref{fig:p4-parser} depicts
a visual representation of a parse graph for three common headers:
Ethernet, VLAN, and IPv4. The instance \texttt{ethernet} is extracted
first, optionally followed by a \texttt{vlan} instance, or an
\texttt{ipv4} instance, or both.

\begin{figure}
  \begin{minipage}{0.4\textwidth}
\begin{lstlisting}[style=p4]
table forward {
 reads {
  ipv4 : valid;
  vlan : valid;
  ipv4.dstAddr: ternary;
 }
 actions = { 
  nop;
  next_hop; 
  remove;
 }
 default_action : nop();
}
\end{lstlisting}
\end{minipage}
\begin{minipage}{0.59\textwidth}
~~\textbf{Runtime Contents of} \texttt{forward}
\[\boxed{\begin{array}{l|l|l|l|l}
  \multicolumn{3}{c|}{\textbf{Pattern}} & \multicolumn{2}{c}{\textbf{Action}} \\
  \hline
  \texttt{ipv4} & \texttt{vlan} & \texttt{ipv4.dstAddr} & \multicolumn{1}{c|}{\textbf{Name}} & \multicolumn{1}{c}{\textbf{Data}}\\
  \hline\hline
  \texttt{1} & \texttt{0} & \texttt{10.0.0.*} & \texttt{next\_hop} &s,d \\
  \hline
  \texttt{0} & \texttt{1} & \texttt{*} & \texttt{remove} & \end{array}} \]
\end{minipage}
\caption{P4 tables. \texttt{forward} reads the validity of the
  \texttt{ipv4} and \texttt{vlan} header instances and the
  \texttt{dstAddr} field of the \texttt{ipv4} header instance, and
  calls one of its actions: \texttt{nop}, \texttt{next\_hop}, or \texttt{remove}.  }
\label{fig:p4-tables}
\end{figure}

\subparagraph{Tables and Actions}
The bulk of the processing for each packet in a P4 program is
performed using match-action tables that are populated by the control
plane. A table (such as the one in Figure~\ref{fig:p4-tables}) is
defined in terms of (i) the data it \texttt{reads} to determine a
matching entry (if any), (ii) the \texttt{actions} it may execute, and
(iii) an optional \texttt{default\_action} it executes if no matching
entry is found.

The behavior of a table depends on the entries installed at run-time
by the control-plane. Each table entry contains a match pattern, an
action, and action data. Intuitively, the match pattern specifies the
bits that should be used to match values, the action is the name of a
pre-defined function (such as the ones in
Figure~\ref{fig:p4-actions}), and the action data are the arguments to
that function.  Operationally, to process a packet, a table first
scans its entries to locate the first matching entry. If such a
matching entry is found, the packet is said to ``hit'' in the table,
and the associated action is executed. Otherwise, if no matching entry
is found, the packet is said to ``miss'' in the table, and the
\texttt{default\_action} (which is a no-op if unspecified) is
executed.

A table also specifies the \emph{match-kind} that describes how each
header field should match with the patterns provided by the control
plane. In this paper, we focus our attention on \texttt{exact},
\texttt{ternary}, and \texttt{valid} matches. An \texttt{exact} match
requires the bits in the packet be exactly equivalent to the bits in
the controller-installed pattern. A \texttt{ternary} match allows
wildcards in arbitrary positions, so the controller-installed pattern
\texttt{0*} would match bit sequences \texttt{00} and \texttt{01}. A
\texttt{valid} match can only be applied to a header instance and
simply checks the validity bit of that instance.

For example, in Figure~\ref{fig:p4-tables}, the \texttt{forward} table
is shown populated with two rules. The first rule tests whether
\texttt{ipv4} is valid, \texttt{vlan} is invalid, and the first 24
bits of \texttt{ipv4.srcAddr} equal \texttt{10.0.0}, and then applies
\texttt{next\_hop} with arguments $s$ and $d$ (which stand for source
and destination addresses). The second rule checks that \texttt{ipv4}
is invalid, then that \texttt{vlan} is valid, and skips evaluating the
value of \texttt{ipv4.dstAddr} (since it is wildcarded), to finally
apply the \texttt{remove} action.

Actions are functions containing sequences of primitive commands that
perform operations such as adding and removing headers, assigning a
value to a field, adding one field to another, etc. For example,
Figure~\ref{fig:p4-actions} depicts two actions: the
\texttt{next\_hop} action updates the Ethernet source and destination
addresses with action data from the controller; and the
\texttt{remove} action copies EtherType field from the \texttt{vlan}
header instance to the \texttt{ethernet} header instance and
invalidates the \texttt{vlan} header.

\begin{figure}[t]
\begin{minipage}{0.55\textwidth}
\begin{lstlisting}[style=p4]
action next_hop(src, dst) {
 modify_field(ethernet.srcAddr, src);
 modify_field(ethernet.dstAddr, dst);
 subtract_from_field(ipv4.ttl, 1); 
}
\end{lstlisting}
\end{minipage}
\begin{minipage}{0.45\textwidth}
\begin{lstlisting}[style=p4]
action remove() {
 modify_field(
   ethernet.etherType, 
   vlan.etherType);
 remove_header(vlan);
}
\end{lstlisting}
\end{minipage}
\caption{P4 actions.}
\label{fig:p4-actions}
\end{figure}

\subparagraph{Control}

A P4 control block can use standard control-flow constructs to execute
a pipeline of match-action tables in sequence. They manage the order
and conditions under which each table is executed. The
\texttt{ingress} control block begins to execute as soon as the parser
completes. The \texttt{apply} command executes a table and
conditionals branch on a boolean expression such as the validity of a
header instance.
\begin{lstlisting}[style=p4]
control ingress {
 if(valid(ipv4) or valid(vlan)) {
  apply(forward);
 }
}
\end{lstlisting}
The above code applies the \texttt{forward} table if one of
\texttt{ipv4} or \texttt{vlan} is valid.

\subparagraph{Deparser} The deparser reassembles the final output packet, 
after all processing has been done by serializing each valid header
instance in some order. In \PPPP{14}, the version of P4 we consider in
this paper, the compiler automatically generates the deparser from the
parser---i.e., for our example program, the deparser produces a packet
with Ethernet, VLAN (if valid), and IPv4 (if valid), in that order.

\subsection{Common Bugs in P4 Programs}
\label{sec:taxonomy}

Having introduced the basic features of P4, we now present five
categories of bugs found in open-source programs that arise due to
reading and writing invalid headers---the main problem that \name
addresses. There is one category for each of the following syntactic
constructs: (1) parsers, (2) controls, (3) table reads, (4) table
actions, and (5) default actions.

To identify the bugs we surveyed a benchmark suite of 15 research and
industrial P4 programs that are publicly available on GitHub and
compile to the BMv2~\cite{barefoot-bmv2:2018} backend. Later, in
Section~\ref{sec:evaluation}, we will report the number of occurrences
of each of these categories in our benchmark suite detected by our
approach.\footnote{We focus on $\text{P4}_{14}$ programs in this
  paper, but the issues we address also persist in the latest version
  of the language, $\text{P4}_{16}$. We did not consider
  $\text{P4}_{16}$ in this paper due to the smaller number of programs
  currently available.}

\subsubsection{Parser Bugs}
\label{sec:parser-too-general}

The first class of errors is due to the parser being too conservative
about dropping malformed packets, which increases the set of headers
that may be invalid in the control pipeline. In most programs, the
parser chooses which headers to \texttt{extract} based on the fields
of previously-extracted headers using P4's version of a switch
statement, \texttt{select}.
Programmers often fail to handle packets fall through to the
\texttt{default} case of these \texttt{select} statements.

An example from the \textsc{NetHCF}~\cite{Zhang:2018,Bai:2018}
codebase illustrates this bug. \textsc{NetHCF} is a research tool
designed to combat TCP spoofing.
As shown in Figure~\ref{fig:NetHCF-parser}, the parser handles TCP
packets in \texttt{parse\_ipv4} and redirects all other packets to the
\texttt{ingress} control. Unfortunately, the \texttt{ingress} control
(bottom right) does not check whether \texttt{tcp} is valid before
accessing \texttt{tcp.syn} to check whether it is equal to
\texttt{1}. This is unsafe since \texttt{tcp} is not guaranteed to be
valid even though it is required to be valid in the \texttt{ingress}
control.

\begin{figure}[t]
\begin{minipage}[t]{\textwidth}
\begin{minipage}[t]{0.495\textwidth}
\begin{lstlisting}[style=fixup]
+/* UNSAFE */+
~parser_exception unsupported {
  parser_drop;
}~
|parser| parse_ethernet {
 |extract|(ethernet);
 |return| select(ethernet.etherType) {
  0x0800 : parse_ipv4;
  |default| : ingress;
 
 }
}
|parser| parse_ipv4 {
 |extract|(ipv4);
 |return| select(ipv4.protocol) {
  6 : parse_tcp;
  |default| : ingress;

 }
}
\end{lstlisting}
\end{minipage}~\vrule~\begin{minipage}[t]{0.5\textwidth}
\begin{lstlisting}[style=fixup]
+/* SAFE */+
|parser_exception| unsupported {
  parser_drop;
}
@parser parse_ethernet {
 extract(ethernet);
 return select(ethernet.etherType) {
  0x0800 : parse_ipv4;
@  |default| :
   |parser_error| unsupported;
@  }
}@
@parser parse_ipv4 {
 extract(ipv4);
 return select(ipv4.protocol) {
 6 : parse_tcp;
@ |default| : |parser_error| unsupported;
@  }
}@
\end{lstlisting}
\end{minipage}
\end{minipage}
\hrule
\begin{minipage}[t]{0.495\textwidth}
\begin{lstlisting}[style=p4]
parser parse_tcp {
  extract(tcp);
  return ingress;
}
\end{lstlisting}
\end{minipage}\begin{minipage}[t]{0.485\textwidth}
\begin{lstlisting}[style=p4]
control ingress {
  if(tcp.syn == 1 and ...){...}
}
\end{lstlisting}
\end{minipage}
  \caption{Left: unsafe code in \textsc{NetHCF}; Right: our type-safe fix; Bottom: common code.}
  \label{fig:NetHCF-parser}
\end{figure}

To fix this bug, we can define a parser exception,
\texttt{unsupported}, with an handler that drops packets, thereby
protecting the \texttt{ingress} from having to handle unexpected
packets. Note however, that this fix might not be the best solution,
since it alters the original behavior of the program. However, without
knowing the programmer's intention, it is generally not possible to
automatically repair a program with undefined behavior.

\subsubsection{Control Bugs}
\label{sec:unsafe-reference}

Another common bug occurs when a table is executed in a context in
which the instances referenced by that table are not guaranteed to be
valid. This bug can be seen in the open-source code for
\textsc{NetCache}~\cite{Jin:2018,Jin:2017}, a system that uses P4 to
implement a load-balancing cache. The parser for \textsc{NetCache}
reserves a specific port (\texttt{8888}) to handle its special-purpose
traffic, a condition that is built into the parser, which extracts
\texttt{nc\_hdr} (i.e., the
\textsc{\underline{N}et\underline{C}ache}-specific header) only when
UDP traffic arrives from port \texttt{8888}. Otherwise, it performs
standard L2 and L3 routing. Unfortunately, the \texttt{ingress}
control node (Figure~\ref{fig:netcache-p4-unsafety}) tries to access
\texttt{nc\_hdr} before checking that it is valid. Specifically, the
\texttt{reads} declaration for the \texttt{check\_cache\_exists}
table, which is executed first in the \texttt{ingress} pipeline,
presupposes that \texttt{nc\_hdr} is valid. The invocation of the
\texttt{process\_value} table (not shown) contains another instance of
the same bug.

\begin{figure}[t]
\begin{minipage}{\textwidth}
\begin{minipage}[t]{0.485\textwidth}    
\begin{lstlisting}[style=p4]
/* UNSAFE */
control ingress { 

 process_cache();
 process_value();

 apply(ipv4_route);
}
\end{lstlisting}
\end{minipage}~\vrule~\begin{minipage}[t]{0.485\textwidth}
\begin{lstlisting}[style=fixup]
+/* SAFE */+
@control ingress {@
  if(|valid|(nc_hdr)) {
@    process_cache();
    process_value();
@ }@
 apply(ipv4_route);
}@
\end{lstlisting}
\end{minipage}
\end{minipage}
\hrule
\begin{minipage}{0.485\textwidth}
\begin{lstlisting}[style=p4]
control process_cache {
    apply(check_cache_exist);
    ...
}
\end{lstlisting}
\end{minipage}\begin{minipage}{0.485\textwidth}
\begin{lstlisting}[style=p4]
table check_cache_exist {
    reads { nc_hdr.key }
    actions { ... }
}
\end{lstlisting}
\end{minipage}
  \caption{Left: unsafe code in \textsc{NetCache}; Right: our type-safe fix; Bottom: Common code}
  \label{fig:netcache-p4-unsafety}
\end{figure}

To fix these bugs, we can wrap the calls to \texttt{process\_cache}
and \texttt{process\_value} in an conditional that checks the validity
of the header \texttt{nc\_hdr}. This ensures that \texttt{nc\_hdr} is valid when
\texttt{process\_cache} refers to it.

\subsubsection{Table Reads Bugs}
\label{sec:read-errors}

A similar bug arises in programs that contain tables that first match
on the validity of certain header instances before matching on the
fields of those instances. The advantage of this approach is that
multiple types of packets can be processed in a single table, which
saves memory. However, if implemented incorrectly, this programming
pattern can lead to a bug, in which the \texttt{reads} declaration
matches on bits from a header that may not be valid!

The \texttt{switch.p4} program exhibits an exemplar of this bug; it is
a ``realistic production switch'' developed by Barefoot Networks,
meant to be used ``as-is, or as a starting point for more advanced
switches''~\cite{Kodeboyina:2015}.

An archetypal example of table reads bugs is the
\texttt{port\_vlan\_mapping} table of \texttt{switch.p4}
(Figure~\ref{fig:switch-reads}).
This table is invoked in a context where it is not known which of the
VLAN tags is valid, despite containing references to both
\texttt{vlan\_tag\_[0]} and \texttt{vlan\_tag\_[1]} in the
\texttt{reads} declaration. Adroitly, the programmer has guarded the
references to $\texttt{vlan\_tag\_[}i\texttt{].vid}$ with keys that
test the validity of $\texttt{vlan\_tag\_[}i\texttt{]}$, for
$i=\texttt 1, \texttt 2$. Unfortunately, as written, it is impossible
for the control plane to install a rule that will always avoid reading
the value of an invalid header. The first match will check whether the
\texttt{vlan\_tag\_[0]} instance is invalid, which is safe. However,
the very next match will try to read the value of the
$\texttt{vlan\_tag\_[0].vid}$ field, even when the instance is
invalid! This attempt to access an invalid header results in undefined
behavior, and is therefore a bug.

It is worthy to note that this code is not actually buggy on some
targets---in particular, on targets where invalid headers are
initialized with \texttt{0}. However, \texttt{0}-initialization is not
prescribed by the language specification, and therefore this code is
not portable across other targets.

\begin{figure}[t]
\begin{minipage}{0.485\textwidth}
\begin{lstlisting}[style=p4]
/* UNSAFE */
table port_vlan_mapping { 
 reads {
  vlan_tag_[0] : valid;
  vlan_tag_[0].vid : exact;
  vlan_tag_[1] : valid;
  vlan_tag_[1].vid : exact;
 } ...
}
\end{lstlisting}
\end{minipage}~\vrule~
\begin{minipage}{0.485\textwidth}
\begin{lstlisting}[style=fixup]
+/* SAFE */+
@table port_vlan_mapping {
 reads {
  vlan_tag_[0] : valid;@
  vlan_tag_[0].vid : |ternary|;
@  vlan_tag_[1] : valid;
@  vlan_tag_[1].vid : |ternary|;@
 } ...
}@
\end{lstlisting}
\end{minipage}
  \caption{Left: a table in \texttt{switch.p4} with unprotected
    conditional reads; Right: our type-safe fix.}
  \label{fig:switch-reads}
\end{figure}

The naive solution to fix this bug is to refactor the table into four
different tables (one for each combination of validity bits) and then
check the validity of each header before the tables are invoked.
While this fix is perfectly safe, it can result in a combinatorial
blowup in the number of tables, which is clearly undesirable both for
efficiency reasons and because it requires modifying the control
plane.

Fortunately, rather than factoring the table into four tables, we can
replace the \texttt{exact} match-kinds with \texttt{ternary}
match-kinds, which permit matching with wildcards. In particular, the
control plane can install rules that match invalid instances using an
all-wildcard patterns, which is safe.

In order for this solution to typecheck, we need to assume that the
control plane is well-behaved---i.e. that it will install wildcards
for the \texttt{ternary} matches whenever the header is invalid. In
our implementation, we print a warning whenever we make this kind of
assumption so that the programmer can confirm that the control plane
is well-behaved.

\subsubsection{Table Action Bugs}
\label{sec:conflicting-actions}

Another prevalent bug, in our experience, arises when distinct actions
in a table require different (and possible mutually exclusive) headers
to be valid.
This can lead to two problems: (i) the control plane can populate the
table with unsafe match-action rules, and (ii) there may be no
validity checks that we can add to the control to make all of the
actions typecheck.

The \texttt{fabric\_ingress\_dst\_lkp} table
(Figure~\ref{fig:switch-buggy}) in \texttt{switch.p4} provides an
example of this misbehavior.
The \texttt{fabric\_ingress\_dst\_lkp} table reads the value of
\texttt{fabric\_hdr.dstDevice} and then invokes one of several
actions: \texttt{term\_cpu\_packet},
\texttt{term\_fabric\_unicast\_packet}, or
\texttt{term\_fabric\_multicast\_packet}. Respectively, these actions
require the \texttt{fabric\_hdr\_cpu}, \texttt{fabric\_hdr\_unicast},
and \texttt{fabric\_hdr\_multicast} (respectively) headers to be
valid.
Unfortunately the validity of these headers is mutually
exclusive.\footnote{There are other actions in the real
  \texttt{fabric\_ingress\_dst\_lkp}, but these three actions
  demonstrate the core of the problem.}

\begin{figure}[ht]
\begin{minipage}{0.485\textwidth}
\begin{lstlisting}[style=p4]
/* UNSAFE */
table fabric_ingress_dst_lkp {
 reads {
  fabric_hdr.dstDevice : exact;
 }



 actions {
  term_cpu_packet; 
  term_fabric_unicast_packet;
  term_fabric_multicast_packet;
 }
} 
\end{lstlisting}
\end{minipage}~\vrule~\begin{minipage}{0.485\textwidth}
\begin{lstlisting}[style=fixup]
+/* SAFE */+
@table fabric_ingress_dst_lkp {
 reads {
  fabric_hdr.dstDevice : exact;
@  fabric_hdr_cpu : |valid|;
  fabric_hdr_unicast: |valid|;
  fabric_hdr_multicast: |valid|;
@ }
 actions {
  term_cpu_packet; 
  term_fabric_unicast_packet;
  term_fabric_multicast_packet;
 }
}@ 
\end{lstlisting}
\end{minipage}
  \caption{Left: unsafe code in \texttt{switch.p4}; Right: our type-safe fix.}
  \label{fig:switch-buggy}
\end{figure}

Since \texttt{fabric\_hdr\_cpu}, \texttt{fabric\_hdr\_unicast}, and
\texttt{fabric\_hdr\_multicast} are mutually exclusive, there is no
single context that makes this table safe.
The only facility the table provides to determine which action should
be called is \texttt{fabric\_hdr.dstDevice}.
However, the P4 program doesn't establish a relationship between the
value of \texttt{fabric\_hdr.dstDevice} and the validity of any of
these three header instances.
So, the behavior of this table is only well-defined when the input
packets are well-formed, an unreasonable expectation for real
switches, which may receive \emph{any} sequence of bits ``on the
wire.''

We fix this bug by including validity matches in the \texttt{reads}
declaration, as shown in Figure~\ref{fig:switch-buggy}.  As in
Section~\ref{sec:read-errors}, this solution avoids combinatorial
blowup and extensive control plane refactoring.

In order to type-check this solution, we need to make an assumption
about the way the control plane will populate the table.
Concretely, if an action $a$ only typechecks if a header $h$ is valid,
and $h$ is not necessarily valid when the table is applied, we assume
that the control plane will only call $a$ if $h$ is matched as valid.
For example, \texttt{fabric\_hdr\_cpu} is not known to be valid when
(the fixed version of) \texttt{fabric\_ingress\_dst\_lkp} is applied,
so we assume that the control plane will only call action
\texttt{term\_cpu\_packet} when \texttt{fabric\_hdr\_cpu} is matched
as valid.
Again, our implementation prints these assumptions as warnings to the
programmer, so they can confirm that the control plane will satisfy
these assumptions.

\subsubsection{Default Action Bugs}
\label{sec:default-action}

Finally, the \emph{default action} bugs occur when the programmer
incorrectly assumes that a table performs some action when a packet
misses.
The \textsc{NetCache} program (described in
Section~\ref{sec:unsafe-reference}) exhibits an example of this bug,
too. The bug is shown in Figure~\ref{fig:default-action}, where the
table \texttt{add\_value\_header\_1} is expected to make the
\texttt{nc\_value\_1} header valid, which is done in the
\texttt{add\_value\_header\_1\_act} action.
The control plane may refuse to add any rules to the table, which
would cause all packets to miss, meaning that the
\texttt{add\_value\_header\_1\_act} action would never be called and
\texttt{nc\_value\_1} may not be valid.
To fix this error, we simply set the default action for the table to
\texttt{add\_value\_header\_1\_act}, which will force the table to
remove the header no matter what rules the controller installs.

\begin{figure}[ht]
\begin{minipage}{0.485\textwidth}
\begin{lstlisting}[style=p4]
/* UNSAFE */
table add_value_header_1 {
 actions {
  add_value_header_1_act;
 }

}
\end{lstlisting}
\end{minipage}~\vrule~
\begin{minipage}{0.485\textwidth}
\begin{lstlisting}[style=fixup]
+/* SAFE */+
@table add_value_header_1 {
 actions {
  add_value_header_1_act;
 }
@ |default_action| : add_value_header_1_act();@
}@
\end{lstlisting}
\end{minipage}
\caption{Left: unsafe code in \textsc{NetCache}; Right: our type-safe fix.}
\label{fig:default-action}
\end{figure}

\subsection{A Typing Discipline to Eliminate Invalid References}

In this paper, we propose a type system to increase the safety of P4
programs by detecting and preventing the classes of bugs defined in
Section~\ref{sec:taxonomy}. 
These classes of bugs all manifest when a program attempts to access
an invalid header---differentiating themselves only in their syntactic
provenance.
The type system that we present in the next section uses a
path-sensitive analysis, coupled with occurrence typing
\cite{Tobin-Hochstadt:2010aa}, to keep track of which headers are
guaranteed to be available at any program point---rejecting programs
that reference headers that \emph{might} be uninitialized---thus,
preventing all references to invalid headers.

Of course, in general, the problem of deciding header-validity can
depend on arbitrary data, so a simple type system cannot hope to fully
determine all scenarios when an instance will be valid. Indeed,
programmers often use a variety of data-dependent checks to ensure
safety. For instance, the control snippet shown on the left-hand side
of Figure \ref{fig:header-validation} will not produce undefined
behavior, given a parser that  chooses between parsing an
$\texttt{ipv4}$ header when $\texttt{ethernet.etherType}$ is
$\texttt{0x0800}$, an $\texttt{ipv6}$ header when
$\texttt{ethernet.etherType}$ is $\texttt{0x86DD}$, and throws a
parser error otherwise.

\begin{figure}[t]
\begin{minipage}{0.63\textwidth}
\begin{lstlisting}[style=p4]
if(ethernet.etherType == 0x0800) {
 apply(ipv4_table); 
} else if(ethernet.etherType == 0x086DD) {
 apply(ipv6_table);
}
\end{lstlisting}
\end{minipage}~\vrule~\begin{minipage}{0.37\textwidth}
\begin{lstlisting}[style=p4]
if(valid(ipv4)) {
 apply(ipv4_table);
} else if(valid(ipv6)) {
 apply(ipv6_table);
}
\end{lstlisting}
\end{minipage}
\caption{Left: data-dependent header validation; Right: syntactic
  header validation.}
\label{fig:header-validation}
\end{figure}

While this code is safe in this very specific context, it quickly
becomes unsafe when ported to other contexts.
For example in \texttt{switch.p4}, which performs tunneling, the
egress control node copies the \texttt{inner\_ethernet} header into
the \texttt{ethernet}; however the \texttt{inner\_ethernet} header may
not be valid at the program point where the copy is performed.
This behavior is left undefined~\cite{P4:2017aa}, a target is free to
read arbitrary bits, in which case it could decide to call the
\texttt{ipv4\_table} despite \texttt{ipv4} being invalid.

To improve the maintainability and portability of the code, we can
replace the data-dependent checks with validity checks, as illustrated
by the control snippet shown on the right-hand side of
Figure~\ref{fig:header-validation}. The validity checks assert
precisely the preconditions for calling each table, so that no matter
what context this code snippet is called in, it is impossible for the
\texttt{ipv4\_table} to be called when the \texttt{ipv4} header is
invalid.

In the next section, we develop a core calculus for \name with a type
system that eliminates references to invalid headers, encouraging
programers to replace data-dependent checks with header-validity
checks.
 \section{\name}
\label{sec:semantics}
This section discusses our design goals for \name and the choices we
made to accommodate them, and formalizes the language's syntax,
small-step semantics, and type system.

\subsection{Design}

Our primary design goal for \name is to develop a core calculus that
models the main features of \PPPP{14} and \PPPP{16}, while
guaranteeing that all data from packet headers is manipulated in a
safe and well-defined manner. We draw inspiration from Featherweight
Java~\cite{Igarashi:2001aa}---i.e., we model the essential features of
P4, but prune away unnecessary complexity. The result is a minimal
calculus that is easy to reason about, but can still express a large
number of real-world data plane programs. For instance, P4 and \name
both achieve protocol independence by allowing the programmer to
specify the types of packet headers and their order in the bit stream.
Similarly, \name mimics P4's use of tables to interface with the
control-plane and decide which actions to execute at run-time. 

So what features does \name prune away? We omit a number of constructs
that are secondary to how packets are processed---e.g.,
\texttt{field\_list\_calculations}, \texttt{parser\_exceptions},
\texttt{counters}, \texttt{meters}, \texttt{action profiles}, etc. It
would be relatively straightforward to add these to the
calculus---indeed, most are already handled in our prototype---at the
cost of making it more complicated. We also modify or distill several
aspects of P4. For instance, P4 separates the parsing phase and the
control phase. Rather than unnecessarily complicating the syntax of
\name, we allow the syntactic objects that represent parsers and
controls to be freely mixed. We make a similar simplification in
actions, informally enforcing which primitive commands can be invoked
within actions (e.g., field modification, but not conditionals).

Another challenge arises in trying to model core behaviors of both
\PPPP{14} and \PPPP{16}, in that they each have different type systems
and behaviors for evaluating expressions. Our calculus abstracts away
expression typing and syntax variants by assuming that we are given a
set of constants $k$ that can represent values like \texttt{0} or
\texttt{True}, or operators such as $\texttt{\&\&}$ and $\texttt{?:}$.
We also assume that these operators are assigned appropriate (i.e.,
sound) types. With these features in hand, one can instantiate our
type system over arbitrary constants.

Another departure from P4 is related to the \syntax{add} command, which
presents a complication for our expression types.
The analogous \texttt{add\_header} action in \PPPP{14} simply modifies the
validity bit, without initializing any of the fields. This means that
accessing any of the header fields before they have been manually
initialized reads a non-deterministic value. Our calculus neatly
sidesteps this issue by defining the semantics of the
\syntax{add(h)} primitive to initialize each of the fields of $h$ to a
default value. We assume that along with our type constants there is a
function $\syntax{init}$ that accepts a header type $\eta$ and
produces a header instance of type $\eta$ with all fields set to their
default value. Note that we could have instead modified our type
system to keep track of the definedness of header fields as well as
their validity. However, for simplicity we choose to focus on
header validity in this paper.

The portion of our type system that analyzes header validity, requires
some way of keeping track of which headers are valid. Naively, we can
keep track of a set of which headers are guaranteed to be valid on all
program paths, and reject programs that reference headers not in this set. However, this
coarse-grained approach would lead to a large number of false positives.
For instance, the parser shown in Figure~\ref{fig:p4-parser} parses an
\texttt{ethernet} header and then either boots to \texttt{ingress} or
parses an \texttt{ipv4} header and then either proceeds to the
\texttt{ingress} or parses an \texttt{vlan} header.
Hence, at the \texttt{ingress} node, the only header that is
guaranteed to be valid is the \texttt{ethernet} header.
However, it is certainly safe to write an \texttt{ingress} program
that references the \texttt{vlan} header after checking it was valid.
To reflect this in the type system we introduce a special construct
called $\syntax{valid(h)~c_1~else~c_2}$, which executes $c_1$
if $h$ is valid and $c_2$ otherwise.
When we type check this command, following previous work on occurrence
typing~\cite{Tobin-Hochstadt:2010aa}, we check $c_1$ with the
additional fact that $h$ is valid, and we check $c_2$ with the
additional fact that $h$ is not valid.

Even with this enhancement, this type system would still be overly
restrictive. To see why, let us augment the parser from
Figure~\ref{fig:p4-parser} with the ability to parse TCP and UDP
packets: after parsing the \texttt{ipv4} header, the parser can
optionally extract the \texttt{vlan}, \texttt{tcp}, or \texttt{udp}
header and then boot control flow to ingress.
Now suppose that we have a table \texttt{tcp\_table} that refers to
both \texttt{ipv4} and \texttt{tcp} in its \texttt{reads} declaration,
and that \texttt{tcp\_table} is (unsafely) applied immediately in the
\texttt{ingress}.
Because the validity of \texttt{tcp} implies the validity of
\texttt{ipv4}, it should be safe to check the validity of \texttt{tcp}
and then apply \texttt{tcp\_table}. However, using the representation
of valid headers as a set, we would need to ascertain the validity of
\texttt{ipv4} and of \texttt{tcp}.

To solve this problem, we enrich our type representation to keep track
of dependencies between headers.
More specifically, rather than representing all headers guaranteed to
be valid in a set, we use a finer-grained representation---a set of
sets of headers that might be valid at the current program point.
For a given header reference to be safe, it must to be a member of all
possible sets of headers---i.e., it must be valid on all paths through
the program that reach the reference.

Overall, the combination of an expressive language of types and a
simple version of occurrence typing allows us to capture
dependencies between headers and perform useful static analysis of the
dynamic property of header validity.

The final challenge with formally modelling P4 lies in its interface
with the control-plane, which populates the tables and provides
arguments to the actions.
While the control-plane's only methodology for managing switch behavior
is to populate the match-action tables with forwarding entries, it is
perfectly capable of producing undefined behavior.
However, if we assume that the controller is well-intentioned, we can
prove the safety of more programs.

In our formalization, to streamline the presentation, we model the
control plane as a function $\mathcal{CA} (t, H) = (a_i, \bar v)$ that
takes in a table $t$ and the current headers $H$ and produces the
action to call $a_i$ and the (possibly empty) action data arguments
$\bar v$.
We also use a function $\mathcal{CV} (t) = \bar S$ that
  analyzes a table $t$ and produces a list of sets of valid headers
  $\bar S$, one set for each action, that can be safely assumed valid
  when the entries are populated by the control plane.  From the table
  declaration and the header instances that can be assumed valid,
  based on the match-kinds, we can derive a list of match key
  expressions $\bar e$ that must be evaluated when the table is
  invoked.
Together, these functions model the run-time interface between the
switch and the controller.
In order to prove progress and preservation, we assume that
$\mathcal{CV}$ and $\mathcal{CA}$ satisfy three simple correctness
properties: (1) the control plane can safely install table entries
that never read invalid headers, (2) the action data provided by the
control plane has the types expected by the action, and (3) the
control plane will only assume valid headers for an action that are
valid for a given packet.
See
\iftoggle{techreport}{Appendix~\ref{app:control-plane-assumptions}}{technical
    report} for details.

\subsection{Syntax}

The syntax of \name is shown in Figure \ref{fig:syntax}. To
lighten the notation, we write $\bar{x}$ as shorthand for a (possibly
empty) sequence $x_1,...,x_n$.

\begin{figure}[t!]
  \begin{minipage}{0.52\textwidth}
    \[\begin{array}{@{}l@{~}c@{~~}l@{~}l@{}}
           \multicolumn{4}{@{}l}{\textbf{Commands}}\\
              c & ::= &  & \\
                & \mid & \syntax{extract(h)} & \textsc{extraction} \\
                & \mid & \syntax{emit(h)} & \textsc{deparsing} \\
                & \mid & \syntax{c_1;c_2} & \textsc{sequence}^\ast\\
                & \mid & \syntax{if(e)\ c_1\ else\ c_2} & \textsc{conditional}\\
                & \mid & \syntax{valid(h)\ c_1\ else\ c_2}  & \textsc{validity}\\
                & \mid & \syntax{t.apply()} & \textsc{application}\\
                & \mid & \syntax{skip} & \textsc{skip}\\
                & \mid & \syntax{add(h)} & \textsc{addition}^\ast\\
                & \mid & \syntax{remove(h)} & \textsc{removal}^\ast \\
                & \mid & \syntax{h.f = e} & \textsc{modification}^\ast\\[1em]
            \multicolumn{4}{@{}l}{\textbf{Actions}}\\
              a & ::= & \syntax{\lambda \bar{x}.c} & \textsc{action} \\[1em]
        \multicolumn{4}{@{}l}{\textbf{Expressions}}\\
           e & ::= \\
             & \mid & \syntax{v} & \textsc{values}\\
             & \mid & \syntax{h.f} & \textsc{header field} \\
             & \mid & \syntax{x} & \textsc{variable}\\ 
             & \mid & \syntax{k^n} & \textsc{constant}
    \end{array}\]
  \end{minipage}\vrule
  \begin{minipage}{0.45\textwidth}
    \[\begin{array}{>{\quad}lcll}
           \multicolumn{4}{l}{\textbf{Declarations}}\\
           d & ::= & \\
             & \mid & \syntax{t(\overbar{h},\overbar{(e, m)},\overbar{a})} & \textsc{table}\\ 
             & \mid & \syntax{\eta\ \{\overbar{f: \tau} \}} & \textsc{header type}\\
             & \mid & \syntax{h \mapsto \eta} & \textsc{instantiation}
      \end{array}\]
    \[\begin{array}{lcl>{\quad}lcl}
        \multicolumn{3}{l}{\textbf{Match Kinds}}      & \multicolumn{3}{l}{\textbf{Constants}}\\
        m & \in & \{\syntax{exact},\syntax{ternary}\} & k & \in & K\\
        \multicolumn{3}{l}{\textbf{Program}}          & \multicolumn{3}{l}{\textbf{Values}} \\
        \mathcal{P} & ::= & (\bar{d},c)               & v & \in & V \\
      
    \end{array}\]
    \hrule
    \vspace{-.5em}
    \[\begin{array}{lcll}
      \multicolumn{4}{l}{\textbf{Header Types}}\\
      \Theta & ::= & \\
             & \mid & 0  & \textsc{contradiction}\\
             & \mid & 1  & \textsc{empty}\\
             & \mid & h & \textsc{instance} \\ 
             & \mid & \Theta_1 \cdot \Theta_2 & \textsc{concatenation}\\ 
             & \mid & \Theta_1 + \Theta_2 & \textsc{choice}\\
      \end{array}\]
    \[\begin{array}{lcl>{\quad}lcl}
      \multicolumn{3}{l}{\textbf{Action Types}} & \multicolumn{3}{>{\quad}l}{\textbf{Expression Types}} \\
        \alpha & ::= & \bar{\tau} \to \Theta    & \tau & ::= & \texttt{Bool} \\
                                                &&&&\mid & \bar \tau \to \tau \\
                                                &&&&\mid& \cdots \\
    \end{array}\]
    \vspace{.2em}
  \end{minipage}
\caption{Syntax of \name}
\label{fig:syntax}
\end{figure}
 
A \name program consists of a sequence of declarations $\bar d$ and a
command $c$. The set of declarations includes header types, header
instances, and tables.
Header type declarations describe the format of individual headers and
are defined in terms of a name and a sequence of field declarations.
The notation ``$f:\tau$'' indicates that field $f$ has type $\tau$. We
let $\eta$ range over header types.
A header instance declaration assigns a name $h$ to a header type
$\eta$.
The map $\mathcal{HT}$ encodes the (global) mapping between header
instances and header types.
Table declarations $t(\overbar{h}, {(e, m)}, \overbar a)$, are
defined in terms of a sequence of valid-match header instances
  $\overbar h$, a sequence of match-key expressions 
  $\overbar{(e,m)}$ read in the table, where $e$ is an
  expression and $m$ is the match-kind used to match this expression,
and a sequence of actions $\bar a$. The notation $t.\mathit{valids}$
denotes the valid-match instances, $t.\mathit{reads}$ denotes the
expressions, and $t.\mathit{actions}$ denotes the actions.

Actions are written as (uncurried) $\lambda$-abstractions. An action
$\lambda \bar x.~c$ declares a (possibly empty) sequence of
parameters, drawn from a fresh set of names, which are in scope for
the command $c$. The run-time arguments for actions (action data) are
provided by the control plane.
Note that we artificially restrict the commands that can be called in
the body of the action to addition, removal, modification and
sequence; these actions are identified with an asterisk in
Figure~\ref{fig:syntax}.

The calculus provides commands for extracting (\syntax{extract}),
creating (\syntax{add}), removing (\syntax{remove}), and modifying
(\syntax{h.f=e}) header instances.
The \syntax{emit} command is used in the deparser and serializes a
header instance back into a bit sequence (\syntax{emit}).
The \syntax{if}-statement conditionally executes one of two commands
based on the value of a boolean condition.
Similarly, the \syntax{valid}-statement branches on the validity of
$h$. Table application commands (\syntax{t.apply()}) are used to
invoke a table $t$ in the current state. 
The $\syntax{skip}$ command is a no-op.

The only built-in expressions in \name are variables $x$ and header
fields, written $h.f$. We let $v$ range over values and assume a
collection of $n$-ary constant operators $k^n \in K$. 

For simplicity, we assume that every header referenced in an
expression has a corresponding instance declaration. We also assume
that header instance names $h$, header type names $\eta$, variable
names $x$, and table names $t$ are drawn from disjoint sets of names
\textsc{h,e,v}, and \textsc{t} respectively and that each name is
declared only once.

\subsection{Type System}

\name provides two main kinds of types, basic types $\tau$ and header
types $\Theta$ as shown in Figure~\ref{fig:syntax}.
We assume that the set of basic types includes booleans (for
conditionals) as well as tuples and function types (for actions).

A header type $\Theta$ represents a set of possible co-valid header
instances. The type $0$ denotes the empty set. This type arises when
there are unsatisfiable assumptions about which headers are valid. The
type $1$ denotes the singleton denoting the empty set of headers. It
describes the type of the initial state of the program. The type $h$
denotes a singleton set, $\{ \{ h \} \}$---i.e., states where only $h$
is valid. The type $\Theta_1 \cdot \Theta_1$ denotes the set obtained
by combining headers from $\Theta_1$ and $\Theta_2$---i.e., a product
or concatenation. Finally, the type $\Theta_1 + \Theta_2$
denotes the union of $\Theta_1$ or $\Theta_2$, which intuitively
represents an alternative.

The semantics of header types, $\llbracket \Theta \rrbracket$, is
defined by the equations in Figure~\ref{fig:type-semantics-auxfn}.
Intuitively, each subset represents one alternative set of headers
that may be valid. For example, the header type $\mathtt{eth} \cdot
(\mathtt{ipv4} + 1)$ denotes the set $\{
\{\mathtt{eth},\mathtt{ipv4}\}, \{\mathtt{eth}\} \}$.

To formulate the typing rules for \name, we also define a set of
operations on header types: \texttt{Restrict}, \texttt{NegRestrict},
\texttt{Includes}, \texttt{Remove}, and \texttt{Empty}.
The restrict operator $\texttt{Restrict}~\Theta~h$ recursively
traverses $\Theta$ and keeps only those choices in which $h$ is
contained, mapping all others to $0$.
Semantically this has the effect of throwing out the subsets of
$\llbracket \Theta \rrbracket$ that do not contain $h$.
Dually $\texttt{NegRestrict}~\Theta~h$ produces only those
choices/subsets where $h$ is invalid.
$\texttt{Includes}~\Theta~h$ traverses $\Theta$ and checks that $h$
is always valid. 
Semantically this says that $h$ is a member of every element of
$\llbracket\Theta\rrbracket$.
$\texttt{Remove}~\Theta~h$ removes $h$ from every path, which
means, semantically that it removes $h$ from ever element of
$\llbracket\Theta\rrbracket$.
Finally, $\texttt{Empty}~\Theta$ checks whether $\Theta$ denotes the
empty set.
We can lift these operators to operate on sets of headers in the obvious way.
An in-depth treatment of these operators can be found in
\iftoggle{techreport}{Appendix~\ref{sec:appendix-header-ops}}{the
  accompanying technical report}.

\subsubsection{Typing Judgement}

The typing judgement has the form
$\cmdtype{\Gamma}{\Theta}{c}{\Theta'}$, which means that in variable
context $\Gamma$, if $c$ is executed in the header context $\Theta$,
then a header instance type $\Theta'$ is assigned. Intuitively, $\T$
encodes the sets of headers that may be valid when type checking a
command.
$\Gamma$ is a standard type environment which maps variables $x$ to
type $\tau$. 
If there exists $\Theta'$ such that
$\cmdtype \Gamma \Theta c {\Theta'}$, we say that $c$ is well-typed in
 $\Theta$. 

\begin{figure}[t]
    \begin{minipage}[c]{0.35\textwidth}
        \begin{align*}
            \llbracket \T \rrbracket & \subseteq \mathcal{P}(\mathit{Header}) \\
            \llbracket 0 \rrbracket &= \{ \} \\
            \llbracket 1 \rrbracket &= \{\{\}\} \\
            \llbracket h \rrbracket &= \{ \{h\} \} \\
            \llbracket \Theta_1 \cdot \Theta_2 \rrbracket &= \llbracket \Theta_1 \rrbracket \bullet \llbracket \Theta_2 \rrbracket \\
            \llbracket \Theta_1 + \Theta_2 \rrbracket &= \llbracket \Theta_1 \rrbracket \cup \llbracket \Theta_2 \rrbracket \\
        \end{align*} 
    \end{minipage}\quad\vrule
    \begin{minipage}[c]{0.55\textwidth}
        \begin{align*}
                    \mathcal{F}(h,f_i) &= \tau_i & \textit{Field lookup}\\
          \mathcal{A}(a) &= \lambda \bar x : \bar \tau.~c & \textit{Action lookup}\\
          \mathcal{CA}(t, H) &= (a_i, \bar{v}) & \textit{Control-plane actions} \\
          \mathcal{CV}(t) &= \bar{S} & \textit{Control-plane validity} \\
          \mathcal{H}(e) &= \bar h & \textit{Referenced Header instances}
                  \end{align*}
        \hrule
        \vspace{-0.5em}
        \begin{align*}
          \mathsf{maskable}(t, e, \mathit{exact}) &\triangleq \mathit{false} \\
          \mathsf{maskable}(t, e, \mathit{ternary}) &\triangleq \mathcal{H}(e) \subseteq t.\mathit{valids}
        \end{align*}
    \end{minipage}
    \caption{Semantics of header types (left) and auxiliary functions (right).}
    \label{fig:type-semantics-auxfn}
\end{figure}
 
The typing rules rely on several auxiliary definitions shown in
Figure~\ref{fig:type-semantics-auxfn}. The field type lookup function
$\mathcal{F}(h, f_i)$ returns the type assigned to a field $f_i$ in
header $h$ by looking it up from the global header type
  declarations via the header instance declarations. The action
lookup function $\mathcal{A}(a)$ returns the action definition
$\lambda \bar{x} : \bar{\tau}.~c$ for action $a$. Finally, the
function $\mathcal{CA}(t,H)$ computes the run-time actions for table
$t$, while $\mathcal{CV}(t)$ computes $t$'s assumptions about
validity. Both of these are assumed to be instantiated by the control
plane in a way that satisfies basic correctness properties---see
\iftoggle{techreport}{Appendix~\ref{app:control-plane-assumptions}}{technical
    report}.

\begin{figure}[t]
    \begin{minipage}{0.45\textwidth}
        \begin{mathpar}
            \inferrule[T-Zero]{
                \mathtt{Empty}~\Theta_1
            }{
              \cmdtype{\Gamma}{\Theta_1}{c}{\Theta_2}
            }
            \and
            \inferrule[T-Skip]{
              \ 
            }{
              \cmdtype{\Gamma}{\Theta}{\syntax{skip}}{\Theta}
            }
            \and
            \inferrule[T-Seq]{
              \cmdtype{\Gamma}{\Theta}{c_1}{\Theta_1}\\
              \cmdtype{\Gamma}{\Theta_1}{c_2}{ \Theta_2}
            }{
              \cmdtype{\Gamma}{\Theta}{c_1 ; c_2}{\Theta_2}
            }
            \and
            \inferrule[T-If]{
                \Gamma; \Theta \vdash e : Bool \\\\
                \cmdtype{\Gamma}{\Theta}{c_1}{\Theta_1} \\
                \cmdtype{\Gamma}{\Theta}{c_2}{\Theta_2}
            }{
              \cmdtype{\Gamma}{\Theta}{\syntax{if~(e)~c_1~else~c_2}}{\Theta_1 + \Theta_2}
            }
            \and
            \inferrule[T-IfValid]{
              \cmdtype{\Gamma}{\texttt{Restrict}~\Theta~h}{c_1}{\Theta_1} \\
                \cmdtype{\Gamma}{\texttt{NegRestrict}~\Theta~h}{c_2}{\Theta_2}
            }{
                \cmdtype{\Gamma}{\Theta}{\syntax{valid(h)~c_1~else~c_2}}{\Theta_1 + \Theta_2}
            }
            \and
            \inferrule[T-Mod]{
                \texttt{Includes}~\Theta~h \\\\
                \mathcal{F}(h, f) = \tau_i\\
                \Gamma; \Theta \vdash e : \tau_i
            }{
              \cmdtype{\Gamma}{\Theta}{h.f = e}{\Theta}
            }           
        \end{mathpar}
    \end{minipage}~\vrule~
    \begin{minipage}{0.47\textwidth}
        \begin{mathpar}
            \inferrule[T-Extr]{
                \
            }{
              \cmdtype{\Gamma}{\Theta}{extract(h)}{\Theta \cdot h}
            }
            \and
            \inferrule[T-Emit]{
                \
            }{
              \cmdtype{\Gamma}{\Theta}{emit(h)}{\Theta}
            }
            \and
            \inferrule[T-Add]{
            \
            }{
              \cmdtype{\Gamma}{\Theta}{\syntax{add(h)}}{\Theta \cdot h}
            }
            \and
            \inferrule[T-Rem]{
                \
            }{
              \cmdtype{\Gamma}{\Theta}{\syntax{remove(h)}}{\texttt{Remove}~\Theta~h}
            }
            \and
            \inferrule[T-Apply]{
              \mathcal{CV}(t) = \bar{S}\\
              t.\mathit{actions} = \bar{a} \\
              t.\mathit{reads} = \bar{r} \\
              \bar{e} = \{e_j \mid (e_j, m_j) \in \bar{r} \wedge \neg \mathsf{maskable}(t,e_j,m_j)\} \\\\
              \cdot ; \Theta \vdash e_j : \tau_j ~~~\textrm{for}~e_j \in \bar{e} \\\\
              \texttt{Restrict}~\Theta~S_i \vdash a_i : \bar{\tau}_i \rightarrow \Theta_i'~~~\textrm{for}~a_i \in \bar{a} 
            }{
              \cmdtype{\Gamma}{\Theta}{t.apply()}{\left(\sum_{a_i \in \bar a} \Theta_i'\right)}
            }
        \end{mathpar}
    \end{minipage}
    \caption{Command typing rules for \name}
    \label{fig:command-typing-rules}
\end{figure}
 
The typing rules for commands are presented in Figure
\ref{fig:command-typing-rules}.
The rule \textsc{T-Zero} gives a command an arbitrary output type if the input type is empty. It is needed to prove preservation.
The rules \textsc{T-Skip} and \textsc{T-Seq} are standard.
The rule \textsc{T-If} a path-sensitive union type between the type
computed for each branch.
The rule \textsc{T-IfValid} is similar, but leverages knowledge about
the validity of $h$. So the true branch $c_1$ is checked in the context
$\texttt{Restrict}~\Theta~h$, and the false branch $c_2$ is
checked in the context $\texttt{NegRestrict}~\Theta~h$. The top-level output
type is the union of the resulting output types for $c_1$ and $c_2$.
The rule \textsc{T-Mod} checks that $h$ is guaranteed to be valid
using the \texttt{Includes} operator, and uses the auxiliary function
$\mathcal{F}$ to obtain the type assigned to $h.f$. Note that the set
of valid headers does not change when evaluating an assignment, so the
output and input types are identical.
The rules \textsc{T-Extr} and \textsc{T-Add} assign header extractions
and header additions the type $\Theta \cdot h$, reflecting the fact
that $h$ is valid after the command executes. 
Emitting packet headers does not change the set of valid headers,
which is captured by rule \textsc{T-Emit}.
The typing rule \textsc{T-Rem} uses the $\texttt{Remove}$ operator to
remove $h$ from the input type $\Theta$.
Finally, the rule \textsc{T-Apply} checks table applications. To
understand how it works, let us first consider a simpler, but less
precise, typing rule:
\begin{mathpar}
\inferrule{
    t.\mathit{reads} = \bar{e}\\
    \cdot ; \Theta \vdash e_i : \tau_i ~~~\textrm{for}~e_i \in \bar{e} \\\\
    t.\mathit{actions} = \bar{a} \\
    \cdot;\Theta \vdash a_i : \bar{\tau_i} \rightarrow \Theta_i'~~~\textrm{for}~a_i \in \bar{a} 
}{
  \cmdtype{\cdot}{\Theta}{t.apply()}{\left(\sum \Theta_i'\right)}
}
\end{mathpar}
Intuitively, this rule says that to type check a table application,
we check each expression it reads and each of its actions. The final
header type is the union of the types computed for the actions. To put
it another way, it models table application as a
non-deterministic choice between its actions. However, while this rule
is sound, it is overly conservative. In particular, it does not model
the fact that the control plane often uses header validity bits to
control which actions are executed. 

Hence, the actual typing rule, \textsc{T-Apply}, is parameterized on a
function $\mathcal{CV}(t)$ that models the choices made by the control
plane, returning for each action $a_i$, a set of headers $S_i$ that
can be assumed valid when type checking $a_i$. From the reads
declarations of the table declaration, we can derive a subset of the
expressions read by the table---e.g., excluding expressions that can
be wildcarded when certain validity bits are false. This is captured
by the function $\mathsf{maskable}(t, e, m)$ (defined in
Figure~\ref{fig:type-semantics-auxfn}) , which determines whether a
reads expression $e$ with match-kind $m$ in table $t$ can be masked
using a wild-card. The $\mathsf{maskable}$ function is defined using
$\mathcal{H}(e)$, which returns the set of header instances referenced
by an expression $e$.

In the example from Section \ref{sec:read-errors}, if an action $a_j$
is matched by the rule $(0,*,0,*)$, both $S_j$ and $e_j$ are empty.

\begin{figure}[t]
  \begin{mathpar}
      \inferrule{
        \cmdtype{\Gamma, \bar{x}: \bar{\tau}}{\Theta}{c}{\Theta'}
      }{
          \Gamma ; \Theta  \vdash \lambda \bar{x}:\bar \tau.c : \bar \tau \rightarrow \Theta'
      }\quad(\textsc{T-Action})
    \end{mathpar}
    \caption{Action typing rule for \name}
    \label{fig:action-typing-rules}
\end{figure}

The typing judgement for actions
(Figure~\ref{fig:action-typing-rules}) is of the form $\Gamma; \Theta
\vdash a : \bar\tau \to \Theta$, meaning that $a$ has type $\bar\tau
\to \Theta$ in variable context $\Gamma$ and header context $\Theta$.
Given a variable context $\Gamma$ and header type $\Theta$, an action
$\lambda \bar{x}.~c$ encodes a function of type
$\bar{\tau}\rightarrow\Theta'$, so long as the body $c$ is well-typed
in the context where $\Gamma$ is extended with $x_i : \tau_i$ for
every $i$.

\begin{figure}[t]
  \begin{mathpar}
        \inferrule[T-Const]{
            \mathtt{typeof}(k) = \bar{\tau} \rightarrow \tau'\\
            \Gamma ; \Theta \vdash e_i: \tau_i
        }{
            \Gamma ; \Theta \vdash k(\bar{e}): \tau'
        }
        \and
        \inferrule[T-Var]{
            x: \tau \in \Gamma
        }{
            \Gamma ; \Theta  \vdash x : \tau
        }
        \and
        \inferrule[T-Field]{
            \texttt{Includes}~\Theta~h \\
            \mathcal{F}(h, f) = \tau
        }{
            \Gamma ; \Theta \vdash h.f : \tau
        }
    \end{mathpar}
    \caption{Expression typing rules for \name}
    \label{fig:expression-typing-rules}
\end{figure}
The typing rules for expressions are shown in
Figure~\ref{fig:expression-typing-rules}.
Constants are typechecked according to rule \textsc{T-Constant}, as
long as each expression that is passed as an argument to the constant
$k$ has the type required by the \texttt{typeof} function.
The rule \textsc{T-Var} is standard.

\subsection{Operational Semantics}

We now present the small-step operational semantics of \name. We
define the operational semantics for commands in terms of four-tuples
$\langle I,O,H,c\rangle$, where $I$ is the input bit stream (which is
assumed to be infinite for simplicity), $O$ is the output bit stream, $H$ is a map
that associates each valid header instance with a records containing
the values of each field, and $c$ is the command to be evaluated. The
reduction rules are presented in Figure \ref{fig:semantics}.

\begin{figure}[t]
    \begin{mathpar}
        \inferrule[E-Extr]{
            \mathcal{HT}(h) = \eta \\
            \mathit{deserialize_\eta}(I) = (v, I')
        }{
            \langle I, O, H, \syntax{extract(h)}\rangle \rightarrow \langle I', O, H[h \mapsto v], \syntax{skip}\rangle
        }
        \and
        \inferrule[E-Emit]{
            \mathcal{HT}(h) = \eta \\
            \mathit{serialize_\eta}(H(h)) = \bar{B}
        }{
            \langle I, O, H, \syntax{emit(h)}\rangle \rightarrow \langle I, O.\bar{B}, H, \syntax{skip}\rangle
        }
        \and
        \inferrule[E-EmitInvalid]{
            h \not\in \mathit{dom(H)}
        }{
            \langle I, O, H, \syntax{emit(h)}\rangle \rightarrow \langle I, O, H, \syntax{skip}\rangle
        }
        \and
        \inferrule[E-IfValidTrue]{
            h \in \mathit{dom(H)}
        }{
          \langle I, O, H, \syntax{valid(h)~c_1~else~c_2} \rangle \rightarrow \langle I, O, H, c_1 \rangle
        }
        \and
        \inferrule[E-IfValidFalse]{
            h \not \in \mathit{dom(H)}
        }{
            \langle I, O, H, \syntax{valid(h)~c_1~else~c_2} \rangle \rightarrow \langle I, O, H, c_2 \rangle
        }
        \and
        \inferrule[E-Mod]{
          H(h) = r \\
          r' = \{r\ with\ f = v\}
        }{
            \langle I, O, H, \syntax{h.f=v}\rangle \rightarrow \langle I, O, H[h \mapsto r'], \syntax{skip}\rangle
        }
        \and
        \inferrule[E-Apply]{
            \mathcal{CA}(t, H) = (a_i, \bar{v})\\
            \mathcal{A}(a_i) = \lambda \bar{x}.c_i
        }{
            \langle I,O,H, t.apply()\rangle \rightarrow \langle I,O,H,c_i[\bar{v}/\bar{x}]\rangle
        }
        \and
        \inferrule[E-Add]{
            \mathcal{HT}(h) = \eta\\
            \mathit{init}_\eta = v
        }{
            \langle I, O, H, add(h)\rangle \rightarrow \langle I, O, H[h \mapsto v], \syntax{skip}\rangle
        }
        \and
        \inferrule[E-AddValid]{
            h \in \mathit{dom(H)}
        }{
            \langle I, O, H, add(h)\rangle \rightarrow \langle I, O, H, \syntax{skip}\rangle
        }
        \and
        \inferrule[E-Rem]{
          \
        }{
            \langle I, O, H, \syntax{remove(h)}\rangle \rightarrow \langle I, O, H \setminus h, \syntax{skip}\rangle
        }
    \end{mathpar}
\caption{Selected rules of the operational semantics of \name; the
  elided rules are standard and can be found in
  \iftoggle{techreport}{Appendix~\ref{sec:appendix-semantics}.}{the technical report.}}
    \label{fig:semantics}
\end{figure}

\begin{figure}[t]
    \begin{mathpar}
                                                        \inferrule[E-Const]{
            \llbracket k \rrbracket (v_1,...,v_n) = v
        }{
            \langle H,k(v_1,...,v_n)\rangle \rightarrow v
        }
        \and
        \inferrule[E-Field]{
            H(h)=\{f_1:n_1,...,f_k:n_k\}
        }{
            \langle H, h.f_i\rangle \rightarrow n_i
        }
    \end{mathpar}
    \caption{Selected rules of the operational semantics for expressions.}
    \label{fig:expr-semantics}
\end{figure}
 
The command $\syntax{extract(h)}$ evaluates via the rule
\textsc{E-Extr}, which looks up the header type in $\mathcal{HT}$ and
then invokes corresponding deserialization function.
The deserialized header value $v$ is added to to the map of valid
header instances, $H$.
For example, assuming the header type
$\eta = \{ f:\mathit{bit\langle3\rangle};\
g:\mathit{bit\langle2\rangle};\}$ has two fields $f$ and $g$ and
$I=11000B$ where $B$ is the rest of the bit stream following, then
$\mathit{deserialize_\eta(I) = (\{f=110;\ g=00;\}, B)}$.

The rule \textsc{E-Emit} serializes a header instance $h$ back into a bit stream.
It first looks up the corresponding header type and header value in
the header table $\mathcal{HT}$ and the map of valid headers
respectively.
The header value is then passed to the serialization function for the
header type to produce a bit sequence that is appended to the output
bit stream.
Similarly, we assume that a serialization function is defined for every
header type, which takes the bit values of the fields of a header
value and concatenates them to produce a single bit sequence.
We adopt the semantics of P4 with respect to emitting invalid headers.
Emitting an invalid header instance---i.e., a header instance which
has not been added or extracted---has no effect on the output bit
stream (rule \textsc{E-EmitInvalid}).
Notice also that the header remains unchanged in $H$.

Sequential composition reduces left to right, i.e., the left command
needs to be reduced to \syntax{skip} before the right command can be
reduced (rule \textsc{E-Seq}).
The evaluation of conditionals (rules \textsc{E-If, E-IfTrue,
  E-IfFalse}) is standard.
Both \textsc{E-Seq}, \textsc{E-If}, \textsc{E-IfTrue} and
\textsc{E-IfFalse} are relegated to the
\iftoggle{techreport}{appendix}{technical report} for brevity.
The rules for validity checks (\textsc{E-IfValidTrue},
\textsc{E-IfValidFalse}) step to the true branch if $h \in
\mathit{dom}(H)$ and to the false branch otherwise.

Table application commands are evaluated according to rule
\textsc{E-Tapply}. We first invoke the control plane function
$\mathcal{CA}(t, H)$ to determine an action $a_i$ and action data $v$.
Then we use $\mathcal{A}$ to lookup the definition of $a_i$, yielding
$\lambda \bar{x}: \bar{\tau}.~c_i$ and step to $c_i[\bar{v}/\bar{x}]$.
Note that for simplicity, we model the evaluation of expressions read
by the table using the control-plane function $\mathcal{CA}$.

The rule \textsc{E-Add} evaluates addition commands $\syntax{add(h)}$.
Similar to header extraction, the $\mathit{init}_\eta()$ function
produces a header instance $v$ of type $\eta$ with all fields set to a
default value and extends the map $H$ with $h \mapsto v$. Note that
according to \textsc{E-Add-Exist}, if the header instance is already
valid, \syntax{add(h)} does nothing.
Finally, the rule \textsc{E-Rem} removes the header from the map $H$.
Again, if a header $h$ is already invalid, removing it has no effect.

The semantics for expressions is defined in Figure
\ref{fig:expr-semantics}, using tuples $\langle H,e \rangle$, where
$H$ is the same map used in the semantics of commands and $e$ is the
expression to evaluate.  The rule \textsc{E-Field} reduces header
field expressions to the value stored in the heap $H$ for the
respective field.
To evaluate constants via the rule \textsc{E-Const} (omitting the
obvious congruence rule), we assume that there is an evaluation
function for constants $\llbracket k \rrbracket (\bar{v}) = v$ that is
well-behaved---i.e., if
$\mathtt{typeof}(k)=\bar{\tau} \rightarrow \tau'$ and
$\overline{v:\tau}$, then
$.;. \vdash \llbracket k \rrbracket (\bar{v}):
\tau'$. We use these facts to prove progress and preservation. 
\subsection{Safety of \name}

We prove safety in terms of progress and preservation. Both theorems
make use of the relation $H \models \Theta$ which intuitively holds if 
$H$ is described by $\Theta$. The formal definition, as given in
Figure \ref{fig:entails-relation}, satisfies $H \models \Theta$ if and
only if $\mathit{dom}(H) \in \llbracket \Theta \rrbracket$.

\begin{figure}[t!]
  \begin{mathpar}
        \inferrule[Ent-Empty]{\ }{
            \boldsymbol{\cdot} \models 1
        }
        \and
        \inferrule[Ent-Inst]{
            \mathit{dom}(H) = \{h\}
        }{
            H \models h
        }
        \and
        \inferrule[Ent-Seq]{
            H_1 \models \Theta_1 \\\\
            H_2 \models \Theta_2
        }{
            H_1 \cup H_2 \models \Theta_1 \cdot \Theta_2
        }
        \and
        \inferrule[Ent-ChoiceL]{
            H \models \Theta_1 \\
        }{
            H \models \Theta_1 + \Theta_2
        }
        \and
        \inferrule[Ent-ChoiceR]{
            H \models \Theta_2 \\
        }{
            H \models \Theta_1 + \Theta_2
        }
    \end{mathpar}
\caption{The \textit{Entailment} relation between header instances and header instance types}
\label{fig:entails-relation}
\end{figure}
 
We prove type safety via progress and preservation theorems. The
respective proofs are mostly straightforward for our system---we
highlight the unusual and nontrivial cases below an relegate the full
proofs to the \iftoggle{techreport}{appendix}{technical report}.

\begin{theorem}[Progress]
If $\cmdtype{\cdot}{\Theta}{c}{\Theta'}$ and $H \models \Theta$, then either,
\vspace*{-.5em}
\begin{itemize}
\item $c=\mathit{skip}$, or 
\item $\exists \langle I',O',H',c' \rangle.~\langle I,O,H,c \rangle \rightarrow \langle I',O',H',c' \rangle$.
\end{itemize}
\end{theorem}

Intuitively, progress says that a well-typed command is fully reduced
or can take a step.

\begin{theorem}[Preservation]
If $\cmdtype{\Gamma}{\Theta_1}{c}{\Theta_2}$ and 
$\langle I,O,H,c \rangle \rightarrow \langle I',O',H',c'\rangle$, where 
$H \models \Theta_1$, then 
$\exists \Theta_1', \Theta_2'.~\cmdtype{\Gamma}{\Theta_1'}{c}{\Theta_2'}$ where $H' \models \Theta_1'$ and $\Theta_2' < \Theta_2$.
\end{theorem}

More interestingly, preservation says that if a command $c$ is
well-typed with input type $\Theta_1$ and output type $\Theta_2$, and
$c$ evaluates to $c'$ in a single step, then there exists an input
type $\Theta_1'$ and an output type $\Theta_2'$ that make $c'$
well-typed. To make the inductive proof go through, we also need to
prove that $\Theta_1'$ describes the same maps of header instance $H$
as $\Theta_1$, and $\Theta_2'$ is semantically contained in
$\Theta_2$. We define syntactic containment to be
$\Theta_1 < \Theta_2 \triangleq \llbracket \Theta_1 \rrbracket
\subseteq \llbracket \Theta_2 \rrbracket$. (These conditions are
somewhat reminiscent of conditions found in languages with subtyping.)

\begin{proof}
By induction on a derivation of $\cmdtype{\Gamma}{\Theta_1}{c}{\Theta_2}$, with a case analysis on the last rule used. 
We focus on two of the most interesting cases. 
See \iftoggle{techreport}{Appendix~\ref{app:safety}}{technical report} for the full proof.
 
\begin{description}
\item{\textit{Case} \textsc{T-IfValid}:}
$c = \syntax{valid(h)~c_1~else~c_2}$ and $\cmdtype{\Gamma}{\mathtt{Restrict}~\Theta_1~h}{c_1}{\Theta_{12}}$ and $\cmdtype{\Gamma}{\mathtt{NegRestrict}~\Theta_1~h}{c_2}{\Theta_{22}}$ and $\Theta_2 = \Theta_{12} + \Theta_{22}$. \\[.5em]
There are two evaluation rules that apply to $c$, \textsc{E-IfValidTrue} and \textsc{E-IfValidFalse}
\begin{description}
\item{\textbf{Subcase} \textsc{E-IfValidTrue}:} $c' = c_1$ and
  $h \in \mathit{dom}(H)$ and $H' = H$.\\[.5em] Let
  $\Theta_1' = \mathtt{Restrict}~\Theta_1~h$ and
  $\Theta_2' = \Theta_{12}$. We have
  $\cmdtype{\Gamma}{\Theta_1'}{c'}{\Theta_2'}$ by assumption, we have
  $H \models \Theta_1'$ by
  \iftoggle{techreport}{Lemma~\ref{lem:restrict-domain-entail}}{a
      lemma formalizing the relationship between \textsc{Restrict}
      and ($\models$) (see tech report)}, and we have
  $\Theta_2' < \Theta_2$ by the definition of $<$ and the semantics of
  union.
\item{\textbf{Subcase} \textsc{E-IfValidFalse}:} $c' = c_2$ and $h \not\in \mathit{dom}(H)$ and $H' = H$.\\[.5em]
 Symmetric to the previous case. 
\end{description}

\item{\textit{Case} \textsc{T-Apply}:}
$c = \syntax{t.apply()}$ and $\mathcal{CV}(t) = (\bar{S}, \bar{e})$ and $t.\mathit{actions} = \bar{a}$ and 
$\cdot; \Theta \vdash e_j : \tau_j$ for $e_j \in \bar{e}$ and 
$\mathtt{Restrict}~\Theta_1~S_i \vdash a_i : \bar{\tau_i} \rightarrow \Theta_i'$ for $a_i \in \bar{a}$ and $\Theta_2 = \sum\left( \Theta_i' \right)$\\[.5em]
Only one evaluation rule applies to $c$, \textsc{E-Apply}. It follows that
$\mathcal{CA}(t, H) = (a_i, \bar{v})$, and $c' = c_i[\bar{v}/\bar{x}]$
where $\mathcal{A}(a_i) = \lambda \bar{x}.~c_i$. By inverting
\textsc{T-Action}, we have $\cmdtype{\Gamma, \bar{x} : \bar{\tau}_i;}{\texttt{Restrict}~\Theta~S_i}{c_i}{\Theta_i'}$. 
By \iftoggle{techreport}{Proposition~\ref{prop:cp-action-data}}{control plane assumption (2)}, we have $\cdot; \cdot \vdash \bar{v} : \bar{\tau}_i$. 
By the substitution lemma, we have $\cmdtype{\Gamma}{\texttt{Restrict}~\Theta~S_i}{c_i[\bar{v}/\bar{x}]}{\Theta_i'}$.
Let $\Theta_1' = \mathtt{Restrict}~\Theta~S_i$ and $\Theta_2' = \Theta_i'$. 
We have shown that $\cmdtype{\Gamma}{\Theta_1'}{c'}{\Theta_2'}$, 
we have that $H' \models \Theta_1'$ by \iftoggle{techreport}{Proposition~\ref{prop:cp-valid}}{control plane assumption (3)}, and 
we have $\Theta_2' < \Theta_2$ by the definition of $<$ and the semantics of union types.\qedhere
\end{description}
\end{proof}
 
\section{Experience (Evaluation)}
\label{sec:evaluation}

We implemented our type system in a tool called \ourchecker that
automatically checks P4 programs and reports violations of the type
system presented in Figure~\ref{fig:command-typing-rules}. \ourchecker
uses the front-end of \texttt{p4v}~\cite{Liu:2018aa} and handles the
full \PPPP{14} language.\footnote{We also have an open-source
    prototype implementation for \PPPP{16} that handles the most
    common features of \PPPP{16}
    (\url{https://github.com/cornell-netlab/p4check}).} Our key findings,
which are reported in detail below, show (i) that our type system
finds bugs ``in the wild'' and (ii) that the programmer effort needed
to repair programs to pass our type checker is modest.

\begin{figure}[t]
  \centering
    \includegraphics[width=\textwidth]{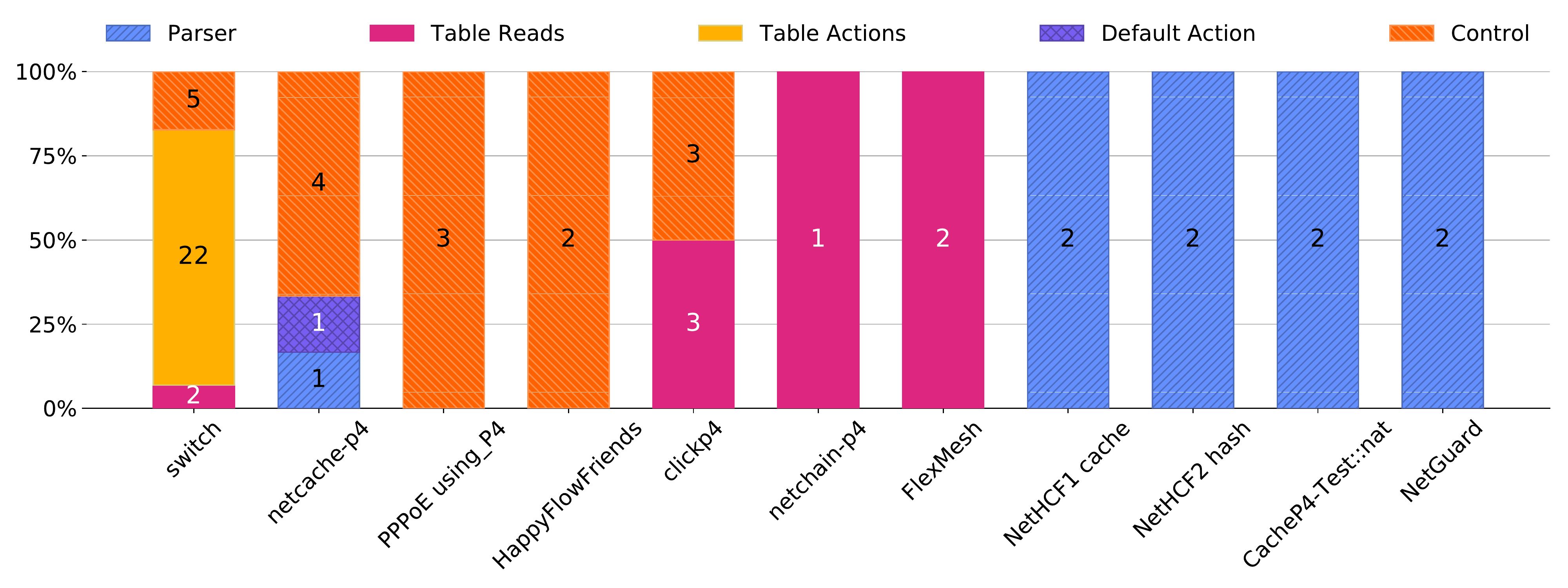}
  \caption{Proportional frequencies of each bug type per-program. The
    raw number of bugs for each program and category is reported at
    the top of each stacked bar. }
  \label{fig:p4check-class}
\end{figure}

\begin{figure}[t]
  \centering
  \includegraphics[width=\textwidth]{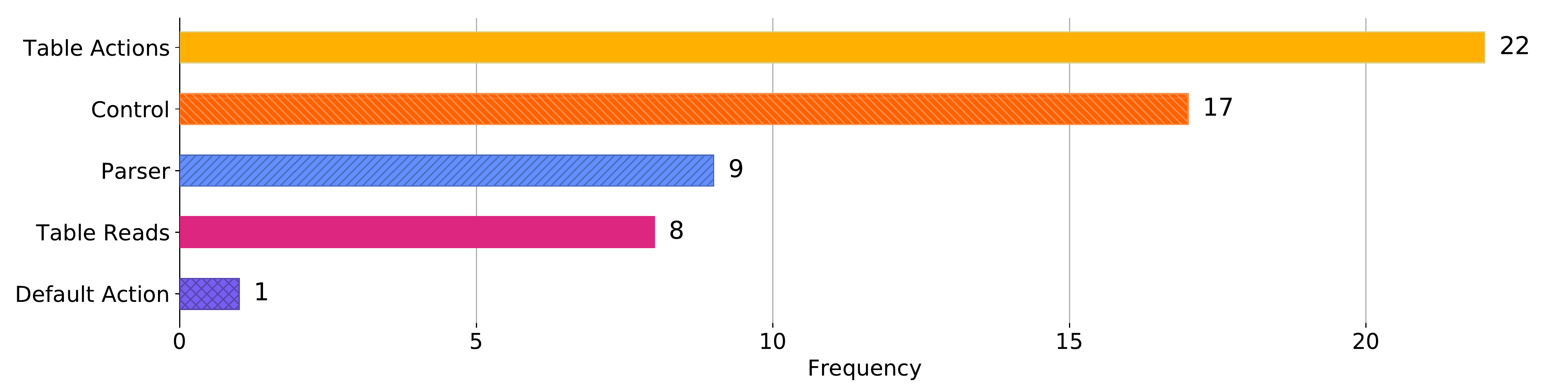}
  \caption{Frequency of each bug across all programs. The raw number
    of bugs in each category is reported to the right of the bar}
  \label{fig:p4check-freq}
\end{figure}
      
\subsection{Overview of Bugs in the Wild}

We ran \ourchecker on 15 open source $\PPPP{14}$ programs\footnote{We
  chose to check \PPPP{14} instead of \PPPP{16}, since there are
  currently more \PPPP{14} programs available on GitHub.} of varying
sizes and complexity, ranging from 143 to 9060 lines of code.
Our criteria for selecting programs was: (1) each program had to be
open source, (2) available on GitHub, and (3) compile without errors,
(4) and be written either by industrial teams developing production
code or by researchers implementing standard or novel network
functionality in P4---i.e., we excluded programs primarily used for
teaching.
Out of the 15 subject programs only 4 passed our type checker, all of
which were simple implementations of routers or DDoS mitigation that
accepted only a small number of packet types and were relatively small
(188--635 lines of code).
For the remaining 11 programs (industrial and research) our checker
found $418$ type checking violations overall.

Frequently, multiple violations produced by \ourchecker have the same
root cause.
For example, if a single action \texttt{rewrite\_ipv4} that rewrites
fields \texttt{srcAddr} and \texttt{dstAddr} for an \texttt{ipv4}
header is called in a context that cannot prove that \texttt{ipv4} is
valid, then both references to \texttt{ipv4.srcAddr} and
\texttt{ipv4.dstAddr} will be reported as violations, even though they
are due to the same \emph{control} bug
(Section~\ref{sec:unsafe-reference})---namely that
\texttt{rewrite\_ipv4} was not called in a context that could prove
the validity of \texttt{ipv4}.
To address this issue, we applied another metric to quantify the
number of bugs (inspired by the method proposed by
others~\cite{Klees:2018}): we equate the number of bugs in each
program with the number of bug \emph{fixes} required to make the
program in question pass our type checker. Using this metric, we
counted $58$ bugs.

We classified the bugs according to the classes described in
Section~\ref{sec:taxonomy}.
Figure~\ref{fig:p4check-class} depicts the per-program breakdown of
the frequency of each bug class, and Figure~\ref{fig:p4check-freq}
depicts the overall frequency of each bug.
Notice that even though table action bugs were the most frequent bug
(with 22 occurrences), they were only found in a single program
(\texttt{switch.p4}). These bugs are especially prevalent in this program
because of its heavy reliance on correct control-plane configuration.
Conversely, there were 9 occurrences across 5 programs for both parser
bugs and table reads bugs.

Readers familiar with previous work on \texttt{p4v}~\cite{Liu:2018aa},
a recent P4 verification tool, may notice that we detected no default
action bugs for the \texttt{switch.p4} program, while \texttt{p4v}
reported many!
The reasons for this are two-fold.
First, \texttt{p4v} allows programmers to verify complex properties,
which means that it can express fine-grained conditions on tables and
relationships between them.
In contrast, we make heuristic assumptions about P4 programs that
automatically eliminate many bugs, including some default action bugs.
Second, our repairs are often coarse-grained and may enforce a
stronger guarantee on the program than may be necessary; using
first-order logic annotations, \texttt{p4v} programmers manually
specify the weakest (and hence more complex) assumptions.

We make no claims about the completeness of our taxonomy. For example,
we found one instance, in the \textsc{HappyFlowFriends} program, where
the programmer had mistakenly instantiated metadata $m$ as a header,
and consequently did not parse $m$ (since metadata is always valid)
causing $m$ to (ironically) always be invalid.

\subsection{\ourchecker in Action}
\label{sec:bugfixes}

We reprise the canonical examples of each class of bugs from
Section~\ref{sec:taxonomy}, describing how \ourchecker detects them
and discussing ways to fix them.

\subsubsection{Parser Bugfixes}

Recall Figure~\ref{fig:NetHCF-parser}, which exhibits the parser bug.
The bug occurs because the parser, which extracts IPv4-TCP packets,
boots unexpected packets (such as IPv6 or UDP packets) directly to
\texttt{ingress}, which then assumes that both the \texttt{ipv4} and
\texttt{tcp} headers are valid, even though the parser does not
guarantee this fact.

\begin{figure}
  \begin{minipage}{\textwidth}
\begin{lstlisting}[style=fixup]
./h.p4, line 350, cols 12-21: +error+ |tcp| not guaranteed to be valid
./h.p4, line 118, cols 8-16: +error+ |ipv4| not guaranteed to be valid
./h.p4, line 101, cols 42-50: +error+ |ipv4| not guaranteed to be valid
./h.p4, line 320, cols 8-15: +error+ |tcp| not guaranteed to be valid
./h.p4, line 362, cols 12-19:+error+ |tcp| not guaranteed to be valid
./h.p4, line 362, cols 29-36: +error+ |tcp| not guaranteed to be valid
./h.p4, line 295, cols 60-69: +error+ |tcp| not guaranteed to be valid
./h.p4, line 107, cols 8-16: +error+ |ipv4| not guaranteed to be valid
./h.p4, line 101, cols 42-50: +error+ |ipv4| not guaranteed to be valid
./h.p4, line 163, cols 8-16: +error+ |ipv4| not guaranteed to be valid
./h.p4, line 101, cols 42-50: +error+ |ipv4| not guaranteed to be valid
\end{lstlisting}
  \end{minipage}
  \hrule
  \begin{minipage}{\textwidth}
\begin{lstlisting}[style=fixup]
./h.p4, line 350, cols 12-21: +error+ |tcp| not guaranteed to be valid
./h.p4, line 320, cols 8-15: +error+ |tcp| not guaranteed to be valid
./h.p4, line 362, cols 12-19: +error+ |tcp| not guaranteed to be valid
./h.p4, line 362, cols 29-36: +error+ |tcp| not guaranteed to be valid
./h.p4, line 295, cols 60-69: +error+ |tcp| not guaranteed to be valid
\end{lstlisting}
  \end{minipage}
  \caption{Curated output from \ourchecker for the parser bug in \textsc{NetHCF} before (above) and after (below) modifying \texttt{parse\_ethernet}}   \label{fig:p4check-output}
\end{figure}

In terms of our type system, the parser produces packets of type
$\texttt{ethernet} \cdot (1 + \texttt{ipv4} \cdot (1 +
\texttt{tcp}))$; however the control only handles packets of type
$\texttt{ethernet} \cdot \texttt{ipv4} \cdot \texttt{tcp}$. Hence,
when typecheck this example, \ourchecker reports every reference to
\texttt{tcp} and \texttt{ipv4} in the whole program as a violation of
the type system.
As shown in the top half of Figure~\ref{fig:p4check-output}, we get an error message at every
reference to \texttt{ipv4} or \texttt{tcp}.
The ubiquity of the reports intimates a mismatch between the parsing
and the control types, which gives the programer a hint as how to fix
the problem.

When we modify the \texttt{default} clause in
\texttt{parse\_ethernet}, as in Figure~\ref{fig:NetHCF-parser}, and
run our tool again, all of the \texttt{ipv4} violations are removed
from the output, as shown in the bottom half of
Figure~\ref{fig:p4check-output}.
Then fixing the \texttt{parse\_ipv4} parser, as in
Figure~\ref{fig:NetHCF-parser}, causes our tool to output no
violations. In particular, the type upon entering the \texttt{ingress}
control function is
$\texttt{ethernet}\cdot\texttt{ipv4}\cdot\texttt{tcp}$, so all
subsequent references to \texttt{ipv4} and \texttt{tcp} are safe.

\subsubsection{Control Bugfixes}

Recall that a control bug occurs when the incoming type presents a
choice between two instances that are not handled by subsequent code.
The program shown in Figure~\ref{fig:netcache-p4-unsafety}
uses a parser that produces the type $\Theta = \texttt{ethernet} \cdot
(1 + \texttt{ipv4} \cdot (1 + \texttt{udp} \cdot(1 + \texttt{nc\_hdr}
\cdot \tau) + \texttt{tcp}))$, where $\tau$ is a type for caching operations.
Note that $\texttt{Includes}~\Theta~\texttt{nc\_hdr}$ does not hold.
However,  \texttt{process\_cache} and \texttt{process\_value} only type
check in contexts where $\texttt{Includes}~\Theta~\texttt{nc\_hdr}$ is
true. \ourchecker reports type violations at every reference to
$\texttt{nc\_hdr}$. Fixing this error is simply a matter of wrapping
the \texttt{process\_cache()} call in a validity check as demonstrated
in Figure~\ref{fig:netcache-p4-unsafety}.
As \textsc{NetCache} handles TCP and UDP packets as well as its
special-purpose packets, we simply apply the IPv4
routing table if the validity check for \texttt{nc\_hdr}
fails.\footnote{Astute readers may detect a parser bug in this
  example. Hint, the \texttt{ipv4\_route} table requires
  $\texttt{Includes}~\Theta~\texttt{ipv4}$ where $\Theta$ is type
  where it is applied.}

\subsubsection{Table Reads Bugfixes}

\begin{figure}[t]
  \begin{minipage}{\textwidth}
\begin{lstlisting}[style=fixup]
port.p4, line 248, cols 8-24: !warning!: assuming either |vlan_tag_[0]| matched as valid or |vlan_tag_[0].vid| wildcarded
  
port.p4, line 250, cols 8-24: !warning!: assuming either |vlan_tag_[1]| matched as valid or |vlan_tag_[1].vid| wildcarded
\end{lstlisting}      
  \end{minipage}
  \hrule
  \begin{minipage}{\textwidth}
\begin{lstlisting}[style=fixup]
fabric.p4 line 42, cols 41-67: !warning!: assuming |fabric_header_cpu| matched as valid for rules with action |terminate_cpu_packet|

fabric.p4, line 57, cols 17-54: !warning!: assuming |fabric_header_unicast| matched as valid for rules with action |terminate_fabric_unicast_packet|

fabric.p4, line 81, cols 17-56: !warning!: assuming |fabric_header_multicast| matched as valid for rules with action |terminate_fabric_multicast_packet|
\end{lstlisting}      
  \end{minipage}
  
  \caption{Warnings printed after fixing \texttt{switch.p4}'s reads bug
    (top), and its actions bug (bottom)}
  \label{fig:warnings}
\end{figure}

Table reads errors, as shown in Figure~\ref{fig:switch-reads}, occur
when a header $h$ is included in the \texttt{reads} declaration of a
table $t$ with match kind $k$, and $h$ is not guaranteed to be valid
at the call site of $t$, and if $h \not \in \texttt{valid\_reads}(t)$
or the match-kind of $k \neq \texttt{ternary}$.

In the case of the \texttt{port\_vlan\_mapping} table in
Figure~\ref{fig:switch-reads}, there is a valid bit for both
\texttt{vlan\_tag\_[0]} and \texttt{vlan\_tag\_[1]}, both of which are
followed by \texttt{exact} matches.
To solve this problem, we need to use the \texttt{ternary} match-kind
instead, which allows the use of wildcard matching.
When a field is matched with a wildcard, the table does not attempt to
compute the value of the \texttt{reads} expression.

This fix assumes that the controller is well behaved and fills the
\texttt{vlan\_tag\_[0].vid} with a wildcard whenever
\texttt{vlan\_tag\_[0]} is matched as invalid (and similarly for
\texttt{vlan\_tag\_[1]}). This also what the \name type system does,
with its $\mathsf{maskable}$ checks in the \textsc{T-Apply} rule
\ourchecker prints warnings describing these assumptions to the
programmer (top of Figure~\ref{fig:warnings}), giving them properties
against which to check their control plane implementation.

\subsubsection{Table Action Bugfixes}

Table actions bugs occur when at least one action cannot be safely
executed in all scenarios. For example, the table
\texttt{fabric\_ingress\_dst\_lkp} shown in
Figure~\ref{fig:switch-buggy} has a table action bug, which can be
fixed by modifying the table's \texttt{reads} declaration.
Recall that the parser will parse exactly one of the headers
\texttt{fabric\_hdr\_cpu}, \texttt{fabric\_hdr\_unicast} and
\texttt{fabric\_hdr\_multicast}, which means that when the table is
applied at type $\Theta$, exactly one of
$\texttt{Includes}~\Theta~\texttt{fabric\_hdr\_i}$ for $i \in
\{\texttt{cpu}, \texttt{unicast}, \texttt{multicast}\}$ will hold.
Now, the action \texttt{term\_cpu\_packet} typechecks only with the
(nonempty) type $\texttt{Restrict}\; \Theta\;
\texttt{fabric\_hdr\_cpu}$, and the actions
$\texttt{term\_fabric\_}i\texttt{\_packet}$ only typecheck with the
(nonempty) types \\$\texttt{Restrict}\; \Theta\;
\texttt{term\_fabric\_}i\texttt{\_packet}$ for $i = \texttt{unicast},
\texttt{multicast}$.
\ourchecker suggests that this is the cause of the bug since it
reports type violations for all of the references to these three
headers in the control paths following from the application of
\texttt{fabric\_ingress\_dst\_lkp}.

The optimal\footnote{Another fix would be to refactor the single
  into multiple tables, each guarded by a separate validity
  check. However, combining this kind of logic in a single table
  helps conserve memory, so in striving to change the
  behavior of the program as little as possible, we propose modifying
  the table reads.} fix here is to augment the \texttt{reads}
declaration to include a validity check for each contentious header.
We then assume that the controller is well-behaved enough to only call
actions when their required headers are valid, allowing us to
typecheck each action in the appropriate type restriction.
\ourchecker alerts the programmer whenever it makes such an assumption.
We show these warnings for the fixed version of
\texttt{fabric\_ingress\_dst\_lkp} below the line in
Figure~\ref{fig:warnings}.

\subsubsection{Default Action Bugfix}

Default action bugs occur when a programmer creates a wrapper table
for an action that modifies the type, and forgets to force the table
to call that action when the packet misses.
The \texttt{add\_value\_header\_1} table from
Figure~\ref{fig:default-action} wraps the action
\texttt{add\_value\_header\_1\_act}, which calls the single line
\texttt{add\_header(nc\_value\_1)}.

The default action, when left unspecified, is \texttt{nop}, which
means that if the pre-application type was $\Theta$, then the
post-application type is
\mbox{$\Theta + \Theta \cdot \texttt{nc\_value\_1}$}, which does not
include $\texttt{nc\_value\_1}$.  Hence, \ourchecker reports every
subsequent reference (on this code path) to \texttt{nc\_header\_1} to
be a type violation.

To fix this bug, we need to set the default action to
$\texttt{add\_value\_1}$---this makes the post-application type
$\Theta\cdot\texttt{nc\_value\_1} + \Theta\cdot\texttt{nc\_value\_1} =
\Theta\cdot\texttt{nc\_value\_1}$, which includes
$\texttt{nc\_value\_1}$, thus allowing the subsequent code to
typecheck.

\subsection{Overhead}

It is important to evaluate two kinds of overhead when considering a
static type system: overhead on programmers and on the underlying
implementation.

Typically, adding a static type system to a dynamic type system
requires more work for the programmer---the field of gradual typing is
devoted breaking the gargantuan task into smaller commit-sized
chunks~\cite{Campora:2017}.
Surprisingly, in our experience, migrating real-world P4 code to pass
the \name type system only required modest programmer effort.

To qualitatively evaluate the effort required to change an unsafe
program into a safe one using our type system, we manually fixed
all of the detected bugs.
The programs that had bugs required us to edit between $0.10\%$ and
$1.4\%$ of the lines of code.
The one exception was \textsc{PPPoE\_using\_P4}, which was a 143 line
program that required 6 line-edits (4\%), all of which were validity
checks.
Conversely, \texttt{switch.p4} required 34 line edits, the greatest
observed number, but this only accounted for $0.37\%$ of the total
lines of code in the program.

Each class of bugs has a simple one-to-two line fix, as
described in Section~\ref{sec:bugfixes}: adding a validity check,
adding a default action, or slightly modifying the parser.
Each of these changes was straightforward to identify and simple to
make.

Another possible concern is that that extending tables with extra read
expressions, or adding run-time validity checks to controls, might
impose a heavy cost on implementations, especially on hardware.
Although we have not yet performed an extensive study of the impact on
compiled code, based on the size and complexity of the annotations we
added, we believe the additional cost should be quite low. We were
able to compile our fixed version of the \texttt{switch.p4} program to
the Tofino architecture~\cite{Tofino} with only a modest increase in
resource usage. Overall, given the large number of potential bugs
located by \ourchecker, we believe the assurance one gains about
safety properties by using a static type system makes the costs well
worth it.
 
\section{Related Work}
\label{sec:related}

Probably the most closely related work to \name is
\texttt{p4v}~\cite{Liu:2018aa}. Unlike \name, which is based on a
static type system, \texttt{p4v} uses Dijkstra's approach to program
verification based on predicate transformer semantics. To model the
behavior of the control plane, \texttt{p4v} uses first-order
annotations. \name's typing rule for table application is inspired by
this idea, but adopts simple heuristics---e.g., we only assume
  that the control plane is well-behaved---rather than requiring
logical annotations.

Both \texttt{p4v} and \textsc{P4Check} can be used to verify safety
properties of data planes modelled in P4---e.g., that no read or write
operations are possible on an invalid header. As it is often the case
when comparing approaches based on types to those based on program
verification, \texttt{p4v} can check more complex properties,
including architectural invariants and program-specific
properties---e.g., that the IPv4 time-to-live field is correctly
decremented on every packet. However, in general, it requires
annotating the program with formal specifications both for the
correctness property itself and to model the behavior of the control
plane.

McKeown et al. developed an operational semantics for
P4~\cite{McKeown:2016aa}, which is translated to Datalog to verify
safety properties and to check program equivalence. An operational
semantics for P4 was also developed in the K
framework~\cite{Rosu:2010aa}, yielding a symbolic model checker and
deductive verification tool~\cite{Kheradmand:2018aa}.
Vera~\cite{Stoenescu:2018} models the semantics of P4 by translation
to SymNet~\cite{Stoenescu:2016}, and develops a symbolic execution
engine for verifying a variety of properties, including header
validity.

Compared to \name, these approaches do not use their formalization of
P4 as a foundation for defining a type system that addresses common
bugs. To the best of our knowledge, \name is the first formal calculus
for a P4-like packet processing language that provides
correct-by-construction guarantees of header safety properties.

Other languages have used type systems to rule our safety problems due
to null references. For example, NullAway~\cite{Sridharan:2018aa}
analyzes all Java programs annotated with \texttt{@Nullable}
annotations, making path-sensitive deductions about which references
may be null. Similar to the validity checks in \name, NullAway
analyses conditionals for null checks of the form \texttt{var != null}
using data flow analysis.

Looking further afield, PacLang~\cite{Ennals:2004aa} is a concurrent
packet-processing language that uses a linear type system to allow
multiple references to a given packet within a single thread. PacLang
and \name share the use of a type system for verifying safety
properties but they differ in the kind of properties they address and,
hence, the kind of type system they employ for this purpose. In
addition, the primary focus in PacLang is on efficient compilation
whereas \name is concerned with ensuring safety of header data.

Domino~\cite{Sivaraman:2016aa} is a domain-specific language for data
plane algorithms supporting {\it packet transactions}---i.e., blocks
of code that are guaranteed to be atomic and isolated from other
transactions. In Domino, the programmer defines the operations needed
for each packet without worrying about other in-flight packets. If it
succeeds, the compiler guarantees performance at the line rate
supported on programmable switches. Overall, Domino focuses on
transactional guarantees and concurrency rather than header safety
properties.

BPF+~\cite{Begel:1999aa} and eEBPF~\cite{Corbet:2014} are
packet-processing frameworks that can be used to extend the kernel
networking stack with custom functionality. The modern eBPF framework
is based on machine-level programming model, but it uses a virtual
machine and code verifier to ensure a variety of basic safety
properties. Much of the recent work on eBPF focuses on techniques such
as just-in-time compilation to achieve good performance.

SNAP~\cite{Arashloo:2016:SSN:2934872.2934892} is a language for
stateful packet processing based on P4. It offers a programming model
with global state registers that are distributed across many physical
switches while optimizing for various criteria, such as minimizing
congestion. More specifically, the compiler analyses read/write
dependencies to automatically optimize the placement of state and the
routing of traffic across the underlying physical topology.

While our approach to track validity is network-specific, is similar to taint analysis~\cite{Volpano:1996,
  Halfond:2006, Huang:2014}, which attempts to identify secure program
parts that can be safely accessed.

Of course, there is a long tradition of formal calculi that aim to
capture some aspect of computation and make it amenable for
mathematical reasoning. The design of \name is directly inspired by
Featherweight Java~\cite{Igarashi:2001aa}, which stands out for its
elegant formalization of a real-world language in an extensible core
calculus.
 
\section{Conclusion}
\label{sec:conc}

P4 provides a powerful programming model for network devices based on
high-level and declarative abstractions. Unfortunately, P4 lacks basic
safety guarantees, which often lead to a variety of bugs in practice.
This paper proposes \name, a domain-specific language for programmable
data planes that comes equipped with a formal semantics and a static
type system which ensures that every read or write to a header at
run-time will be safe. Under the hood, \name uses a rich set of
types that tracks dependencies beween headers, as well as a
path-sensitive analysis and domain-specific heuristics that model
common idioms for programming control planes and minimize false
positives. Our experiments using an OCaml prototype and a suite of
open-source programs found on GitHub show that most P4 applications
can be made safe with minimal programming effort. We hope that our
work can help lay the foundation for future enhancements to P4 as well
as the next generation of data plane languages. In the future, we plan
to explore enriching \name's type system to track additional
properties, investigate correct-by-construction techniques for writing
control-plane code, and develop a compiler for the language.

\bibliography{report}

\iftoggle{techreport}{
  
\appendix

\section{Additional Operational Semantics rules}
\label{sec:appendix-semantics}
This section presents additional evaluation rules.
\begin{mathpar}
        \inferrule[E-Seq]{
            \ 
        }{
            \langle I, O, H, \syntax{skip;c_2}\rangle \rightarrow \langle I, O, H, \syntax{c_2}\rangle
        }
        \and
        \inferrule[E-Seq1]{
            \langle I, O, H, \syntax{c_1}\rangle \rightarrow \langle I', O', H', \syntax{c_1'}\rangle 
        }{
            \langle I, O, H, \syntax{c_1;c_2}\rangle \rightarrow \langle I', O', H', \syntax{c_1';c_2}\rangle
        }
        \and
    \inferrule[E-If]{
        \langle H, e \rangle \rightarrow e'
    }{
        \langle I, O, H, \syntax{if\ (e)\ then\ c_1\ else\ c_2}\rangle \rightarrow \langle I, O, H, \syntax{if\ (e')\ then\ c_1\ else\ c_2}\rangle
    }
    \and
    \inferrule[E-IfTrue]{\ }{
        \langle I, O, H, \syntax{if\ (true)\ then\ c_1\ else\ c_2}\rangle \rightarrow \langle I, O, H, c_1\rangle
    }
    \and
    \inferrule[E-IfFalse]{\ }{
        \langle I, O, H, \syntax{if\ (false)\ then\ c_1\ else\ c_2}\rangle \rightarrow \langle I, O, H, c_2\rangle
    }
    \and
    \inferrule[E-Mod1]{
        \langle H, e \rangle \rightarrow e'
      }{
          \langle I, O, H, \syntax{h.f=e}\rangle \rightarrow \langle I, O, H, \syntax{h.f=e'}\rangle
      }
\end{mathpar}
 \section{Operations on header types}
\label{sec:appendix-header-ops}
This section presents an in-depth treatment of the operations defined on header instance types. 
In the following we assume that $S$ ranges over elements of the domain $\mathcal{P}(\mathcal{P}(H))$.
\subparagraph{Restriction}
The restrict operator $\texttt{Restrict}~\Theta~h$ recursively traverses $\Theta$ and keeps only those choices in which $h$ is contained, zeroing out the others.
Semantically this has the effect of throwing out the subsets of $\llbracket \Theta \rrbracket$ that do not contain $h$, i.e., we define restriction semantically as $S|_h \triangleq \{hs \mid hs \in S \land h \in hs\}$.
Syntactically we define restriction by induction on $\Theta$ as shown in Figure \ref{fig:restrict}.
\begin{figure}[ht]
    \boxedpage{\textwidth}{
    \begin{align*}
        \mathtt{Restrict}\ 0\ h &\triangleq 0 \\
        \mathtt{Restrict}\ 1\ h &\triangleq 0 \\
        \mathtt{Restrict}\ g\ h &\triangleq \begin{cases}
            g & \text{if }g=h\\
            0 & \text{otherwise}
        \end{cases}\\
        \mathtt{Restrict}\ (\Theta_1 \cdot \Theta_2)\ h &\triangleq ((\mathtt{Restrict}\ \Theta_1\ h) \cdot \Theta_2) + (\Theta_1 \cdot (\mathtt{Restrict}\ \Theta_2\ h)) \\
        \mathtt{Restrict}\ (\Theta_1 + \Theta_2)\ h &\triangleq (\mathtt{Restrict}\ \Theta_1\ h) + (\mathtt{Restrict}\ \Theta_2\ h)\\
    \end{align*}}
    \caption{Syntactic definition of the $\texttt{Restrict}~\Theta~h$ operator.}
    \label{fig:restrict}
\end{figure}
The equivalence of the syntactic and the semantic definition is captured by Lemma \ref{lem:restrict-equal}.
\begin{lemma}
\label{lem:restrict-equal}
$\llbracket \Theta \rrbracket |_h == \llbracket \mathtt{Restrict}\ \Theta\ h \rrbracket.$
\end{lemma}

\begin{proof}
  \iftoggle{fullproofs}{
    By induction on $\Theta$.\\
    \textit{Case} $\Theta=0$:
    \begin{align*}
        &\phantom{{}={}} \llbracket 0 \rrbracket | h \\
        &= \{\}|h & \text{by definition of } \llbracket.\rrbracket \\
        &= \{hs|hs \in \{\} \land h \in hs \} & \text{by definition of } .|h \\
        &= \{\} & \text{by set theory} \\
        &= \llbracket 0 \rrbracket & \text{by definition of } \llbracket . \rrbracket \\
        &= \llbracket \mathtt{Restrict}\ 0\ h \rrbracket & \text{by definition of } \mathtt{Restrict}\ .\ h 
    \end{align*}
    \textit{Case} $\Theta=1$:
    \begin{align*}
        &\phantom{{}={}} \llbracket 1 \rrbracket | h \\
        &= \{\{\}\}|h & \text{by definition of } \llbracket . \rrbracket \\
        &= \{hs|hs \in \{\{\}\} \land h \in hs \} & \text{by definition of } .|h \\
        &= \{\} & \text{by set theory} \\
        &= \llbracket 0 \rrbracket & \text{by definition of } \llbracket . \rrbracket \\
        &= \llbracket \mathtt{Restrict}\ 1\ h \rrbracket & \text{by definition of } \mathtt{Restrict}\ .\ h
    \end{align*}
    \textit{Case} $\Theta=g$:
    \begin{align*}
        &\phantom{{}={}} \llbracket g \rrbracket | h \\
        &= \{\{g\}\}|h & \text{by definition of } \llbracket . \rrbracket \\
        &= \{hs|hs \in \{\{g\}\} \land h \in hs \} & \text{by definition of } .|h \\
        &\phantom{{}={}} \textit{Subcase } h = g \\
        &= \{\{g\}\} & \text{by set theory} \\
        &= \llbracket g \rrbracket & \text{by definition of } \llbracket . \rrbracket \\
        &= \llbracket \mathtt{Restrict}\ g\ h \rrbracket & \text{by assumption } h=g \text{ and by definition of } \mathtt{Restrict}\ .\ h \\
        &\phantom{{}={}} \textit{Subcase } h \neq g \\
        &= \{\} & \text{by set theory} \\
        &= \llbracket 0 \rrbracket & \text{by definition of } \llbracket . \rrbracket \\
        &= \llbracket \mathtt{Restrict}\ g\ h \rrbracket & \text{by definition of } \mathtt{Restrict}\ .\ h \text{ and by assumption } h \neq g \\
    \end{align*}
    \textit{Case} $\Theta=\Theta_1 \cdot \Theta_2$: \\
    \begin{align*}
        &\phantom{{}={}} \llbracket \Theta_1 \cdot \Theta_2 \rrbracket | h \\
        &= \{ hs_1 \cup hs_2 | hs_1 \in \llbracket \Theta_1 \rrbracket \land hs_2 \in \llbracket \Theta_2 \llbracket \} | h & \text{by definition of } \llbracket . \rrbracket \\
        &= \{ hs_1 \cup hs_2 | hs_1 \in \llbracket \Theta_1 \rrbracket \land hs_2 \in \llbracket \Theta_2 \rrbracket \land h \in (hs_1 \cup hs_2)\} & \text{by definition of } .|h \\
        &= \{ hs_1 \cup hs_2 | hs_1 \in \llbracket \Theta_1 \rrbracket \land hs_2 \in \llbracket \Theta_2 \rrbracket \land h \in hs_1\}\ \cup & \text{by set theory}  \\
        &\phantom{{}={}} \{ hs_1 \cup hs_2 | hs_1 \in \llbracket \Theta_1 \rrbracket \land hs_2 \in \llbracket \Theta_2 \rrbracket \land h \in hs_2\} \\
        &= \{ hs_1 \cup hs_2 | (hs_1 \in \llbracket \Theta_1 \rrbracket \land h \in hs_1) \land hs_2 \in \llbracket \Theta_2 \rrbracket\}\ \cup & \text{by logic and set theory}  \\
        &\phantom{{}={}} \{ hs_1 \cup hs_2 | hs_1 \in \llbracket \Theta_1 \rrbracket \land (hs_2 \in \llbracket \Theta_2 \rrbracket \land h \in hs_2)\}\\
        &= \{hs_1 | hs_1 \in \llbracket \Theta_1 \rrbracket \land h \in hs_1\} \bullet \{hs_2 | hs_2 \in \llbracket \Theta_2 \rrbracket \}\ \cup & \text{ by definition of } S_1 \bullet S_2 \\
        &\phantom{{}={}} \{hs_1 | hs_1 \in \llbracket \Theta_1 \rrbracket \} \bullet \{hs_2 | hs_2 \in \llbracket \Theta_2 \rrbracket \land h \in hs_2\} \\
        &= \{hs_1 | hs_1 \in \llbracket \Theta_1 \rrbracket \}|h \bullet \{hs_2 | hs_2 \in \llbracket \Theta_2 \rrbracket \}\ \cup & \text{by definition of } .|h\\
        &\phantom{{}={}} \{hs_1 | hs_1 \in \llbracket \Theta_1 \rrbracket \} \bullet \{hs_2 | hs_2 \in \llbracket \Theta_2 \rrbracket\}| h \\
        &= \llbracket \Theta_1 \rrbracket | h \bullet \llbracket \Theta_2 \rrbracket \cup \llbracket \Theta_1 \rrbracket \bullet \llbracket \Theta_2 \rrbracket | h & \text{by definition of } \llbracket . \rrbracket \\
        &= \llbracket \mathtt{Restrict}\ \Theta_1\ h \rrbracket \bullet \llbracket \Theta_2 \rrbracket \cup \llbracket \Theta_1 \rrbracket \bullet \llbracket \mathtt{Restrict}\ \Theta_2\ h \rrbracket  & \text{by induction hypothesis}  \\
        &= \llbracket \mathtt{Restrict}\ \Theta_1\ h \cdot \Theta_2 + \Theta_1 \cdot \mathtt{Restrict}\ \Theta_2\ h \rrbracket  & \text{by definition of } S_1 \bullet S_2 \text{ and } \llbracket . \rrbracket \\
        &= \llbracket \mathtt{Restrict}\ (\Theta_1 \cdot \Theta_2)\ h \rrbracket & \text{by definition of } \mathtt{Restrict}\ .\ h 
    \end{align*}
    \textit{Case} $\Theta = \Theta_1 + \Theta_2$:
    \begin{align*}
        &\phantom{{}={}} \llbracket \Theta_1 + \Theta_2 \rrbracket | h \\
        &= (\llbracket \Theta_1 \rrbracket \cup \llbracket \Theta_2 \rrbracket) | h & \text{by definition of } \llbracket . \rrbracket\\
        &= \{hs | hs \in (\llbracket \Theta_1 \rrbracket \cup \llbracket \Theta_2 \rrbracket) \land h \in hs \} & \text{by definition of } .|h \\
        &= \{hs_1|hs_1 \in \llbracket \Theta_1 \rrbracket \land h \in hs_1 \} \cup \{hs_2|hs_2 \in \llbracket \Theta_2 \rrbracket \land h \in hs_2\} & \text{by set theory}\\
        &= \llbracket \Theta_1 \rrbracket | h \cup \llbracket \Theta_2 \rrbracket | h & \text{by definition of } .|h \\
        &= \llbracket \mathtt{Restrict}\ \Theta_1\ h\rrbracket \cup \llbracket \mathtt{Restrict}\ \Theta_2\ h\rrbracket & \text{by induction hypothesis} \\
        &= \llbracket \mathtt{Restrict}\ \Theta_1\ h + \mathtt{Restrict}\ \Theta_2\ h\rrbracket & \text{by definition of } \llbracket . \rrbracket \\
        &= \llbracket \mathtt{Restrict}\ (\Theta_1 + \Theta_2)\ h \rrbracket & \text{by definition of } \mathtt{Restrict}\ .\ h
    \end{align*}
    }{ By straight-forward induction on $\Theta$. }
\end{proof}

\subparagraph{Negated Restriction}
Dually to the restrict operator, $\texttt{NegRestrict}~\Theta~h$ produces only those choices/subsets where $h$ is invalid.
Semantically, negated restriction is defined as $S|_{\neg h} \triangleq \{hs \mid hs \in S \land h \not\in hs\}$.
Syntactically we define \textit{Negated Restriction} by induction on $\Theta$ as shown in Figure \ref{fig:negrestrict}.
\begin{figure}[ht]
    \boxedpage{\textwidth}{\begin{align*}
        \mathtt{NegRestrict}\ 0\ h &\triangleq 0 \\
        \mathtt{NegRestrict}\ 1\ h &\triangleq 1 \\
        \mathtt{NegRestrict}\ g\ h &\triangleq \begin{cases}
            0 & \text{if }g=h\\
            g & \text{otherwise}
        \end{cases}\\
        \mathtt{NegRestrict}\ (\Theta_1 \cdot \Theta_2)\ h &\triangleq (\mathtt{NegRestrict}\ \Theta_1\ h) \cdot (\mathtt{NegRestrict}\ \Theta_2\ h)\\
        \mathtt{NegRestrict}\ (\Theta_1 + \Theta_2)\ h &\triangleq (\mathtt{NegRestrict}\ \Theta_1\ h) + (\mathtt{NegRestrict}\ \Theta_2\ h) \\
    \end{align*}}
    \caption{Syntactic definition of the $\texttt{NegRestrict}~\Theta~h$ operator}
    \label{fig:negrestrict}
\end{figure}

The equivalence of the syntactic and semantic definition is captured by Lemma \ref{lem:negrestrict-equal}.
\begin{lemma}
\label{lem:negrestrict-equal}
$\llbracket \Theta \rrbracket |_{ \neg h} == \llbracket \mathtt{NegRestrict}\ \Theta\ h \rrbracket.$
\end{lemma}

\begin{proof}
  \iftoggle{fullproofs}{
    By induction on $\Theta$.\\
    \textit{Case} $\Theta = 0$:
    \begin{align*}
        &\phantom{{}={}} \llbracket 0 \rrbracket | \neg h \\
        &= \{\}|\neg h & \text{by definition of } \llbracket.\rrbracket \\
        &= \{hs|hs \in \{\} \land h \not\in hs \} & \text{by definition of } .|\neg h \\
        &= \{\} & \text{by set theory} \\
        &= \llbracket 0 \rrbracket & \text{by definition of } \llbracket . \rrbracket \\
        &= \llbracket \mathtt{NegRestrict}\ 0\ h \rrbracket & \text{by definition of } \mathtt{NegRestrict}\ .\ h 
    \end{align*}
    \textit{Case} $\Theta = 1$:
    \begin{align*}
        &\phantom{{}={}} \llbracket 1 \rrbracket | \neg h \\
        &= \{\{\}\}|\neg h & \text{by definition of } \llbracket.\rrbracket \\
        &= \{hs|hs \in \{\{\}\} \land h \not\in hs \} & \text{by definition of } .|\neg h \\
        &= \{\{\}\} & \text{by set theory} \\
        &= \llbracket 1 \rrbracket & \text{by definition of } \llbracket . \rrbracket \\
        &= \llbracket \mathtt{NegRestrict}\ 1\ h \rrbracket & \text{by definition of } \mathtt{NegRestrict}\ .\ h 
    \end{align*}
    \textit{Case} $\Theta=g$:
    \begin{align*}
        &\phantom{{}={}} \llbracket g \rrbracket | \neg h \\
        &= \{\{g\}\}|\neg h & \text{by definition of } \llbracket . \rrbracket \\
        &= \{hs|hs \in \{\{g\}\} \land h \not\in hs \} & \text{by definition of } .|\neg h \\
        &\phantom{{}={}} \textit{Subcase } h = g \\
        &= \{\} & \text{by set theory} \\
        &= \llbracket 0 \rrbracket & \text{by definition of } \llbracket . \rrbracket \\
        &= \llbracket \mathtt{NegRestrict}\ 0\ h \rrbracket & \text{by assumption } h=g \text{ and by definition of } \mathtt{NegRestrict}\ .\ h \\
        &\phantom{{}={}} \textit{Subcase } h \neq g \\
        &= \{\{g\}\} & \text{by set theory} \\
        &= \llbracket g \rrbracket & \text{by definition of } \llbracket . \rrbracket \\
        &= \llbracket \mathtt{NegRestrict}\ g\ h \rrbracket & \text{by definition of } \mathtt{NegRestrict}\ .\ h \text{ and by assumption } h \neq g \\
    \end{align*}
    \textit{Case} $\Theta = \Theta_1 \cdot \Theta_2$:\\
    \begin{align*}
        &\phantom{{}={}} \llbracket \Theta_1 \cdot \Theta_2 \rrbracket | \neg h \\
        &= (\llbracket \Theta_1 \rrbracket \bullet \llbracket \Theta_2 \rrbracket)|\neg h & \text{by definition of } S_1 \bullet S_2\\
        &= \{ hs_1 \cup hs_2 | hs_1 \in \llbracket \Theta_1 \rrbracket \land hs_2 \in \llbracket \Theta_2 \llbracket \} | \neg h & \text{by definition of } \llbracket . \rrbracket \\
        &= \{ hs_1 \cup hs_2 | hs_1 \in \llbracket \Theta_1 \rrbracket \land hs_2 \in \llbracket \Theta_2 \rrbracket \land h \not\in (hs_1 \cup hs_2)\} & \text{by definition of } .|\neg h \\
        &= \{ hs_1 \cup hs_2 | hs_1 \in \llbracket \Theta_1 \rrbracket \land hs_2 \in \llbracket \Theta_2 \rrbracket \land h \not\in hs_1 \land h \not\in hs_2)\} & \text{by set theory and logic} \\
        &= \{ hs_1 \cup hs_2 | (hs_1 \in \llbracket \Theta_1 \rrbracket \land h \not\in hs_1) \land (hs_2 \in \llbracket \Theta_2 \rrbracket \land h \not\in hs_2)\} & \text{by set theory and logic} \\
        &= \{hs_1 | hs_1 \in \llbracket \Theta_1 \rrbracket \land h \not\in hs_1\} \bullet \{hs_2 | hs_2 \in \llbracket \Theta_2 \rrbracket \land h \not\in hs_2\} & \text{by definition of } S_1 \bullet S_2 \\
        &= \llbracket \Theta_1 \rrbracket | \neg h \bullet \llbracket \Theta_2 \rrbracket | \neg h & \text{ by definition of } .|\neg h \\
        &= \llbracket \mathtt{NegRestrict}\ \Theta_1\ h\rrbracket \bullet \llbracket \mathtt{NegRestrict}\ \Theta_2\ h\rrbracket & \text{by induction hypothesis} \\
        &= \llbracket (\mathtt{NegRestrict}\ \Theta_1\ h) \cdot (\mathtt{NegRestrict}\ \Theta_2\ h)\rrbracket & \text{By definition of } \llbracket . \rrbracket \\
        &= \llbracket \mathtt{NegRestrict}\ (\Theta_1\ \cdot \Theta_2)\ h\rrbracket & \text{by definition of } \mathtt{NegRestrict}\ .\ h \\ 
    \end{align*}
    \textit{Case} $\Theta = \Theta_1 + \Theta_2$:
    \begin{align*}
        &\phantom{{}={}} \llbracket \Theta_1 + \Theta_2 \rrbracket | \neg h \\
        &= (\llbracket \Theta_1 \rrbracket \cup \llbracket \Theta_2 \rrbracket) | \neg h & \text{by definition of } \llbracket . \rrbracket\\
        &= \{hs | hs \in (\llbracket \Theta_1 \rrbracket \cup \llbracket \Theta_2 \rrbracket) \land h \not\in hs \} & \text{by definition of } .|\neg h \\
        &= \{hs_1|hs_1 \in \llbracket \Theta_1 \rrbracket \land h \not\in hs_1 \} \cup \{hs_2|hs_2 \in \llbracket \Theta_2 \rrbracket \land h \not\in hs_2\} & \text{by set theory}\\
        &= \llbracket \Theta_1 \rrbracket | \neg h \cup \llbracket \Theta_2 \rrbracket | \neg h & \text{by definition of } .|\neg h \\
        &= \llbracket \mathtt{NegRestrict}\ \Theta_1\ h\rrbracket \cup \llbracket \mathtt{NegRestrict}\ \Theta_2\ h\rrbracket & \text{by induction hypothesis} \\
        &= \llbracket \mathtt{NegRestrict}\ \Theta_1\ h + \mathtt{NegRestrict}\ \Theta_2\ h\rrbracket & \text{by definition of } \llbracket . \rrbracket \\
        &= \llbracket \mathtt{NegRestrict}\ (\Theta_1 + \Theta_2)\ h \rrbracket & \text{by definition of } \mathtt{NegRestrict}\ .\ \neg h
    \end{align*}
    }{By straight-forward induction on $\Theta$.}
\end{proof}

\subparagraph{Inclusion}
$\texttt{Includes}~\Theta~h$ traverses $\Theta$ and checks to make sure that $h$ is valid in every path.
Semantically this says that $h$ is a member of every element of $\llbracket\Theta\rrbracket$, i.e., $h \sqsubset S \triangleq \bigwedge (hs \in S \land h \in hs)$.
Syntactically we define \textit{Inclusion} by induction on $\Theta$ as shown in Figure \ref{fig:inclusion}.
\begin{figure}[ht]
    \boxedpage{\textwidth}{
        \begin{align*}
            \mathtt{Includes}\ 0\ h &\triangleq \mathit{false}\\
            \mathtt{Includes}\ 1\ h &\triangleq \mathit{false}\\
            \mathtt{Includes}\ g\ h &\triangleq \begin{cases}
                \mathit{true} & \text{if } g=h\\
                \mathit{false} & \text{otherwise}
            \end{cases}\\
            \mathtt{Includes}\ (\Theta_1 \cdot \Theta_2)\ h &\triangleq (\mathtt{Includes}\ \Theta_1\ h) \lor (\mathtt{Includes}\ \Theta_2\ h)\\
            \mathtt{Includes}\ (\Theta_1 + \Theta_2)\ h &\triangleq (\mathtt{Includes}\ \Theta_1\ h) \land (\mathtt{Includes}\ \Theta_2\ h)\\
        \end{align*}}
    \caption{Syntactic definition of the $\texttt{Includes}~\Theta~h$ operator}
    \label{fig:inclusion}
\end{figure}

The equivalence of the syntactic and semantic definition is captured by Lemma \ref{lem:includes-equal}.
\begin{lemma}
\label{lem:includes-equal}
$\forall hs\in S.h\in hs == \mathtt{Includes}\ \Theta\ h.$
\end{lemma}

\begin{proof}
  \iftoggle{fullproofs}{
    By induction on $\Theta$. \\
    \textit{Case} $\Theta = 0$:
    \begin{align*}
        &\phantom{{}={}} h \sqsubset \llbracket 0 \rrbracket \\
        &= \bigwedge (hs \in \{\} \land h \in hs) & \text{by definition of } \llbracket.\rrbracket \\
        &= \mathtt{false} & \text{by logic and set theory}\\
        &= \mathtt{Includes}\ 0\ h & \text{by definition of } \mathtt{Includes}\ .\ h
    \end{align*}
    \textit{Case} $\Theta = 1$:
    \begin{align*}
        &\phantom{{}={}} h \sqsubset \llbracket 1 \rrbracket \\
        &= \bigwedge (hs \in \{\{\}\} \land h \in hs) & \text{by definition of } \llbracket.\rrbracket \\
        &= \mathtt{false} & \text{by logic and set theory}\\
        &= \mathtt{Includes}\ 1\ h & \text{by definition of } \mathtt{Includes}\ .\ h
    \end{align*}
    \textit{Case} $\Theta = g$:
    \begin{align*}
        &\phantom{{}={}} h \sqsubset \llbracket g \rrbracket \\
        &= \bigwedge (hs \in \{\{g\}\} \land h \in hs) & \text{by definition of } \llbracket.\rrbracket \\
        &\phantom{{}={}} \textit{Subcase } h=g \\
        &= \mathtt{true} & \text{by logic and set theory} \\
        &= \mathtt{Includes}\ g\ h & \text{by definition of } (\mathtt{Includes}\ .\ h) \text{ and assumption } h=g \\
        &\phantom{{}={}} \textit{Subcase } h\neq g \\
        &= \mathtt{false} & \text{by logic and set theory}\\
        &= \mathtt{Includes}\ g\ h & \text{by definition of } (\mathtt{Includes}\ .\ h) \text{ and assumption } h\neq g \\
    \end{align*}
    \textit{Case} $\Theta = \Theta_1 \cdot \Theta_2$:
    \begin{align*}
        &\phantom{{}={}} h \sqsubset \llbracket \Theta_1 \cdot \Theta_2 \rrbracket \\
        &= h \sqsubset (\llbracket \Theta_1 \rrbracket \bullet \llbracket \Theta_2 \rrbracket) & \text{by definition of } S_1 \bullet S_2 \\
        &= h \sqsubset \{hs_1 \cup hs_2 | hs_1 \in \llbracket \Theta_1 \rrbracket \land hs_2 \in \llbracket \Theta_2 \rrbracket \} & \text{by definition of } \llbracket . \rrbracket \\
        &= h \sqsubset \{hs_1 | hs_1 \in \llbracket \Theta_1 \rrbracket\} \lor h \sqsubset \{hs_2 | hs_2 \in \llbracket \Theta_2 \rrbracket \} & \text{by set theory and logic} \\
        &= \bigwedge (hs_1 \in \llbracket \Theta_1 \rrbracket \land h \in hs_1) \lor \bigwedge (hs_2 \in \llbracket \Theta_2 \rrbracket \land h \in hs_2) & \text{by definition of } h \sqsubset .\\
        &= h \sqsubset \llbracket \Theta_1 \rrbracket \lor h \sqsubset \llbracket \Theta_2 \rrbracket & \text{by definition of } \llbracket . \rrbracket \\
        &= (\mathtt{Includes}\ \Theta_1\ h) \lor (\mathtt{Includes}\ \Theta_2\ h) & \text{by induction hypothesis} \\
        &= (\mathtt{Includes}\ \Theta_1 \cdot \Theta_2\ h) & \text{by definition of } \mathtt{Includes}\ .\ h
    \end{align*}
    \textit{Case} $\Theta = \Theta_1 + \Theta_2$:
    \begin{align*}
        &\phantom{{}={}} h \sqsubset \llbracket \Theta_1 + \Theta_2 \rrbracket \\
        &= h \sqsubset (\llbracket \Theta_1 \rrbracket \cup \llbracket \Theta_2 \rrbracket) & \text{by definition of } \llbracket .\rrbracket \\
        &= \bigwedge (hs \in (\llbracket \Theta_1 \rrbracket \cup \llbracket \Theta_2 \rrbracket) \land h \in hs) & \text{by definition of } h\sqsubset . \\
        &= \bigwedge (hs_1 \in \llbracket \Theta_1 \rrbracket \land h \in hs_1 \land hs_2 \in \llbracket \Theta_2 \rrbracket \land h \in hs_2) & \text{by set theory and logic} \\
        &= \bigwedge (hs_1 \in \llbracket \Theta_1 \rrbracket \land h \in hs_1) \land \bigwedge (hs_2 \in \llbracket \Theta_2 \rrbracket \land h \in hs_2) & \text{by set theory and logic} \\
        &= h \sqsubset \llbracket \Theta_1 \rrbracket \land h \sqsubset \llbracket \Theta_2 \rrbracket & \text{by definition of } h\sqsubset .\\
        &= (\mathtt{Includes}\ \Theta_1\ h) \land (\mathtt{Includes}\ \Theta_2\ h) & \text{by induction hypothesis} \\
        &= (\mathtt{Includes}\ (\Theta_1 + \Theta_2)\ h) & \text{by definition of } \mathtt{Includes}\ .\ h
    \end{align*}
    }{ By straight-forward induction on $\Theta$. }
\end{proof}

\subparagraph{Removal}
$\texttt{Remove}~\Theta~h$ removes $h$ from every path, which means, semantically that it removes $h$ from every element of
$\llbracket\Theta\rrbracket$, i.e., $S \setminus h \triangleq \{hs \mid hs \in S \land hs \setminus \{h\} \}$.
Syntactically we define \textit{Removal} by induction on $\Theta$ as shown in Figure \ref{fig:remove}.
\begin{figure}[ht]
    \boxedpage{\textwidth}{
        \begin{align*}
            \mathtt{Remove}\ 0\ h &\triangleq 0\\
            \mathtt{Remove}\ 1\ h &\triangleq 1\\
            \mathtt{Remove}\ g\ h &\triangleq \begin{cases}
                1 & \text{if } g=h\\
                g & \text{otherwise}
            \end{cases}\\
            \mathtt{Remove}\ (\Theta_1 \cdot \Theta_2)\ h &\triangleq (\mathtt{Remove}\ \Theta_1\ h) \cdot (\mathtt{Remove}\ \Theta_2\ h)\\
            \mathtt{Remove}\ (\Theta_1 + \Theta_2)\ h &\triangleq (\mathtt{Remove}\ \Theta_1\ h) + (\mathtt{Remove}\ \Theta_2\ h)\\
        \end{align*}}
    \caption{Syntactic definition of the $\texttt{Remove}~\Theta~h$ operator}
    \label{fig:remove}
\end{figure}

The equivalence of the syntactic and semantic definition is captured by Lemma \ref{lem:remove-equal}.
\begin{lemma}
\label{lem:remove-equal}
$\llbracket \Theta \rrbracket \setminus h == \llbracket \mathtt{Remove}\ \Theta\ h \rrbracket$.
\end{lemma}

\begin{proof}
  \iftoggle{fullproofs}{
    By induction on $\Theta$.\\
    \textit{Case} $\Theta=0$:
    \begin{align*}
        &\phantom{{}={}} \llbracket 0 \rrbracket \setminus h \\
        &= \{\} \setminus h & \text{by definition of } \llbracket . \rrbracket \\
        &= \{hs | hs \in \{\} \land hs \setminus \{h\}\} & \text{by definition of } .\setminus h \\
        &= \{\} & \text{by set theory} \\
        &= \llbracket 0 \rrbracket & \text{by definition of } \llbracket . \rrbracket \\
        &= \llbracket \mathtt{Remove}\ 0\ h \rrbracket & \text{by definition of }\mathtt{Remove}\ .\ h
    \end{align*}
    \textit{Case} $\Theta=1$:
    \begin{align*}
        &\phantom{{}={}} \llbracket 1 \rrbracket \setminus h \\
        &= \{\{\}\} \setminus h & \text{by definition of } \llbracket . \rrbracket \\
        &= \{hs | hs \in \{\{\}\} \land hs \setminus \{h\}\} & \text{by definition of } .\setminus h \\
        &= \{\{\}\} & \text{by set theory} \\
        &= \llbracket 1 \rrbracket & \text{by definition of } \llbracket . \rrbracket \\
        &= \llbracket \mathtt{Remove}\ 1\ h \rrbracket & \text{by definition of }\mathtt{Remove}\ .\ h
    \end{align*}
    \textit{Case} $\Theta=g$:
    \begin{align*}
        &\phantom{{}={}} \llbracket g \rrbracket \setminus h \\
        &= \{\{g\}\} \setminus h & \text{by definition of } \llbracket . \rrbracket \\
        &= \{hs | hs \in \{\{g\}\} \land hs \setminus \{h\}\} & \text{by definition of } .\setminus h \\
        &\phantom{{}={}} \textit{Subcase } h = g \\
        &= \{\{\}\} & \text{by set theory} \\
        &= \llbracket 1 \rrbracket & \text{by definition of } \llbracket . \rrbracket \\
        &= \llbracket \mathtt{Remove}\ 1\ h \rrbracket & \text{by definition of }\mathtt{Remove}\ .\ h \\
        &\phantom{{}={}} \textit{Subcase } h \neq g \\
        &= \{\{g\}\} & \text{by set theory} \\
        &= \llbracket g \rrbracket & \text{by definition of } \llbracket . \rrbracket \\
        &= \llbracket \mathtt{Remove}\ g\ h \rrbracket & \text{by definition of }\mathtt{Remove}\ .\ h
    \end{align*}
    \textit{Case} $\Theta = \Theta_1 \cdot \Theta_2$:
    \begin{align*}
        &\phantom{{}={}} \llbracket \Theta_1 \cdot \Theta_2 \rrbracket \setminus h \\
        &= (\llbracket \Theta_1 \rrbracket \bullet \llbracket \Theta_2 \rrbracket) \setminus h & \text{by definition of } S_1 \bullet S_2\\
        &= \{ hs_1 \cup hs_2 | hs_1 \in \llbracket \Theta_1 \rrbracket \land hs_2 \in \llbracket \Theta_2 \llbracket \} \setminus h & \text{by definition of } \llbracket . \rrbracket \\
        &= \{ hs_1 \cup hs_2 | hs_1 \in \llbracket \Theta_1 \rrbracket \land hs_2 \in \llbracket \Theta_2 \rrbracket \land (hs_1 \cup hs_2) \setminus h\} & \text{by definition of } .\setminus h \\
        &= \{ hs_1 \cup hs_2 | hs_1 \in \llbracket \Theta_1 \rrbracket \land hs_2 \in \llbracket \Theta_2 \rrbracket \land hs_1 \setminus h \land hs_2 \setminus h)\} & \text{by set theory and logic} \\
        &= \{ hs_1 \cup hs_2 | (hs_1 \in \llbracket \Theta_1 \rrbracket \land hs_1 \setminus h) \land (hs_2 \in \llbracket \Theta_2 \rrbracket \land hs_2 \setminus h)\} & \text{by set theory and logic} \\
        &= \{hs_1 | hs_1 \in \llbracket \Theta_1 \rrbracket \land hs_1 \setminus h\} \bullet \{hs_2 | hs_2 \in \llbracket \Theta_2 \rrbracket \land hs_2 \setminus h\} & \text{by definition of } S_1 \bullet S_2 \\
        &= \llbracket \Theta_1 \rrbracket \setminus h \bullet \llbracket \Theta_2 \rrbracket \setminus h & \text{ by definition of } .\setminus h \\
        &= \llbracket \mathtt{Remove}\ \Theta_1\ h\rrbracket \bullet \llbracket \mathtt{Remove}\ \Theta_2\ h\rrbracket & \text{by induction hypothesis} \\
        &= \llbracket (\mathtt{Remove}\ \Theta_1\ h) \cdot (\mathtt{Remove}\ \Theta_2\ h)\rrbracket & \text{By definition of } \llbracket . \rrbracket \\
        &= \llbracket \mathtt{Remove}\ (\Theta_1\ \cdot \Theta_2)\ h\rrbracket & \text{by definition of } \mathtt{Remove}\ .\ h \\ 
    \end{align*}
    \textit{Case} $\Theta = \Theta_1 + \Theta_2$:
    \begin{align*}
        &\phantom{{}={}} \llbracket \Theta_1 + \Theta_2 \rrbracket \setminus h \\
        &= (\llbracket \Theta_1 \rrbracket \cup \llbracket \Theta_2 \rrbracket) \setminus h & \text{by definition of } \llbracket . \rrbracket \\
        &= \{hs | hs \in (\llbracket \Theta_1 \rrbracket \cup \llbracket \Theta_2 \rrbracket) \land hs \setminus \{h\}\} & \text{by definition of } .\setminus h \\
        &= \{hs | hs \in (\llbracket \Theta_1 \rrbracket \cup \llbracket \Theta_2 \rrbracket) \land hs \setminus \{h\}\} & \text{by definition of } .\setminus h \\
        &= \{hs_1 | hs_1 \in \llbracket \Theta_1 \rrbracket \land hs_1 \setminus \{h\}\}\ \cup  & \text{by logic and set theory}\\
        &\phantom{{}={}} \{hs_2 | hs_2 \in \llbracket \Theta_2 \rrbracket \land hs_2 \setminus \{h\}\} \\
        &= \llbracket \Theta_1 \rrbracket \setminus h \cup \llbracket \Theta_2 \rrbracket \setminus h & \text{by definition of } .\setminus h \\
        &= \llbracket \mathtt{Remove}\ \Theta_1\ h \rrbracket \cup \llbracket \mathtt{Remove}\ \Theta_2\ h \rrbracket & \text{by induction hypothesis} \\
        &= \llbracket (\mathtt{Remove}\ \Theta_1\ h) \cdot (\llbracket \mathtt{Remove}\ \Theta_2\ h) \rrbracket & \text{by definition of } \llbracket .\rrbracket \\
        &= \llbracket \mathtt{Remove}\ (\Theta_1 \cdot \Theta_2)\ h \rrbracket & \text{by definition of } \mathtt{Remove}\ .\ h \\
    \end{align*}
    }{By straight-forward induction on $\Theta$. }
\end{proof}

\subparagraph{Emptiness}
$\texttt{Empty}~\Theta$ checks if $\Theta$ is semantically empty. Syntactically we define \textit{Empty} by induction on $\Theta$ as shown in Figure \ref{fig:empty}.

\begin{figure}[ht]
    \boxedpage{\textwidth}{
        \begin{align*}
            \mathtt{Empty}\ 0 &\triangleq \mathit{true}\\
            \mathtt{Empty}\ 1 &\triangleq \mathit{false}\\
            \mathtt{Empty}\ h &\triangleq \mathit{false}\\
            \mathtt{Empty}\ (\Theta_1 \cdot \Theta_2) &\triangleq \mathtt{Empty}\ \Theta_1 \wedge \mathtt{Empty}\ \Theta_2\\
            \mathtt{Empty}\ (\Theta_1 + \Theta_2) &\triangleq \mathtt{Empty}\ \Theta_1 \wedge \mathtt{Empty}\ \Theta_2\\
        \end{align*}}
    \caption{Syntactic definition of the $\texttt{Remove}~\Theta~h$ operator}
    \label{fig:empty}
\end{figure}

The equivalence of the syntactic and semantic definition is captured by Lemma \ref{lem:empty-iff}.
\begin{lemma}
\label{lem:empty-iff}
$\llbracket \Theta \rrbracket == \{ \}$ if and only if $\mathtt{Empty}\ \Theta$.
\end{lemma}

\begin{proof}
  \iftoggle{fullproofs}{
    By induction on $\Theta$.\\
    \textit{Case} $\Theta=0$:
    We have $\llbracket 0 \llbracket = \{\}$ and $\mathit{Empty}~0 = \mathit{true}$.

    \textit{Case} $\Theta=1$:\\
    We have $\llbracket 1 \llbracket \neq \{ \}$ and $\mathit{Empty}~1 = \mathit{false}$.

    \textit{Case} $\Theta=h$:\\
    We have $\llbracket h \llbracket \neq \{ \}$ and $\mathit{Empty}~h = \mathit{false}$.

    \textit{Case} $\Theta=\Theta_1 \cdot \Theta_2$:
    By definition we have $\llbracket \Theta_1 \cdot \Theta_2 \rrbracket = \llbracket \Theta_1 \rrbracket \bullet \llbracket \Theta_2 \rrbracket$ which is equal to 
   $\{ s_1 \cup s_2 \mid s_1 \in \llbracket \Theta_1 \rrbracket \wedge s_2 \in \llbracket \Theta_2 \rrbracket \}$. It follows that $\llbracket \Theta_1 \cdot \Theta_2 \llbracket \neq \{ \}$ iff 
   $\llbracket \Theta_1 \llbracket \neq \{ \}$ and  $\llbracket \Theta_2 \llbracket \neq \{ \}$.     By induction hypothesis, we have 
    $\llbracket \Theta_1 \llbracket \neq \{ \}$ if and only if $\mathit{Empty}~\Theta_1 = \mathit{true}$, and 
    $\llbracket \Theta_2 \llbracket \neq \{ \}$ if and only if $\mathit{Empty}~\Theta_2 = \mathit{true}$. 
   The result follows as $\mathit{Empty}~(\Theta_1 \cdot \Theta_2) = \mathtt{Empty}\ \Theta_1 \wedge \mathtt{Empty}\ \Theta_2$.

    \textit{Case} $\Theta=\Theta_1 + \Theta_2$:
    By definition we have $\llbracket \Theta_1 + \Theta_2 \rrbracket = \llbracket \Theta_1 \rrbracket \cup \llbracket \Theta_2 \rrbracket$. 
    It follows that $\llbracket \Theta_1 \cdot \Theta_2 \llbracket \neq \{ \}$ iff  $\llbracket \Theta_1 \llbracket \neq \{ \}$ and  $\llbracket \Theta_2 \llbracket \neq \{ \}$.
    By induction hypothesis, we have 
    $\llbracket \Theta_1 \llbracket \neq \{ \}$ if and only if $\mathit{Empty}~\Theta_1 = \mathit{true}$, and 
    $\llbracket \Theta_2 \llbracket \neq \{ \}$ if and only if $\mathit{Empty}~\Theta_2 = \mathit{true}$. 
    The result follows as $\mathit{Empty}~(\Theta_1 + \Theta_2) = \mathtt{Empty}\ \Theta_1 \wedge \mathtt{Empty}\ \Theta_2$.
  }{  By straight-forward induction on $\Theta$.}
\end{proof}

 \section{Safety of \name}
\label{app:safety}

We prove safety in terms of progress and preservation.
Both theorems make use of the relation $H \models \Theta$ as defined in Figure \ref{fig:entails-relation}.
The empty header instance map only entails the empty header instance type $1$ (Rule \textsc{Ent-Empty}).
If a header instance $h$ is contained in the map of valid header instances $H$, $H$ entails the header instance type $h$ (Rule \textsc{Ent-Inst}).
The sequence type $\Theta_1 \cdot \Theta_2$ is entailed by the distinct union of the maps entailing $\Theta_1$ and $\Theta_2$ respectively (Rule \textsc{Ent-Seq}) and the choice type $\Theta_1 + \Theta_2$ is entailed either by the map entailing $\Theta_1$ or the map entailing $\Theta_2$ (Rules \textsc{Ent-ChoiceL} and \textsc{Ent-ChoiceR}).

We prove progress and preservation only for commands.
For expressions we formulate these properties as additional lemmas (Lemmas \ref{lem:expr-progress} and \ref{lem:expr-preservation}).
The respective proofs are straightforward for our system.

\begin{lemma}[Expression Progress]
    \label{lem:expr-progress}
    If $\cdot ;\Theta \vdash e:\tau$ and $H \models \Theta$, then either $e$ is a value or $\exists e'.\langle H,e \rangle \rightarrow e'$.
\end{lemma}

\begin{lemma}[Expression Preservation]
\label{lem:expr-preservation}
    If $\Gamma;\Theta \vdash e:\tau$ and $H \models \Theta$ and $\langle H,e\rangle \rightarrow e'$ then $\Gamma;\Theta\vdash e':\tau$.
\end{lemma}

\begin{lemma}[Expression Substitution]
\label{lem:expr-substitution}
    If $\Gamma, x : \tau; \Theta \vdash e : \tau'$ and $\cdot; \cdot \vdash \bar{v} : \bar{\tau}$ then $\Gamma;\Theta \vdash e[\bar{v}/\bar{x}] : \tau'$
\end{lemma}

\begin{lemma}[Entailment is Type Alternative]
    \label{lem:entailment}
    If $H \models \Theta$ then $\mathit{dom(H)} \in \llbracket \Theta \rrbracket$.
\end{lemma}

\begin{proof}
  \iftoggle{fullproofs}{
    By induction on $\Theta$. \\
    \textit{Case} $\Theta = 0$: The case immediately holds as $H \models 0$ is a contradiction.\\
    \textit{Case} $\Theta = 1$: By inversion of \textit{Entailment}, $H = \boldsymbol{\cdot}$, and so $\mathit{dom(H)=\{\} \in \llbracket 1 \rrbracket = \{\{\}\}}$. \\
    \textit{Case} $\Theta = h$: By inversion of \textit{Entailment}, $\mathit{dom(H)} = \{h\} \in \llbracket h \rrbracket = \{ \{h\}\}$. \\
    \textit{Case} $\Theta = \Theta_1 \cdot \Theta_2$: By inversion of \textit{Entailment}, $H = H_1 \cup H_2, H_1 \models \Theta_1, H_2 \models \Theta_2$.\\
    By induction hypothesis, $\mathit{dom(H_1)} \in \llbracket \Theta_1 \rrbracket$ and $\mathit{dom(H_2)} \in \llbracket \Theta_2 \rrbracket$.\\
    By set theory, $\mathit{dom}(H)=\mathit{dom}(H_1)\cup\mathit{dom}(H_2)$\\
    By induction hypothesis, $\mathit{dom}(H_1) \in \llbracket \Theta_1\rrbracket$ and $\mathit{dom}(H_2)\in\llbracket\Theta_2\rrbracket$.\\
    By definition of $\llbracket.\rrbracket$ and $(\bullet)$, $\llbracket\Theta\rrbracket = \llbracket \Theta_1\rrbracket \bullet \llbracket \Theta_2\rrbracket = \{hs_1 \cup hs_2|hs_1 \in \llbracket\Theta_1\rrbracket \land hs_2 \in \llbracket \Theta_2\rrbracket\}$ and therefore $\mathit{dom}(H_1)\cup\mathit{dom}(H_2)\in\{hs_1 \cup hs_2 | hs_1 \in \llbracket \Theta_1\rrbracket \land hs_2 \in \llbracket \Theta_2 \rrbracket\}$, i.e., $\mathit{dom}(H)\in\llbracket\Theta\rrbracket$.
    \textit{Case} $\Theta = \Theta_1 + \Theta_2$: By inversion of \textit{Entailment}, either $H \models \Theta_1$ or $H \models \Theta_2$.\\
    \textit{Subcase} $H \models \Theta_1$: By the induction hypothesis, $\mathit{dom(H)} \in \llbracket \Theta_1 \rrbracket$.
    and by set theory $\mathit{dom(H)} \in \llbracket \Theta_1 \rrbracket \cup \llbracket \Theta_2 \rrbracket$. \\
    \textit{Subcase} $H \models \Theta_2$: Symmetric to the previous subcase.
    }{ By straight-forward induction on $\Theta$.}
\end{proof}

\begin{lemma}[Included Instances in Domain]
    \label{lem:includes-domain}
    If $H \models \Theta$ and $\mathtt{Includes}~\Theta~h$, then $h \in \mathit{dom(H)}$.
\end{lemma}

\begin{proof}
  \iftoggle{fullproofs}{
    By induction on $\Theta$.\\
    \textit{Case} $\Theta=0$:
    The case immediately holds as $H \models 0$ is a contradiction.\\
    \textit{Case} $\Theta=1$:
    By inversion of \textit{Entailment}, $H=\boldsymbol{\cdot}$
    The case immediately holds, as $\mathtt{Includes}~\Theta~h$ is a contradiction.\\
    \textit{Case} $\Theta=g$:\\
    By inversion of \textit{Entailment}, $\mathit{dom}(H)=\{g\}$\\
    By assumption $\mathtt{Includes}~\Theta~h$, $h=g$.$\mathtt{Includes}~\{g\}~g$, i.e., $h \in \mathit{dom}(H)$. \\
    \textit{Case} $\Theta = \Theta_1 \cdot \Theta_2$:\\
    By inversion of \textit{Entailment}, $H=H_1 \cup H_2, H_1 \models \Theta_1, H_2 \models \Theta_2$. \\
    By set theory $\mathit{dom}(H) = \mathit{dom}(H_1) \cup \mathit{dom}(H_2)$\\
    By definition of \textit{Inclusion} and by assumption $\mathtt{Includes}~\Theta~h$, $\mathtt{Includes}~\Theta_1~h \lor \mathtt{Includes}~\Theta_2~h$\\
    \textit{Subcase} $\mathtt{Includes}~\Theta_1~h$: By induction hypothesis, $h\in\mathit{dom}(H_1)$ and by assumption $\mathit{dom}(H_1) \subseteq \mathit{dom}(H)$, we can conclude $h \in \mathit{dom}(H)$.\\
    \textit{Subcase} $\mathtt{Includes}~\Theta_2~h$: Symmetric to the previous subcase. \\
    \textit{Case} $\Theta = \Theta_1 + \Theta_2$: \\
    By inversion of \textit{Entailment}, either $H \models \Theta_1$ or $H \models \Theta_2$.\\
    By definition of \textit{Inclusion} and by assumption $\mathtt{Includes}~\Theta~h$, $\mathtt{Includes}~\Theta_1~h$ and $\mathtt{Includes}~\Theta_2~h$.\\
    \textit{Subcase} $H \models \Theta_1$: By induction hypothesis, we can conclude $h\in\mathit{dom}(H)$. \\
    \textit{Subcase} $H \models \Theta_2$: Symmetric to the previous subcase.
    }{By straight-forward induction on $\Theta$.}
\end{proof}

\subsection{Control Plane Assumptions}
\label{app:control-plane-assumptions}

The following propositions model the assumptions about the control
plane functions $\mathcal{CA}$ and $\mathcal{CV}$ that are required to
prove type safety.

\begin{proposition}[Control Plane Reads]
\label{prop:cp-reads}
If $H \models \Theta$ and 
$\mathcal{CV}(t) = \bar{S}$ and
$\bar{e} = \{e_j \mid (e_j,m_j) \in t.\mathit{reads()} \wedge \neg \mathsf{maskable}(t,e_j,m_j)\}$ and 
$\Gamma; \Theta \vdash e_j : \tau_j$ for $e_j \in \bar{e}$ then 
$\mathcal{CA}(t,H) = (a_i, \bar{v})$.
\end{proposition}

\begin{proposition}[Control Plane Action Data]
\label{prop:cp-action-data}
If $H \models \Theta$ and 
$\mathcal{CA}(t,H) = (a_i, \bar{v})$ and
$\mathcal{A}(a_i) = \lambda \bar{x} : \bar{\tau}.~c_i$ then 
$\cdot; \cdot \vdash \bar{v} : \bar{\tau}$  
\end{proposition}

\begin{proposition}[Control Plane Assumptions]
\label{prop:cp-valid}
If $H \models \Theta$ and 
$\mathcal{CA}(t,H) = (a_i, \bar{v})$ and 
$\mathcal{CV}(t) = \bar{S}$ then 
$H \models \mathtt{Restrict}~\Theta~S_i$.
\end{proposition} 

\subsection{Progress}

\begin{theorem}[Progress]
If $\cmdtype{\cdot}{\Theta}{c}{\Theta'}$ and $H \models \Theta$, then either $c=\mathit{skip}$ or $\exists \langle I',O',H',c' \rangle.~\langle I,O,H,c \rangle \rightarrow \langle I',O',H',c' \rangle$
\end{theorem}

\begin{proof}
By induction on typing derivations of $\cmdtype{\cdot}{\Theta}{c}{\Theta'}$. 
\begin{description}

\item{\textit{Case} \textsc{T-Skip}:}
$c = \syntax{skip}$\\[.25em]
Immediate. 

\item{\textit{Case} \textsc{T-Extr}:}
$c = \syntax{extract(h)}$\\[.25em]
Let $(I',v) = \syntax{deserialize_\eta(I)}$ and $O'=O$ and $H'=H[h \mapsto v]$ and $c'=\syntax{skip}$.
The result follows by \textsc{E-Skip}.

\item{\textit{Case} \textsc{T-Emit}:}
$c=\syntax{emit(h)}$\\[.25em]
If $h\not\in\mathit{dom}(H)$, let $I'=I$ and $O'=O$ and $H'=H$, and $c'=\syntax{skip}$. The result follows by \textsc{E-EmitInvalid}. 
Otherwise, $h\in\mathit{dom(H)}$. Let $H(h)=v$ and $\bar{B} = \syntax{serialize}_\eta(v)$ and $I' = I$ and $O' = O.\bar{B}$ and $H' = H$ and $c' = \syntax{skip}$. The result follows by \textsc{E-Emit}.

\item{\textit{Case} \textsc{T-Seq}:}
$c=\syntax{c_1;c_2}$ and $\cmdtype{\cdot}{\Theta}{c_1}{\Theta_1}$ and $\cmdtype{\cdot}{\Theta_1}{c_2}{\Theta_2}$\\[.25em]
By induction hypothesis, $c_1$ is either \syntax{skip} or there is some $\langle I',O',H',c_1'\rangle$, such that $\langle I,O,H,c_1\rangle \rightarrow \langle I',O',H',c_1'\rangle$.\\
If $c_1=\syntax{skip}$, let $I' = I$ and $O' = O$ and $H' = H$ and $c' = c_2$. The result follows by \textsc{E-Seq}.
Otherwise, the result follows by \textsc{E-Seq1}.

\item{\textit{Case} \textsc{T-If}:}
$c=\syntax{if\ (e)\ then\ c_1\ else\ c_2}$ and $\cdot; \Theta \vdash e: \mathit{Bool}$ and $\cmdtype{\cdot}{\Theta}{c_1}{\Theta_1}$ and $\cmdtype{\cdot}{\Theta}{c_2}{\Theta_2}$\\[.25em]
By the progress theorem for expressions, we have that $e$ is either \texttt{true}, \texttt{false}, or there is some $e'$ such that  $\langle H, e \rangle \rightarrow e'$. 
\begin{description}
\item{\textit{Subcase} $e = \mathtt{true}$:}
Let $I'=I$ and $O'=O$ and $H'=H$ and $c'=c_1$. The result follows by \textsc{E-IfTrue}. 
\item{\textit{Subcase} $e = \mathtt{false}$:}
Symmetric to the previous case.
\item{\textit{Subcase} $\langle H, e \rangle \rightarrow e'$:}
Let $I'=I$ and $O'=O$ and $H'=H$ and $c'=\syntax{if~(e')~c_1~c_2}$. The result follows by \textsc{E-If}. 
\end{description}

\item{\textit{Case} \textsc{T-IfValid}:}
$c=\syntax{valid(h)\ c_1\ else\ c_2}$\\[.25em]
If  $h\in\mathit{dom}(H)$, let $I'=I$ and $O'=O$ and $H'=H$ and $c'=c_1$. The result follows by \textsc{E-IfValidTrue}
Otherwise, $h\not\in\mathit{dom}(H)$. Let $I'=I$ and $O'=O$ and $H'=H$ and $c'=c_2$. The result follows by \textsc{E-IfValidFalse}

\item{\textit{Case} \textsc{T-Apply}}:
$c=\syntax{t.apply()}$\\[.25em]
By Proposition~\ref{prop:cp-reads}, we have $\mathcal{CA}(t,H) = (a_i, \bar{v})$.
Let $\mathcal{A}(a) = \lambda \bar{x}:\bar{\tau}.~c_i$.
Let $I' = I$ and $O' = O$ and $H' = H$ and $c' = c_i[\bar{v}/\bar{x}]$. The result follows by \textsc{E-Apply}.

\item{\textit{Case} \textsc{T-Add}:}
 $c=\syntax{add(h)}$\\[.25em]
If  $h\in\mathit{dom}(H)$, let $I'=I$ and $O'=O$ and $H'=H$ and $c'=\syntax{skip}$. The result follows by \textsc{E-AddValid}.
Otherwise, $h\not\in\mathit{dom}(H)$. Let $v = \mathit{init}_\eta$ and $I'=I$ and $O'=O$ and $H'=H[h \mapsto v]$ and $c'=\syntax{skip}$. The result follows by \textsc{E-Add}

\item{\textit{Case} \textsc{T-Remove}:}
 $c=\syntax{remove(h)}$\\[.25em]
Let $I'=I$ and $O'=O$ and $H'=H \setminus h$ and $c'=\syntax{skip}$. The result follows by \textsc{E-Remove}.

\item{\textit{Case} \textsc{T-Mod}:}
$c=\syntax{h.f=e}$ and $\mathtt{Includes}~\Theta~h$ and $\mathcal{F}(h,f) = \tau_i$ and $\cdot; \Theta \vdash e: \tau_i$\\[.25em]
By the progress rule for expressions, either $e$ is a value or there is some $e'$ such that $\langle H,e \rangle \rightarrow e'$.
\begin{description}
\item{\textit{Subcase} $e = v$:}
By Lemma \ref{lem:includes-domain}: $h\in\mathit{dom(H)}$. Let $r = H(h)$ and $r' = \{ r~\textit{with}~f = v\}$. 
Also let $I' = I$ and $O' = O$ and $H' = H[h \mapsto r']$ and $c' = \syntax{skip}$. The result follows by \textsc{E-Mod}.
\item{\textit{Subcase} $\langle H, e \rangle \rightarrow e'$:}
Let $I'=I$ and $O'=O$ and $H'=H$ and $c'=\syntax{h.f = e'}$. The result follows by \textsc{E-Mod1}. 
\end{description}

\item{\textit{Case} \textsc{T-Zero}:}
$\mathtt{Empty}~\Theta_1$\\[.25em]
By Lemma~\ref{lem:entailment}, we have $\mathit{dom}(H) \in \llbracket \Theta_1 \rrbracket$. 
By Lemma~\ref{lem:empty-iff}, we have $\llbracket \Theta_1 \rrbracket = \{ \}$, which is a contradiction.\qedhere
\end{description}
\end{proof}

\subsection{Preservation}

\begin{lemma}[Restriction Entailed]
    \label{lem:restrict-domain-entail}
    If $H \models \Theta$ and $h \in \mathit{dom(H)}$ then $H \models \mathtt{Restrict}~\Theta~h$.
\end{lemma}
\begin{proof}
  \iftoggle{fullproofs}{
    By induction on $\Theta$.\\
    \textit{Case} $\Theta = 0$: The case immediately holds as $H \models 0$ is a contradiction.\\
    \textit{Case} $\Theta=1$: By inversion of \textit{Entailment}, $H=\boldsymbol{\cdot}$.
    The case immediately holds as $h \in \mathit{dom(\boldsymbol{\cdot})}$ is a contradiction. \\
    \textit{Case} $\Theta = g$: By inversion of \textit{Entailment}, $\mathit{dom(H) = \{ g \}}$, and so $h = g$. \\
    By definition of \textit{Restriction} $\mathtt{Restrict}~\Theta~h = \mathtt{Restrict}~g~g = g$.
    By \textsc{Ent-Inst} $H \models g$, i.e., $H \models \Theta$. \\
    \textit{Case} $\Theta = \Theta_1 \cdot \Theta_2$: By inversion of \textit{Entailment} $H=H_1 \cup H_2, H_1 \models \Theta_1, H_2 \models \Theta_2$.
    By $h \in \mathit{dom(H)}$, either $h \in \mathit{dom(H_1)}$ or $h \in \mathit{dom(H_2)}$.\\
    \textit{Subcase} $h \in \mathit{dom(H_1)}$: By the induction hypothesis, we have $H_1 \models \mathtt{Restrict}~\Theta_1~h$.
    By \textsc{Ent-Seq}, we have $H_1 \cup H_2 \models \mathtt{Restrict}~\Theta_1~h \cdot \Theta_2$.
    By \textsc{Ent-ChoiceL}, we have $H_1 \cup H_2 \models (\mathtt{Restrict}~\Theta_1~h \cdot \Theta_2) + (\Theta_1 \cdot \mathtt{Restrict}~\Theta_2~h)$ which finishes the case.\\
    \textit{Subcase} $h \in \mathit{dom(H_2)}$: Symmetric to the previous subcase. \\
    \textit{Case} $\Theta = \Theta_1 + \Theta_2$: By inversion of \textit{Entailment}, either $H \models \Theta_1$ or $H \models \Theta_2$.\\
    \textit{Subcase} $H \models \Theta_1$: By the induction hypothesis, we have $H \models \mathtt{Restrict}~\Theta_1~h$.
    By \textsc{Ent-ChoiceL}, $H \models \mathtt{Restrict}~\Theta_1~h + \mathtt{Restrict}~\Theta_2~h$.\\
    \textit{Subcase} $H \models \Theta_2$: Symmetric to the previous subcase.
  }{By straight-forward induction on $\Theta$.}
\end{proof}

\begin{lemma}[NegRestriction Entailed]
    \label{lem:negrestrict-domain-entail}
    If $H \models \Theta$ and $h \not \in \mathit{dom(H)}$ then $H \models \mathtt{NegRestrict}~\Theta~h$.
\end{lemma}

\begin{proof}
  \iftoggle{fullproofs}{
    By induction on $\Theta$.\\
    \textit{Case} $\Theta = 0$: The case immediately holds as $H \models 0$ is a contradiction.\\
    \textit{Case} $\Theta = 1$: By inversion of \textit{Entailment}, $H = \boldsymbol{\cdot}$.
    By definition of \textit{Negated Restriction}, $\mathtt{NegRestrict}~\Theta~h = \mathtt{NegRestrict}~1~h = 1$.
    By \textsc{Ent-Empty} $\boldsymbol{\cdot} \models 1$, i.e., $H \models \mathtt{NegRestrict}~\Theta~h$. \\
    \textit{Case} $\Theta = g$: By inversion of \textit{Entailment}, $\mathit{dom(H) = \{g \}}$.
    By the induction hypothesis $h \neq g$.
    By definition of \textit{Restriction} $\mathtt{NegRestrict}~\Theta~h = \mathtt{NegRestrict}~g~h = g$.
    By \textsc{Ent-Inst} $H \models g$, i.e., $H \models \mathtt{NegRestrict}~\Theta~h$. \\
    \textit{Case} $\Theta = \Theta_1 \cdot \Theta_2$: By inversion of \textit{Entailment}, $H=H_1 \cup H_2, H_1 \models \Theta_1, H_2 \models \Theta_2$.
    By $h \not \in \mathit{dom(H)}$, $h \not\in \mathit{dom(H_1)}$ and $h \not\in \mathit{dom(H_2)}$.
    By the induction hypothesis, $H_1 \models \mathtt{NegRestrict}~\Theta_1~h$ and $H_2 \models \mathtt{NegRestrict}~\Theta_2~h$.
    By \textsc{Ent-Seq}, $H_1 \cup H_2 \models \mathtt{NegRestrict}~\Theta_1~h \cdot \mathtt{NegRestrict}~\Theta_2~h$ which finishes the case. \\
    \textit{Case} $\Theta = \Theta_1 + \Theta_2$: By inversion of \textit{Entailment}, either $H \models \Theta_1$ or $H \models \Theta_2$. \\
    \textit{Subcase} $H \models \Theta_1$: By the induction hypothesis, we have $H \models \mathtt{NegRestrict}~\Theta_1~h$.
    By \text{Ent-ChoiceL}, $H \models \mathtt{NegRestrict}~\Theta_1~h + \mathtt{NegRestrict}~\Theta_2~h$. \\
    \textit{Subcase} $H \models \Theta_2$: Symmetric to the previous subcase.
  }{By straightforward induction on $\Theta$.}
\end{proof}

\begin{lemma}[Substitution]
If $\Gamma, x : \tau ; \Theta \vdash c : \Theta'$ and $\cdot; \cdot \vdash v : \tau$ then $\Gamma; \Theta \vdash c[v/x] : \Theta'$
\end{lemma}
\begin{proof}
By straightforward induction on the derivation $\Gamma, x : \tau ; \Theta \vdash c : \Theta'$.
\end{proof}

\begin{lemma}[Entails Subsumption]
    \label{lem:entails-add}
    If $H \models \Theta$ then $H[h \mapsto v] \models \Theta \cdot h$
\end{lemma}
\begin{proof}
  \iftoggle{fullproofs}{
  We analyze two cases.
  \begin{description}
  \item{\textit{Case} $h \in \mathit{dom}(H)$:}
  By the assumption of the case, we have $\mathit{dom}(H) = \mathit{dom}(H[h \mapsto v])$. 
  Let $H_1 = H[h \mapsto v]$ and $H_2 = \{ h \mapsto v \}$. 
  Observe that $H[h \mapsto v] = H_1 \cup H_2$.
  By Lemma~\ref{lem:entails-domain}, we have that $H_1 \models \Theta$.
  By \textsc{Ent-Inst} we have $H_2 \models h$. 
  By \textsc{Ent-Seq} we have $H[h \mapsto v] \models \Theta \cdot h$.

  \item{\textit{Case} $h not\in \mathit{dom}(H)$:}
  Let $H_1 = H$ and $H_2 = \{ h \mapsto v\}$. 
  Observe that $H[h \mapsto v] = H_1 \cup H_2$. 
  By assumption we have $H_1 \models \Theta$. 
  By \textsc{Ent-Inst} we have $H_2 \models h$. 
  By \textsc{Ent-Seq} we have $H[h \mapsto v] \models \Theta \cdot h$.
\end{description}
}{Immediate by definitions and Lemma~\ref{lem:entails-domain} }
\end{proof}

\begin{lemma}[Entails Removal]
    \label{lem:entails-remove}
    If $H \models \Theta$ then $H \setminus h \models \mathtt{Remove}\ \Theta\ h$.
\end{lemma}

\begin{proof}
  \iftoggle{fullproofs}{
    By induction on $\Theta$.\\
    \textit{Case} $\Theta=0$: The case immediately holds, as $H\models 0$ is a contradiction.

    \textit{Case} $\Theta=1$: By inversion of \textit{Entailment}, $H=\boldsymbol{\cdot}$. 
    By set theory, $\boldsymbol{\cdot}\setminus h = \boldsymbol{\cdot}$ and $\mathtt{Remove}\ 1\ h = 1$.
    By \textsc{Ent-Empty}, $\boldsymbol{\cdot}\models 1$.

    \textit{Case} $\Theta=g$:
    By inversion of \textit{Entailment}, $\mathit{dom}(H)=\{g\}$.\\
    \textit{Subcase} $g=h$: 
    By set theory $H\setminus h = \boldsymbol{\cdot}$.
    By definition of \textit{Remove}, $\mathtt{Remove}\ \Theta\ h=1$.
    By \textsc{Ent-Empty}, $\boldsymbol{\cdot} \models 1$, which concludes the case.\\
    \textit{Subcase} $g\neq h$:
    By set theory $H\setminus h = H$.
    By definition of \textit{Remove}, $\mathtt{Remove}\ \Theta\ h=g$.
    By assumption, $H\models\Theta$, which concludes the case.

    \textit{Case} $\Theta=\Theta_1 \cdot \Theta_2$:
    By inversion of \textit{Entailment}, $H=H_1\cup H_2, H_1\models\Theta_1, H_2\models\Theta_2$.
    By induction hypothesis, $H_1\setminus h \models \mathtt{Remove}\ \Theta_1\ h$ and $H_2\setminus h \models \mathtt{Remove}\ \Theta_2\ h$.
    By set theory, $H_1\setminus h \cup H_2\setminus h=(H_1\cup H_2) \setminus h$.
    By definition of \textit{Removal}, $\mathtt{Remove}\ \Theta_1\ h \cdot \mathtt{Remove}\ \Theta_2\ h = \mathtt{Remove}\ (\Theta_1\cdot\Theta_2)\ h$.
    By \textsc{Ent-Seq}, $(H_1\cup H_2) \setminus h \models \mathtt{Remove}\ (\Theta_1\cdot\Theta_2)\ h$.

    \textit{Case} $\Theta=\Theta_1 + \Theta_2$:
    By inversion of \textit{Entailment}, either $H\models\Theta_1$ or $H\models\Theta_2$.\\
    By definition of \textit{Removal}, $\mathtt{Remove}\ \Theta_1\ h + \mathtt{Remove}\ \Theta_2\ h = \mathtt{Remove}\ (\Theta_1 + \Theta_2)\ h$.
    \textit{Subcase} $H\models\Theta_1$:
    By induction hypothesis, $H \setminus h \models \Theta_1 \setminus h$.
    By \textsc{Ent-ChoiceL}, applied to $H \setminus h \models \mathtt{Remove}\ \Theta_1\ h$, and $\mathtt{Remove}\ \Theta_2\ h$, we can conclude $H \setminus h \models \mathtt{Remove}\ \Theta_1\ h + \mathtt{Remove}\ \Theta_2\ h$.
    By definition of \textit{Removal}, $H \setminus h \models \mathtt{Remove}\ (\Theta_1 + \Theta_2)\ h$.
  }{
    By straight-forward induction on $\Theta$.
  }
\end{proof}

\begin{lemma}[Entailment Congruence]
    \label{lem:entails-domain}
    If $H\models\Theta$ and $\mathit{dom}(H)=\mathit{dom}(H')$ then $H'\models\Theta$.
\end{lemma}
\begin{proof}
  \iftoggle{fullproofs}{
By induction on $\Theta$.\\
\textit{Case} $\Theta=0$: The case immediately holds as $H\models 0$ is a contradiction.\\
\textit{Case} $\Theta=1$: By inversion of \textit{Entailment} $H=\boldsymbol{\cdot}$.
By assumption $\mathit{dom}(H)=\mathit{dom}(H')$, $H'=\boldsymbol{\cdot}$ and by \textsc{Ent-Empty}, $H'\models\Theta$.\\
\textit{Case} $\Theta=g$: By inversion of \textit{Entailment}, $\mathit{dom}(H)=\{g\}$.
By assumption $\mathit{dom}(H)=\mathit{dom}(H')$ and by \text{Ent-Inst}, $H'\models\Theta$.\\
\textit{Case} $\Theta=\Theta_1\cdot\Theta_2$:
By inversion of \textit{Entailmment}, $H=H_1\cup H_2, H_1\models\Theta_1, H_2\models\Theta_2$.
By set theory, $\mathit{dom}(H)=\mathit{dom}(H_1)\cup\mathit{dom}(H_2)$.
By induction hypothesis if $\mathit{dom}(H_1')=\mathit{dom}(H_1)$ and $\mathit{dom}(H_2')=\mathit{dom}(H_2)$, then $H_1'\models\Theta_1,H_2'\models\Theta_2$.
By \textsc{Ent-Seq}, $H'=H_1'\cup H_2' \models \Theta_1\cdot\Theta_2$. \\
\textit{Case} $\Theta=\Theta_1+\Theta_2$:
By inversion of \textit{Entailment}, either $H\models\Theta_1$ or $H\models\Theta_2$.\\
\textit{Subcase} $H\models\Theta_1$:
By induction hypothesis, we have $H'\models\Theta_1$.
By \textsc{Ent-ChoiceL}, $H'\models\Theta_1+\Theta_2$.\\
\textit{Subcase} $H\models\Theta_2$: Symmetric to the previous subcase.
} {
  By straight-forward induction on $\Theta$
}
\end{proof}

\noindent We define $\Theta_1 < \Theta_2 \triangleq \llbracket \Theta_1 \rrbracket \subseteq \llbracket \Theta_2 \rrbracket$, i.e., for every $S \in \llbracket \Theta_1 \rrbracket, S \in \llbracket \Theta_2 \rrbracket$.

\begin{lemma}[Order Extend]
\label{lem:type-less-than-extend}
If $\Theta_1'<\Theta_1$ then $\Theta_1'\cdot h < \Theta_1 \cdot h$.
\end{lemma}

\begin{proof}
  \iftoggle{fullproofs}{
We calculate as follows:
\begin{enumerate}
    \item $\llbracket \Theta_1' \rrbracket \subseteq \llbracket \Theta_1 \rrbracket$ by (B) and the definition of $<$.
    \item $\llbracket \Theta_1'\cdot h\rrbracket == \llbracket \Theta_1'\rrbracket \bullet \sos{h}$ by definition of \teval{.}
    \item $\teval{\Theta_1\cdot h} == \teval{\T_1} \bullet \sos{h}$ by definition of \teval{.}
    \item Let $S \in \teval{\T_1'}\bullet\sos{h}$.
    \item $S=S'\cup \{h\}$, where $S'\in\teval{\T_1'}$ by def of $\bullet$.
    \item By 1., $S'\in\teval{\T_1}$
    \item By set theory, $S'\cup\{h\}\in\teval{\T_1}\bullet\sos{h}$.
    \item Then $\teval{\T_1'}\bullet\sos{h}\subseteq \teval{\T_1}\bullet\sos{h}$
    \end{enumerate}
  } { By definitions of $\teval{}$ and $<$.}
\end{proof}

\begin{lemma}[Order Remove]
\label{lem:type-less-than-remove}
If $\T_1'<\T_1$ then $\teval{\mathtt{Remove}~\T_1'~h}\subseteq \teval{\mathtt{Remove}~\T_1~h}$.
\end{lemma}

\begin{proof}
  \iftoggle{fullproofs}{
    Since $\teval{\mathtt{Remove}~\T_1'~h}==\teval{\T_1'}\setminus h$ and $\teval{\mathtt{Remove}~\T_1~h}==\teval{\T_1}\setminus h$ by Lemma \ref{lem:remove-equal}, we can equivalently show that $\teval{\T_1'}\setminus h \subseteq \teval{\T_1}\setminus h$, which follows from set theory.
    } { Immediate by definitions and Lemma~\ref{lem:remove-equal}. }
\end{proof}

\begin{lemma}[Order Restrict]
\label{lem:type-less-than-restrict}
If $\T_1' < \T_1$ then $\teval{\mathtt{Restrict}~\T_1'~h}\subseteq\teval{\mathtt{Restrict}~\T_1~h}$
\end{lemma}
\begin{proof}
  \iftoggle{fullproofs}{
By Lemma \ref{lem:restrict-equal}, $\teval{\mathtt{Restrict}~\T_1'~h}==\teval{\T_1'}|h$ and $\teval{\mathtt{Restrict}~\T_1~h}==\teval{\T_1}|h$.
By set theory, $\teval{\T_1'}|h\subseteq \teval{\T_1}|h$ when $\teval{\T_1'}\subseteq\teval{\T_1}$, so we are done.
} { Immediate by definitions and Lemma~\ref{lem:restrict-equal}. }
\end{proof}

\begin{lemma}[Order NegRestrict]
\label{lem:type-less-than-negrestrict}
If $\T_1' < \T_1$ then $\teval{\mathtt{NegRestrict}~\T_1'~h}\subseteq\teval{\mathtt{NegRestrict}~\T_1~h}$
\end{lemma}
\begin{proof}
  \iftoggle{fullproofs}{
    By Lemma \ref{lem:negrestrict-equal}, $\teval{\mathtt{NegRestrict}~\T_1'~h}==\teval{\T_1'}|\neg h$ and $\teval{\mathtt{NegRestrict}~\T_1~h}==\teval{\T_1}|\neg h$.
    By set theory, $\teval{\T_1'}|\neg h\subseteq \teval{\T_1}|\neg h$ when $\teval{\T_1'}\subseteq\teval{\T_1}$, so we are done.
} {
  Immediate by definitions and Lemma~\ref{lem:negrestrict-equal}.
}
\end{proof}

\begin{lemma}[Order Include]
\label{lem:type-less-than-includes}
If $\T'<\T$ and $\mathtt{Includes}~\T_1~h$ then $\mathtt{Includes}~\T_1'~h$.
\end{lemma}
\begin{proof}
  \iftoggle{fullproofs}{
By Lemma~\ref{lem:includes-equal}, $\mathtt{Includes}~\T_1'~h = h
\sqsubset \T_1'$. By the same lemma,
$\mathtt{Includes}~\T_1~h = h \sqsubset \T_1$. Let $S \in \llbracket \T_1' \rrbracket$ to show $h \in S$
and hence $h \sqsubset  \T_1'$.  Since $\llbracket
\T_1' \rrbracket \subseteq \llbracket \T_1 \rrbracket$, by assumption
and definition of $<$, then $S \in \llbracket \T_1 \rrbracket$. Since
$h \sqsubset \T_1$, conclude $h \in S$ and we are done.
}{
  By Lemma~\ref{lem:includes-equal} and definitions.
}

\end{proof}

\begin{lemma}[Order Choice]
\label{lem:type-less-than-plus}
If $\T_a'<\T_a$ and $\T_b'<\T_b$ then $\T_a'+\T_b'<\T_a+\T_b$.
\end{lemma}
\begin{proof}
  \iftoggle{fullproofs}{
We have to show that $\teval{\T_a'+\T_b'}\subseteq\teval{\T_a+\T_b}$ when $\T_a'<\T_a$ and $\T_b'<\T_b$.
By definition of $\teval{.}$ we can equally show that $\teval{\T_a'}\cup\teval{\T_b'}\subseteq \teval{\T_a}\cup \teval{\T_b}$, which follows from set theory.
} {
  Immediate.
}
\end{proof}

\begin{lemma}[Action Type Bounds]
    \label{lem:type-less-than-action}
    If $\Gamma;\T_1\vdash a: \bar{\tau}\rightarrow \T_2$ and $\T_1'<\T_1$, then $\exists \T_2'.\Gamma;\T_1' \vdash a: \bar{\tau}\rightarrow\T_2'$ and $\T_2' < \T_2$.
\end{lemma}
\begin{proof}
  There is only one way to have concluded that $\Gamma;\T_1\vdash a :
  \bar{\tau}\rightarrow \T_2$: via the \textsc{[T-Action]} rule, which
  gives us two facts: we know $a = \lambda \bar x : \bar \tau. c$, and
  $\cmdtype{\Gamma, \bar x : \bar \tau}{\T_1}{c}{\T_2}$.

  Since this $c$ is an action command, is only generated by the add,
  remove, modification and sequence commands. So we perform a limited
  induction on the structure of $c$:

  \textit{Case} $c = \syntax{add(h)}$. The only typing rule that
  applies is \textsc{T-Add}, so we know $\T_2 = \T_1 \cdot h$. Now let
  $\T_2' = \T_1' \cdot h$. Then \textsc{T-Add} shows $\cmdtype{\Gamma,
    \bar x : \bar \tau} {\T_1'}{\syntax{add(h)}}{\T_1' \cdot h}$. Then
  $\T_1' \cdot h < \T_1 \cdot h$ follows by
  Lemma~\ref{lem:type-less-than-extend}, and we are done.

  \textit{Case} $c = \syntax{remove(h)}$. The only typing rule that
  could have applied is \textsc{T-Remove}, so we know that $\T_2 =
  \mathtt{Remove}~\T_1~h$. Let $\T_2' = \mathtt{Remove}~\T_1'~h$. Then
  \textsc{T-Remove} shows $\cmdtype{\Gamma, \bar x : \bar
    \tau}{\T_1'}{\syntax{remove(h)}}{\mathit{Remove}~\T_1~h}$. Then
  $\mathit{Remove}~\T_1'~h < \mathit{Remove}~\T_1~h$ by
  Lemma~\ref{lem:type-less-than-remove}.

  \textit{Case} $c = \syntax{h.f = v}$. The only typing rule that
  could have applied is \textsc{T-Mod}, so we know that $\T_2 = \T_1$,
  Let $\T_2' = \T_1'$, which proves $\T_2' < \T_2$ by assumption.

  We know by our case assumtion that $\Gamma, \bar x: \bar \tau ; \T_1
  \vdash e : \mathcal{F}(h,f)$ and $\mathtt{Includes}~\T_1~h$.  By
  \textsc{T-Mod}, we only need to show that (1) $\Gamma, \bar x : \bar
  \tau \vdash e : \mathcal{F}(h,f)$ and (2)
  $\mathtt{Includes}~\T_1'~h$. (1) follows by
  Lemma~\ref{lem:type-less-than-expr}, and (2) follows by
  Lemma~\ref{lem:type-less-than-includes}.

  \textit{Case} $c = c_1;c_2$. The only rule that could have applied
  is \textsc{T-Seq}, so we know that
  $\cmdtype{\Gamma, \bar x : \bar \tau}{\T_1}{c_1}{\T_{11}}$, and
  $\cmdtype{\Gamma, \bar x : \bar \tau}{\T_{11}}{c_2}{\T_2}$, and
  $\T_1' \leq \T_1$.

  The inductive hypothesis on $c_1$ gives us a $\T_{11}' < \T_{11}$
  such that $\cmdtype{\Gamma, \bar x : \bar
    \tau}{\T_1'}{c_1}{\T_{11}'}$.

  The inductive hypothesis on $c_2$ gives us a $\T_2' < \T_2$ such
  that $\cmdtype{\Gamma, \bar x : \bar \tau}{\T_{11}'}{c_2}{\T_2'}$.

  The result follows by \texttt{T-Seq}.
  
\end{proof}

\begin{lemma}[Expression Type Bounds]
\label{lem:type-less-than-expr}
If $\Gamma;\Theta\vdash e:\tau$ and $\T'<\T$, then $\Gamma;\Theta'\vdash e:\tau$.
\end{lemma}
\begin{proof}
    By induction on the typing derivation.

    \textit{Case} \textsc{T-Constant}
    We know $e=k(\bar{e}), \Gamma;\Theta\vdash e_i:\tau_i$ for all $i$, $\mathtt{typeof}(k)=\bar{\tau}\rightarrow \tau$ and $\T'<\T$.
    By induction hypothesis, $\Gamma;\Theta' \vdash e_i:\tau_i$ for all i and we are done by \textsc{T-Constant}.

    \textit{Case} \textsc{T-Var}
    We know $e=x, x:\tau \in \Gamma$, and $\T'<\T$.
    We are done by \textsc{T-Var}.

    \textit{Case} \textsc{T-Field}
    We know $e=h.f, \mathtt{Includes}~\T~h$ and $\T'<\T$.
    By Lemma \ref{lem:type-less-than-includes}, we know $\mathtt{Includes~\T'~h}$ and the result follows by \textsc{T-Field}.
\end{proof}

\begin{lemma}[Control Type Bounds]
\label{lem:type-less-than}
If $\cmdtype{\Gamma}{\Theta_1}{c}{\Theta_2}$ and $\Theta_1' < \Theta_1$, then $\exists \Theta_2'. \cmdtype{\Gamma}{\Theta_1'}{c}{\Theta_2'}$ and $\Theta_2' < \Theta_2$.
\end{lemma}

\begin{proof}
By induction on a derivation of \cmdtype{\Gamma}{\Theta_1}{c}{\Theta_2}.
We refer to assumptions \cmdtype{\Gamma}{\Theta_1}{c}{\Theta_2} and $\Theta_1'<\Theta_1$ as (A) and (B) respectively.
Similarly, we use (1) and (2) to refer to the proof goals $\exists \Theta_2'. \cmdtype{\Gamma}{\Theta_1'}{c}{\Theta_2'}$ and $\Theta_2' < \Theta_2$ respectively.

\textit{Case} \textsc{T-Zero}: 
By assumption, we have
$\mathtt{Empty}~\Theta_1$. By Lemmas~\ref{lem:empty-iff} and
\ref{lem:type-less-than-includes} we have $\mathtt{Empty}~\Theta_1'$.
Let $\Theta_2' = \Theta_2$. We have
$\cmdtype{\Gamma}{\Theta_1'}{c}{\Theta_2}$ by \textsc{T-Zero}, proving
(1), and $\Theta_2' < \Theta_2$ by reflexivity, proving (2).

\textit{Case} \textsc{T-Skip}:
We know $c=\syntax{skip}$ and $\Theta_2=\Theta_1$ and $\Theta_1'<\Theta_1$.
Let $\Theta_2'=\Theta_1'$.
Then by assumption (B) $\Theta_2'=\Theta_1'<\Theta_1=\Theta_2$, proving (2) and \cmdtype{\Gamma}{\Theta_1'}{\syntax{skip}}{\Theta_1'} by \textsc{T-Skip}, proving (1).

\textit{Case} \textsc{T-Emit}:
We know $c=\syntax{emit(h)}$ and $\Theta_2 = \Theta_1$ and $\Theta_1' < \Theta_1$.
Let $\Theta_2'=\Theta_1'$.
Then by assumption (B), $\Theta_2'=\Theta_1'<\Theta_1=\Theta_2$, proving (2) and \cmdtype{\Gamma}{\Theta_1'}{\syntax{emit(h)}}{\Theta_1'} by \textsc{T-Emit}, proving (1).

\textit{Case} \textsc{T-Add}:
We know $c=\syntax{add(h)}$ and $\Theta_2=\Theta_1 \cdot h$ and $\Theta_1'<\Theta_1$.
(1) follows since we can prove \cmdtype{\Gamma}{\Theta_1'}{\syntax{add(h)}}{\Theta_1'\cdot h} by \textsc{T-Add}.
(2), i.e., $\T_1'\cdot h < \T_1 \cdot h$, follows from Lemma \ref{lem:type-less-than-extend}.

\textit{Case} \textsc{T-Extr}:
Similar to case \textsc{T-Add}.
We know $c=\syntax{extract(h)}$ and $\T_2=\T_1\cdot h$ and $\T_1'<\T_1$.
Let $\T_2'=\T_1'\cdot h$.
(1) follows since we can prove \cmdtype{\Gamma}{\Theta_1'}{\syntax{extract(h)}}{\Theta_1'\cdot h} by \textsc{T-Extract}.
(2) follows by Lemma \ref{lem:type-less-than-extend}.

\textit{Case} \textsc{T-Rem}:
We know $c=\syntax{remove(h)}$ and $\T_2=\mathtt{Remove}~\T_1~h$ and $\T_1'<\T_1$.
Let $\T_2'=\mathtt{Remove}~\T_1'~h$.
(1) follows by \textsc{T-Rem} and for (2) we have to show that $\mathtt{Remove}~\T_1'~h<\mathtt{Remove}~\T_1~h$, which follows from Lemma \ref{lem:type-less-than-remove}.

\textit{Case} \textsc{T-Mod}:
We know $c=\syntax{h.f=e}$ and $\T_2=\T_1$ and $\T_1'<\T_1$.
Let $\T_2' = \T_1'$.
If $\teval{\T_1'}=\teval{0}$ then $\T_1 == 0$ by idempotent semiring equality and (1) follows by \textsc{T-Zero}.
Otherwise $\teval{\T_1'}$ is nonempty.
To show (1) we need to show 
\begin{enumerate}[(a)]
    \item $\mathtt{Includes}~\T_1'~h$,
    \item $\mathcal{F}(h,f)=\tau$,
    \item $\Gamma;\Theta_1\vdash e: \tau$
\end{enumerate}
(b) and (c) follow from the assumption that the previous rule in the typing derivation was \textsc{T-Mod}.
This inversion also gives us $\mathtt{Includes}~\T_1~h$.
To show (a) we calculate as follows.
$h\sqsubset \teval{\T_1}$ by Lemma \ref{lem:includes-equal}, i.e. $h\in S$ for every $S \in \teval{\T_1}$ by definition.
Since $\teval{\T_1'}\subseteq\teval{\T_1}$, $h\in S$ for every $S \in \teval{\T_1'}$, by set theory.
By definition we get $h \sqsubset \teval{\T_1'}$.
By Lemma \ref{lem:includes-equal}, we can conclude $\mathtt{Includes}~\T_1'~h$.

\textit{Case} \textsc{T-Seq}:
We know $c=\syntax{c_1;c_2}$ and \cmdtype{\Gamma}{\T_1}{c_1}{\T_{11}} and \cmdtype{\Gamma}{\T_{11}}{c_2}{\T_2} and $\T_1'<\T_1$.
By induction hypothesis, $\exists \T_{11}'.\cmdtype{\Gamma}{\T_1'}{c_1}{\T_{11}'}$ and $\T_{11}'<\T_{11}$.
Again, by induction hypothesis, $\exists \T_2'.\cmdtype{\Gamma}{\T_{11}'}{c_2}{\T_2'}$ (proving 1) and $\T_2'<\T_2$ (proving 2) which concludes the case.

\textit{Case} \textsc{T-IfValid}:
We know $c=\syntax{valid(h)~c_1~else~c_2}$ and $\cmdtype{\Gamma}{\mathtt{Restrict}~\T_1~h}{c_1}{\T_t}, \cmdtype{\Gamma}{\mathtt{NegRestrict}~\T_1~h}{c_2}{\T_f}, \T_2=\T_t+\T_f$, and $\T_1'<\T_1$.
Let $\T_2'=\mathtt{Restrict}~\T_1'~h+\mathtt{NegRestrict}~\T_1'~h$.
(1) is immediate from \textsc{T-IfValid}.
(2) follows from Lemmas \ref{lem:type-less-than-restrict}, \ref{lem:type-less-than-negrestrict} and \ref{lem:type-less-than-plus}.

\textit{Case} \textsc{T-If}:
We know $c=\syntax{if~(e)~c_1~else~c_2}, \cmdtype{\Gamma}{\T_1}{c_1}{\T_{11}}, \cmdtype{\Gamma}{\T_1}{c_2}{\T_{12}}, \Gamma;\Theta_1\vdash e: \mathit{Bool}$ and $\T_1'<\T_1$.
By induction hypothesis, there exists $\T_{11}'$ such that (1a) $\cmdtype{\Gamma}{\T_1'}{c_1}{\T_{11}'}$ and (2a) $\T_{11}'<\T_{11}$.
Also by induction hypothesis, there exists $\T_{12}'$ such that (1b) $\cmdtype{\Gamma}{\T_1'}{c_2}{\T_{12}'}$ and (2b) $\T_{12}'<\T_{12}$.
Let $\T_2'=\T_{11}'+\T_{12}'$.
(1) follows by \textsc{T-If} (1a), (1b), and the fact that $\Gamma;\Theta_1\vdash e: \mathit{Bool}$.
(2) follows by Lemma \ref{lem:type-less-than-plus}, (2a), (2b).

\textit{Case} \textsc{T-Apply}:
We know $c=\syntax{t.apply()}, \T_2=\T_{11}+\T_{12}+...+\T_{1n}, \syntax{t.actions}=a_1+a_2+...+a_n, \cdot;\Theta_1\vdash e_j:\tau_j$ for $j=1,...,m$, $\mathcal{CV}(t)=(S_1...S_n)$, $(e_1...e_m)=\{e_i|(e_i,m_i)\in\mathit{t.reads()} \land \neg\mathsf{maskable}(t,e_i,m_i))\}$ and $\mathtt{Restrict} \T_1 S_i \vdash a_i: \bar{\tau}_i \rightarrow \T_{1i}$.
We want to construc $\T_2'<\T_2$ such that $\cmdtype{\Gamma}{\T_1'}{\syntax{t.apply()}}{\T_2'}$.
By repeated application of Lemma \ref{lem:type-less-than-restrict}, $\mathtt{Restrict}~\T_1'~S_i < \mathtt{Restrict}~\T_1~S_i$.
For every $i$ apply Lemma \ref{lem:type-less-than-action} which gives us $\Gamma;\mathtt{Restrict}~\Theta_1'~S_i\vdash a:\bar{\tau}\rightarrow \T_{1_i}$ and $\T_{1i}'<\T_{1i}$.
Let $\T_2'=\sum_{i}\T_{1i}'$.
(2) follows by \textsc{T-Apply}.
To show (1), i.e., $\T_2'=\sum_{i}\T_{1i}'<\sum_{i}\T_{1i}=\T_2$.
We know $\T_{1i}'<\T_{1i}$ for all $i$.
The result follows by repeated application of Lemma \ref{lem:type-less-than-plus}.
\end{proof}

\begin{theorem}[Preservation]
If $\cmdtype{\Gamma}{\Theta_1}{c}{\Theta_2}$ and 
$\langle I,O,H,c \rangle \rightarrow \langle I',O',H',c'\rangle$, where 
$H \models \Theta_1$, then 
$\exists \Theta_1', \Theta_2'.~\cmdtype{\Gamma}{\Theta_1'}{c'}{\Theta_2'}$ where $H' \models \Theta_1'$ and $\Theta_2' < \Theta_2$.
\end{theorem}

\begin{proof}
By induction on a derivation of $\cmdtype{\Gamma}{\Theta_1}{c}{\Theta_2}$, with a case analysis on the last rule used.
\begin{description}
\item{\textit{Case} \textsc{T-Skip}:}
$c = \syntax{skip}$ and $\Theta_2 = \Theta_1$\\[.25em]
  Vacuously holds as there is no $c'$ such that $\langle I,O,H,c
  \rangle \rightarrow \langle I',O',H',c'\rangle$.

\item{\textit{Case} \textsc{T-Extr}:} $c = \syntax{extract(h)}$ and
  $\Theta_2 = \Theta_1 \cdot h$\\[.25em] 
  The only evaluation rule that
  applies to $c$ is \textsc{E-Extr}, so we also have $c' =
  \syntax{skip}$ and $\mathcal{HT}(h)=\eta$ and $H' = H[h\mapsto v]$
  where $\mathit{deserialize}_\eta(I) = (v,I')$. Let $\Theta_1' = \Theta_2'
  = \Theta_2$. We have $\cmdtype{\Gamma}{\Theta_1'}{c'}{\Theta_2'}$
  by \textsc{T-Skip}, we have $H' \models \Theta_2'$ by
  Lemma~\ref{lem:entails-add}, and we have $\Theta_2' < \Theta_2$ by reflexivity.

\item{\textit{Case} \textsc{T-Emit}:} $c = \syntax{emit(h)}$ and
  $\Theta_2 = \Theta_1$.\\[.25em] 
There are two evaluation rules that apply to $c$, \textsc{E-Emit} and
\textsc{E-EmitInvalid}. In either case, $c' = \syntax{skip}$ and $H' = H$. 
Let $\Theta_1' = \Theta_2' = \Theta_1$. 
We have $\cmdtype{\Gamma}{\Theta_1'}{c'}{\Theta_2'}$ by
\textsc{T-Skip}, we have $H' \models \Theta_1'$ by assumption, and we
have $\Theta_2' < \Theta_2$ by reflexivity.

\item{\textit{Case} \textsc{T-Seq}:}
$c = \syntax{c_1; c_2}$ and $\cmdtype{\Gamma}{\Theta_1}{c_1}{\Theta_{12}}$ and $\cmdtype{\Gamma}{\Theta_{12}}{c_2}{\Theta_2}$\\[.25em]
There are two evaluation rules that apply to $c$, \textsc{E-Seq1} and
\textsc{E-Seq}. 
\begin{description}
\item{\textbf{Subcase} \textsc{E-Seq}:}
 $c' = \syntax{c_2}$ and $H' = H$\\[.25em] 
By inversion of $\cmdtype{\Gamma}{\Theta_1}{c_1}{\Theta_{12}}$ we have
$\Theta_{12} = \Theta_1$. Let $\Theta_1' = \Theta_1$ and $\Theta_2' =
\Theta_2$. We have $\cmdtype{\Gamma}{\Theta_1'}{c'}{\Theta_2'}$ by
assumption, we have $H \models \Theta_1'$ also by assumption, and
$\Theta_2' < \Theta_2$ by reflexivity.
\item{\textbf{Subcase} \textsc{E-Seq1}:}
  $c' = \syntax{c_1';c_2}$ and $\langle I,O,H,c_1 \rangle \rightarrow \langle I',O',H',c_1'\rangle$.\\[.25em]
By IH we have $\cmdtype{\Gamma}{\Theta_1'}{c_1}{\Theta_{12}'}$ such that $H' \models \Theta_1'$ and $\Theta_{12}' < \Theta_{12}$.
By Lemma~\ref{lem:type-less-than} we have $\cmdtype{\Gamma}{\Theta_{12}'}{c_2}{\Theta_2}$ for some $\Theta_2' < \Theta_2$.
We have $\cmdtype{\Gamma}{\Theta_1'}{c_1; c_2}{\Theta_2'}$ by \textsc{T-Seq}, which finishes the case.
\end{description}

\item{\textit{Case} \textsc{T-If}:}
$c = \syntax{if~(e)~c_1~else~c_2}$ and $\Gamma; \Theta_1 \vdash e : \mathit{Bool}$ and $\cmdtype{\Gamma}{\Theta_1}{c_1}{\Theta_{12}}$ and $\cmdtype{\Gamma}{\Theta_1}{c_2}{\Theta_{22}}$ and $\Theta_2 = \Theta_{12} + \Theta_{22}$. \\[.25em]
There are three evaluation rules that apply to $c$, \textsc{E-If}, \textsc{E-IfTrue}, and \textsc{E-IfFalse}.
\begin{description}
\item{\textbf{Subcase} \textsc{E-If}:} $c' =
  \syntax{if~(e')~c_1~else~c_2}$ and $H' = H$\\[.25em] Let $\Theta_1' =
  \Theta_1$ and $\Theta_2' = \Theta_2$. We have
  \cmdtype{\Gamma}{\Theta_1'}{\syntax{if~(e)~c_1~else~c_2}}{\Theta_2'} by
  \textsc{T-If}, we have $H \models \Theta_1$ by assumption, and we
  have $\Theta_2 < \Theta_2'$ by reflexivity.
\item{\textbf{Subcase} \textsc{E-IfTrue}:} $c' = c_1$ and $H' = H$.\\[.25em]
Let $\Theta_1' = \Theta_1$ and $\Theta_2' = \Theta_{12}$. We have $\cmdtype{\Gamma}{\Theta_1'}{c'}{\Theta_2'}$ by assumption, we have $H \models \Theta_1'$ also by assumption, and we have $\Theta_2' < \Theta_2$ by the definition of $<$ and the semantics of types.
\item{\textbf{Subcase} \textsc{E-IfFalse}:} $c' = c_2$ and $H' = H$.\\[.25em]
Symmetric to the previous case.
\end{description}

\item{\textit{Case} \textsc{T-IfValid}:}
$c = \syntax{valid(h)~c_1~else~c_2}$ and $\cmdtype{\Gamma}{\mathtt{Restrict}~\Theta_1~h}{c_1}{\Theta_{12}}$ and $\cmdtype{\Gamma}{\mathtt{NegRestrict}~\Theta_1~h}{c_2}{\Theta_{22}}$ and $\Theta_2 = \Theta_{12} + \Theta_{22}$. \\[.25em]
There are two evaluation rules that apply to $c$, \textsc{E-IfValidTrue} and \textsc{E-IfValidFalse}
\begin{description}
\item{\textbf{Subcase} \textsc{E-IfValidTrue}:} $c' = c_1$ and $h \in
  \mathit{dom}(H)$ and $H' = H$.\\[.25em] Let $\Theta_1' =
  \mathtt{Restrict}~\Theta_1~h$ and $\Theta_2' = \Theta_{12}$. We have
  $\cmdtype{\Gamma}{\Theta_1'}{c'}{\Theta_2'}$ by assumption, we have
  $H \models \Theta_1'$ by Lemma~\ref{lem:restrict-domain-entail}, and
  we have $\Theta_2' < \Theta_2$ by the definition of $<$ and
  semantics of types.
\item{\textbf{Subcase} \textsc{E-IfValidFalse}:} $c' = c_2$ and $h \not\in \mathit{dom}(H)$ and $H' = H$.\\[.25em]
 Symmetric to the previous case. 
\end{description}

\item{\textit{Case} \textsc{T-Apply}:}
$c = \syntax{t.apply()}$ and $\mathcal{CV}(t) = \bar{S}$ and $t.\mathit{actions} = \bar{a}$ and $\bar{e} = \{e_j \mid (e_j, m_j) \in \mathit{t.reads()} \wedge \neg \mathsf{maskable}(t,e_j,m_j)\}$
and 
$\cdot; \Theta \vdash e_i : \tau_i$ for $e_i \in \bar{e}$ and 
$\mathtt{Restrict}~\Theta_1~S_i \vdash a_i : \bar{\tau_i} \rightarrow \Theta_i'$ for $a_i \in a$ and $\Theta_2 = \sum\left( \Theta_i' \right)$\\[.25em]
There is only one evaluation rule that applies to $c$, \textsc{E-Apply}. It follows that
$\mathcal{CA}(t, H) = (a_i, \bar{v})$, and $c' = c_i[\bar{v}/\bar{x}]$
where $\mathcal{A}(a_i) = \lambda \bar{x}.~c_i$. Next, inverting
\textsc{T-Action}, we have $\cmdtype{\Gamma, \bar{x} : \bar{\tau}_i;}{\texttt{Restrict}~\Theta~S_i}{c_i}{\Theta_i'}$. 
By Proposition~\ref{prop:cp-action-data}, we have $\cdot; \cdot \vdash \bar{v} : \bar{\tau}_i$. 
Hence, by the substitution lemma, we have $\cmdtype{\Gamma}{\texttt{Restrict}~\Theta~S_i}{c_i[\bar{v}/\bar{x}]}{\Theta_i'}$.
Let $\Theta_1' = \mathtt{Restrict}~\Theta~S_i$ and $\Theta_2' = \Theta_i'$. 
We have already shown that $\cmdtype{\Gamma}{\Theta_1'}{c'}{\Theta_2'}$, 
we have that $H' \models \Theta_1'$ by Proposition~\ref{prop:cp-valid}, and 
we have $\Theta_2' < \Theta_2$ by the definition of $<$ and the semantics of union types.

\item{\textit{Case} \textsc{T-Add}:}
$c = \syntax{add(h)}$ and $\Theta_2 = \Theta_1 \cdot h$\\[.25em]
There are two evaluation rules that apply to $c$, \textsc{E-Add} and \textsc{E-AddValid}.
\begin{description}
\item{\textbf{Subcase} \textsc{E-Add}:} $c'=\syntax{skip}$ and $\mathcal{HT}(h)=\eta$ and $\mathit{init}_\eta=v$ and $H'=H[h\mapsto v]$\\[.25em]
Let $\Theta_1'=\Theta_2'=\Theta_2$.
We have $\cmdtype{\Gamma}{\Theta_1'}{c'}{\Theta_2'}$ by \textsc{T-Skip}, we have 
$H' \models \Theta_1'$ by Lemma~\ref{lem:entails-add},
and we have $\Theta_2' < \Theta_2$ by reflexivity.
\item{\textbf{Subcase} \textsc{E-AddValid}:} $c'=\syntax{skip}$ and $H' = H$\\[.25em]
Let $\Theta_1' = \Theta_2' = \Theta_2$.
We have $\cmdtype{\Gamma}{\Theta_1'}{c'}{\Theta_2'}$ by \textsc{T-Skip}, 
we have $H' \models \Theta_1'$ by Lemma~\ref{lem:entails-domain} and \ref{lem:entails-add} since $\mathit{dom}(H') = \mathit{dom}(H[h \mapsto v])$ for any $v$, and
we have $\Theta_2' < \Theta_2$ by reflexivity.

\end{description}

\item{\textit{Case} \textsc{T-Rem}:}
$c = \syntax{remove(h)}$ and $\Theta_2 = \mathtt{Remove}~\Theta_1~h$\\[.25em]
There is only one evaluation rule that applies to $c$, \textsc{E-Rem}, so we have $c' = \syntax{skip}$ and $H' = H \setminus h$. 
Let $\Theta_1' = \Theta_2' = \mathtt{Remove}~\Theta~h$.
We have $\cmdtype{\Gamma}{\Theta_1'}{c'}{\Theta'_2}$ by \textsc{T-Skip}, 
we have $H' \models \Theta_1'$ by Lemma~\ref{lem:entails-remove}, and
we have $\Theta_2' < \Theta_2$ by reflexivity.

\item{\textit{Case} \textsc{T-Mod}:}
$c = \syntax{h.f = e}$ and $\mathtt{Includes}~\Theta_1~h$ and $\mathcal{HT}(h,f) = \tau_i$ and $\cdot; \Theta_1 \vdash e : \tau_i$ and $\Theta_2 = \Theta_1$\\[.25em]
There are two evaluation rules that applies to $c$, \textsc{E-Mod1} and \textsc{E-Mod}.
\begin{description}
\item{\textbf{Subcase} \textsc{E-Mod1}:} $c'=\syntax{h.f = e'}$ and $e \rightarrow e'$ and $H' = H$\\[.25em]
By preservation for expressions we have $\cdot; \Theta_1 \vdash e' : \tau_i$.
Let $\Theta_1' = \Theta_2' = \Theta_1$. W
We have $\cmdtype{\Gamma}{\Theta_1'}{c'}{\Theta_2'}$ by \textsc{T-Mod}, 
we have $H' \models \Theta_1'$ by assumption, and
we have $\Theta_2' < \Theta_2$ by reflexivity.
\item{\textbf{Subcase} \textsc{E-Mod}:} $c'=\syntax{skip}$ and $\mathit{dom}(H') = \mathit{dom}(H)$\\[.25em]
Let $\Theta_1' = \Theta_2' = \Theta_1$. 
We have $\cmdtype{\Gamma}{\Theta_1'}{c'}{\Theta_2'}$ by \textsc{T-Skip}, 
we have $H' \models \Theta_1'$ by Lemma~\ref{lem:entails-domain}, and 
we have $\Theta_2' < \Theta_2$ by reflexivity.
\end{description}
\item{\textit{Case} \textsc{T-Zero}:}
$\mathtt{Empty}~\Theta_1$\\[.25em]
By Lemma~\ref{lem:entailment}, we have $\mathit{dom}(H) \in \llbracket \Theta_1 \rrbracket$. 
By Lemma~\ref{lem:empty-iff}, we have $\llbracket \Theta_1 \rrbracket = \{ \}$, which is a contradiction.\qedhere
\end{description}
\end{proof}

 }{

}

\end{document}